\documentclass[11pt, letterpaper]{amsart}

\usepackage{amsxtra}
\usepackage{graphicx}
\usepackage{newlfont}
\usepackage{amscd}
\usepackage{hyperref}
\usepackage[german,english]{babel}
\usepackage{cancel}
\usepackage{amsmath,amsthm}
\usepackage{amssymb}
\usepackage{enumerate}

 \newtheorem{thm}{Theorem}[section]
  \newtheorem*{Lenard}{Lenard Scheme}
 
 \newtheorem{cor}[thm]{Corollary}
 \newtheorem{lem}[thm]{Lemma}
 \newtheorem{prop}[thm]{Proposition}
 
  \theoremstyle{definition}
 \newtheorem{defn}[thm]{Definition}
  
  \newtheorem{defn-thm}[thm]{Definition-Theorem}
 \theoremstyle{remark}

 \newtheorem{rem}[thm]{Remark}
 \newtheorem{ex}[thm]{Example}
 
 \newtheorem{notation}[thm]{Notation}

\numberwithin{equation}{section}

\numberwithin{thm}{section}

\numberwithin{table}{section}

\numberwithin{figure}{section}

\ifx\pdfoutput\undefined
  \DeclareGraphicsExtensions{.pstex, .eps}
\else
  \ifx\pdfoutput\relax
    \DeclareGraphicsExtensions{.pstex, .eps}
  \else
    \ifnum\pdfoutput>0
      \DeclareGraphicsExtensions{.pdf}
    \else
      \DeclareGraphicsExtensions{.pstex, .eps}
    \fi
  \fi
\fi

\newcommand{\ZZ}{\mathbb{Z}}

\newcommand{\RR}{\mathbb{R}}

\newcommand{\CC}{\mathbb{C}}

\newcommand{\End}{\mbox{End\ }}
\newcommand{\g}{\mathfrak{g}}

\newcommand{\m}{\mathfrak{m}}
\newcommand{\n}{\mathfrak{n}}

\newcommand{\vac}{\left| 0 \right>}

\newcommand{\lt}{\left<}
\newcommand{\rt}{\right>}
\newcommand{\ad}{\text{ad}}
\newcommand{\gr}{\text{gr}}
\newcommand{\sll}{\text{sl}}
\newcommand{\WW}{\mathcal{W}}

\begin{document}

\title{ Structure of classical affine and classical affine fractional $\WW$-algebras }
\author{ Uhi Rinn Suh}
\date{}
\maketitle

\begin{abstract}
We show that one can construct a classical affine $\WW$-algebra via a classical BRST complex. This definition clarifies that classical affine $\WW$-algebras can be considered as  quasi-classical limits of quantum affine $\WW$-algebras. 
 
We also give a definition of a classical affine fractional $\WW$-algebra as a Poisson vertex algebra. As in the classical affine case, a classical affine fractional $\WW$-algebra has two compatible $\lambda$-brackets and is isomorphic to an algebra of differential polynomials as a differential algebra. When a classical affine fractional $\WW$-algebra is associated to a minimal nilpotent,  we describe explicit forms of free generators and compute $\lambda$-brackets between them.  Provided some assumptions on a classical affine fractional $\WW$-algebra, we find an infinite sequence of integrable systems related to the algebra, using the generalized Drinfel'd and Sokolov  reduction.

\end{abstract}

\tableofcontents

\setcounter{tocdepth}{-1}

\pagestyle{plain}

\section{Introduction}\label{Sec:Introduction}

An affine classical $\WW$-algebra is closely related to the theory of integrable systems. The connection between two structures  was constructed by Fateev and Lukyanov \cite{FL}. Moreover, Drinfel'd and Sokolov \cite{DS} explained that, given a Lie algebra and a principle element, an affine classical $\WW$-algebra is obtained by the Drin'feld-Sokolov Hamiltonian reduction. 

To be precise, let  $\g$ be a simple Lie algebra, let $(e,h,f)$ be an $\sll_2$-triple with a principle nilpotent $f$ and let $\g(i)$ be the eigenspace $\{ g\in \g | [\frac{h}{2},g]= i g\}.$ 
Consider the Lax operator
\begin{equation}\label{DSred} 
L= \partial_x+q(x, t)+\Lambda,\ \ \ x \in S^1, \ t\in \RR,
\end{equation}
 where $\Lambda=-f-pz^{-1}\in \g[z,z^{-1}]$, $p\in\g$ is a central element of $\n=\bigoplus_{i\geq 0} \g(i)$ and $q(x,t) \in \bigoplus_{i>-1} \g(i)$. The gauge transformation by $S \in \n$ on the space of matrix valued smooth  functions, $ S^1 \to \bigoplus_{i>-1} \g(i)$,  is defined by  
 $q(x,t)\mapsto \widetilde{q}(x,t)$, where $ e^{\ad S}L=\partial_x+\widetilde{q}(x,t)+ \Lambda $. Then the gauge invariant functionals consist of the affine classical $\WW$-algebra associated to $\g$ and $f$. Moreover, the $\WW$-algebra is associated to an integrable system obtained by a commutator of the Lax operator.
 
Burroughs, De Groot, Hollywood and Miramontes generalized this idea replacing $\Lambda$ by $\Lambda_m= z^{-m}\cdot \Lambda \in \g[z,z^{-1}]$ and $q(x)$ by $q_m(x) \in \bigoplus_{j=-m+1}^{0} \g z^j \oplus \bigoplus_{i>-1} \g(i) z^{-m}$ (see \cite{BDHM, DHM}). In this way, a classical affine fractional $\WW$-algebra associated to $\g$ and $\Lambda_m$ was described as a differential algebra using gauge invariant functionals. The authors also explained relations between fractional $\WW$-algebras and integrable systems. 

 On the other hand, Barakat, De Sole and Kac \cite{BDK} developed the theory of Hamiltonian equations in the language of the Poisson vertex algebra (PVA) theory. In particular, the authors described the Drinfel'd-Sokolov Hamiltonian reduction using $\lambda$-brackets and defined a classical affine $\WW$-algebra as a PVA. Indeed, the $\WW$-algebra in Drinfel'd-Sokolov \cite{DS} and the $\WW$-algebra in Barakat-De Sole-Kac \cite{BDK} are equivalent as differential algebras (see \cite{BDK, DKV}).

As one predicts from the name of a ``classical affine" $\WW$-algebra, there are other three types of $\WW$-algebras: classical finite, quantum finite and quantum affine $\WW$-algebras. These three types of $\WW$-algebras have algebraic structures of Poisson algebra, Lie algebra and Vertex algebra, respectively. In \cite{GG}, Gan and Ginzburg introduced two equivalent definitions of a quantum (classical) finite $\WW$-algebra, 
by a Lie algebra cohomology and  by a Hamiltonian reduction. In the quantum affine case,  a $\WW$-algebra is defined by a BRST-complex, which can be considered as a quantization of the Lie algebra cohomology in \cite{GG}  (see \cite{DK}). It is still open that if a quantum affine $\WW$-algebra has an analogous definition to the definition of a finite $\WW$-algebra by a Hamiltonian reduction. \\

In this paper, we give a new construction of a classical affine $\WW$-algebra by a classical BRST-complex and we prove that the two definitions,  via a classical BRST-complex and via a Hamiltonian reduction, are equivalent. Moreover, we define a classical affine fractional $\WW$-algebra as a PVA with two compatible $\lambda$-brackets. \\

Another main result of this paper is finding free generators of classical affine $\WW$-algebras, as differential algebras. In general, finding free generators of $\WW$-algebras is not easy. However, provided that $f$ is a minimal nilpotent in $\g$, Premet found the generating elements of finite quantum $\WW$-algebras associated to  $f$  as associative algebras. Also, they computed the Lie brackets between the generators.  In the affine quantum cases,   Kac and Wakimoto found the generating elements of $\WW$-algebras associated to $f$ as differential algebras and they computed $\lambda$-brackets between the generators. ( See \cite{KW}, \cite{P1} and \cite{P2}.)\\
  
As in the other cases, with the assumption that $f$ is a minimal nilpotent, one can find free generators of  affine classical $\WW$-algebras and  Poisson $\lambda$-brackets between them \cite{S}. Moreover, in this paper, we describe  generating elements of classical affine fractional $\WW$-algebras and $\lambda$-brackets between them.\\
\\

\centerline{$\textbf{Outline of this paper}$}

In Section \ref{Section:PVA,IS}, we review several notions which are used in following sections. In Section \ref{Subsec:PVA},  (nonlinear) Lie conformal  algebras and Poisson vertex algebras are introduced, and in Section \ref{Subsec:IS} integrable systems are explained in the theory of PVAs.
 
In Section \ref{Subsec:3.1} and \ref{Subsec:3.2}, we give a definition of a classical affine $\WW$-algebra by a classical BRST complex  and by a Hamiltonian reduction, respectively. In Section \ref{Subsec:3.3}, we prove the two definitions in Section \ref{Subsec:3.1} and Section \ref{Subsec:3.2} are equivalent, which is the first goal of this paper. 

In Section \ref{Sec:7}, we derive a definition of a classical affine fractional $\WW$-algebra as a PVA. In Section \ref{Subsec:7.1} we review the construction of  a classical affine fractional $\WW$-algebra explained in  \cite{BDHM, DHM} and show this algebra has two local Poisson brackets. In Section \ref{Subsec:7.2}, we give PVA structures  on the algebra using relations between local Poisson brackets and $\lambda$-brackets. The definition of a classical affine fractional $\WW$-algebra in Section \ref{Subsec:7.2} is analogous to the definition of a classical affine $\WW$-algebra via Hamiltonian reduction.

In Section \ref{Sec:8}, we find generating elements of classical affine fractional $\WW$-algebras when the algebras are associated to semisimple elements. In Section \ref{Subsec:8.1} and  \ref{Subsec:8.2}, we explain how to find generators of classical affine fractional $\WW$-algebras and we give an example using the method.  In Section \ref{Subsec:8.3}, we describe explicit formulas of generating elements of a classical affine fractional $\WW$-algebra when it is associated to a minimal nilpotent element. We notice that all the results in this section hold for classical affine $\WW$-algebras as corollaries and these are written in \cite{S}. 
  
In Section \ref{Sec:9}, we explain relations between integrable systems and classical affine fractional $\WW$-algebras using the language of PVAs. Here, the idea comes from  \cite{BDHM, DHM}.  \\
\\

\centerline{$\textbf{Acknowledgement}$}

I would like to thank my Ph.D. thesis advisor, Victor Kac, for valuable discussions.

\section{ Poisson vertex algebras and Integrable systems} \label{Section:PVA,IS}

\subsection{Lie conformal superalgebras and Poisson vertex algebras} \label{Subsec:PVA}

An associative algebra $\mathcal{D}$ endowed with a linear operator $\partial: \mathcal{D} \to \mathcal{D}$ is called a differential algebra if the operator $\partial$ satisfies
$$ \partial(AB)= A\partial(B)+\partial(A)B, \ \text{ for } A,B \in \mathcal{D}.$$
Let $I=\{1, \cdots , l\}$ be an index set. An important example of differential algebras is  the algebra of differential polynomials 
$
\CC_{\text{diff}}[ \ a_i\ |\  i\in  I \ ]:=\CC[\ a_i^{(n)}| \ a_i^{(n)}:= \partial^n a_i,   i\in I,  n \in \ZZ_{\geq 0}]
$
in the variables $a_1, \cdots, a_l$. \\

Now we introduce a Lie conformal superalgebra and a Poisson vertex algebra.  For this purpose, we review $\ZZ/2\ZZ$-graded algebraic structures and a $\lambda$-bracket on a $\CC[\partial]$-module. 

\begin{defn}
\begin{enumerate}[(i)]
\item
A vector superspace $V$ is a vector space with a $\ZZ/2\ZZ$-graded  decomposition $V=V_{\bar{0}} \oplus V_{\bar{1}}$. We call $V_{\bar{0}}$ the even space and $V_{\bar{1}}$ the odd space.
\item
Let $V$ be a vector superspace and $a \in V_{\bar{i}}$ . Then we say parity $p$ of $a$ is $i$ and we write $p(a)=i.$  
\item
Given a vector superspace $V=V_{\bar{0}} \oplus V_{\bar{1}}$, the algebra $\End V$ acquires a $\ZZ/2\ZZ$-grading by letting 
$$(\End V)_{\bar{\alpha}} = \{ A \in \End V | A(V_{\bar{\beta}}) \subset V_{\bar{\alpha}+\bar{\beta}}\}$$
for $\bar{\alpha}, \bar{\beta} \in \ZZ/2\ZZ.$
\item 
A commutative superalgebra $A=A_{\bar{0}}\oplus A_{\bar{1}}$ is a superalgebra such that 
$ab={p(a,b)}ba$, where  $a, b\in A_{\bar{0}} \cup A_{\bar{1}}$ and $p(a,b)=(-1)^{ p(a)p(b)}$. 
\end{enumerate}
\end{defn}

\begin{defn}[$\lambda$-bracket]
Let $R$ be a $\CC[\partial]$-module. A $\lambda$-bracket $[ \cdot_\lambda \cdot]: R \otimes R \to R[\lambda]$ on $R$ is a $\CC$-bilinear map satisfying
\begin{equation}\label{sesqui}
 [\partial a _\lambda b] = -\lambda[a _\lambda b]\ \text{and}\  \ [a _\lambda \partial b] = (\lambda + \partial) [a _\lambda b],
 \end{equation}
which are called sesquilinearities.
\end{defn}
We denote by $a_{(n)}b$ the $n$-th coefficient in $[a _\lambda b]$, i.e. $[a_\lambda b]= \sum_{n\in \ZZ_{\geq 0}}\lambda^n a_{(n)}b,$ $a_{(n)}b \in R$. Since the formal variable $\lambda$ commutes with $\partial$, sesquiliearities (\ref{sesqui}) are equivalent to 
\begin{equation}
 [\partial a _\lambda b] = -\sum_{i\geq 1}\lambda^i a_{(i-1)}b \ \text{and} \ \ [a_\lambda \partial b] = \sum_{i\geq 0} \lambda^i (\partial a_{(i)}b +a_{ (i-1)}b), \ \text{by letting} \ a_{(-1)}b:=0.
\end{equation}
We also note that $\partial [a_\lambda b]=[\partial a_\lambda b]+[a_\lambda \partial b].$

\begin{defn} [Lie conformal superalgebra] \label{Def:0402_2.2}
 A Lie conformal superalgebra $R$ is a $\CC[\partial]$-module endowed with a $\lambda$-bracket $[\cdot_\lambda\cdot]$ satisfying
\begin{enumerate}[(i)]
\item skewsymmetry  : $[ b _\lambda a ]= - p(a,b)[a _{-\partial-\lambda} b]= \sum_{j\geq 0} \frac{(-\partial-\lambda)^j}{j!}a_{(j)} b,$ 
 \item Jacobi identity  :  $[a _\lambda [b _\mu c]] =p(a,b) [b _\mu [a _\lambda c]] + [[a _\lambda b]_{\lambda+\mu} c],$
\end{enumerate}
where $p(a,b)=(-1)^{p(a)p(b)}$.
\end{defn}

\begin{defn}[Poisson vertex algebra] \label{def-thm} 
A Poisson vertex algebra (PVA) is a quintuple  $(\mathcal{V}, \vac, \partial, \{ \cdot _\lambda \cdot \}, \cdot )$ satisfying the following three properties:
\begin{enumerate}[(i)]
\item $ (\mathcal{V}, \partial, \{ \cdot _\lambda \cdot \})$  is a Lie conformal superalgebra,
\item $ (\mathcal{V}, \vac, \partial, \cdot  )$ is a unital differential associative commutative superalgebra, 
\item $\{a _\lambda bc\} = p(a,b)b\{a _\lambda c\} +  \{ a _\lambda b\} c $, \ \ $a,b,c \in \mathcal{V}.$
\end{enumerate}
\end{defn}

\begin{rem}
Property (iii) in  Definition \ref{def-thm} is called the left Leibniz rule. Along with the skewsymmetry of the $\lambda$-bracket, the right Leibniz rule (\ref{rLeib}) follows:
\begin{equation} \label{rLeib}
\{ab_\lambda c\} =  \{b_{\lambda+\partial}c\}_{\to} a
+  \{a_{\lambda+\partial}c\}_{\to}b, \ \ a,b,c\in \mathcal{V}
\end{equation}
Here the small arrows in (\ref{rLeib}) indicate that the operator $\partial$ acts on the right side of the $\lambda$-brackets. For instance,  $\{a_{\lambda+\partial} b\}_{\to}c= \sum_{i\geq 0} a_{(i)}b (\lambda+\partial)^i c.$ As in the skewsymmetry in Definition \ref{Def:0402_2.2}, without arrows, the derivation $\partial$ acts on the left side of the $\lambda$-brackets. 
\end{rem}

\begin{ex} The Virasoro-Magri PVA on
 $\mathcal{V}=\CC_{\text{diff}}[u]$ with central charge $c \in \CC$ is defined by the $\lambda$-bracket
 \begin{equation}\label{Vir}
\{u_\lambda u\}= (\partial+2\lambda) u+\lambda^3 c.
\end{equation}
The $\lambda$-bracket on $\mathcal{V}$  is completely determined by (\ref{Vir}) using the sesquilinearities and Leibniz rules. Also, one can check the skewsymmetry and Jacobi identity by direct computations. 
\end{ex}

As in the previous example, a $\lambda$-bracket structure on a PVA  is determined by $\lambda$-brackets between generating elements. If a PVA is an algebra of differential polynomials and the $\lambda$-brackets between generating elements are given, then one can find a $\lambda$-bracket  between any two elements in the PVA, by the master formula (\ref{master formula}). 

\begin{prop}
Let $A=\CC_{\text{diff}}[\ a_i\ |\ i \in I\ ]$ be a PVA endowed with the $\lambda$-bracket $\{\cdot_\lambda\cdot\}$. Then $A$ satisfies the following equation called ``master formula'' \cite{DK} :
\begin{equation}\label{master formula}
\{f_{ \ \lambda \ } g\} = \sum_{i,j \in I, m,n \in \ZZ_{\geq0}} \frac{\partial g}{\partial a_j^{(n)}}(\partial+\lambda)^n \{a_{i \ \lambda+\partial \ } a_j \}_{\to}(-\partial-\lambda)^m \frac{\partial f}{\partial a_i^{(m)}}.
\end{equation}
\end{prop} 

\vskip 5mm

To construct an affine $\WW$-algebra via a BRST complex, we recall the main ingredient called a nonlinear Lie conformal superalgebra in the remaining part of this subsection. 

Let $\Gamma_+$ be a discrete additive subset of $\RR_+$ containing $0$ and let $\Gamma_+' := \Gamma_+ \backslash \{0\}$. Let $R$ be a $\CC[\partial]$-module endowed with the $\Gamma_+$-grading
$ R = \bigoplus_{ \zeta \in \Gamma_+ ' } R_\zeta,$
where $R_\zeta$ is a $\CC[\partial]$-submodule. We denote by $\zeta(a)$ the $\Gamma_+'$-grading of $a\in R.$ 

 The tensor superalgebra $\mathcal{T}(R)$  of $R$ is also endowed with the $\Gamma_+$-grading
$$\mathcal{T} (R):= \bigoplus_{\zeta \in \Gamma_+} \mathcal{T}(R)[\zeta], \text{  \ where\  } \mathcal{T}(R)[\zeta] = \{ a \in \mathcal{T}(\g)| \zeta(a) =\zeta\},$$
 such that $ \zeta(a \otimes b) = \zeta(a) +\zeta(b)$ and  $ \zeta(1)=0.$  For $\zeta\in \Gamma_+'$, let us denote by $\zeta_-$ the largest element in $\Gamma_+$ strictly smaller than $\zeta$ and let $ \mathcal{T}_\zeta(R)$ be the direct sum  $\bigoplus_{\alpha \leq \zeta_-} \mathcal{T}(R)[\alpha]$ of $\CC[\partial]$-modules.

\begin{defn}
A nonlinear Lie conformal superalgebra $R$ is a $\Gamma_+'$-graded $\CC[\partial]$-module endowed with a nonlinear $\lambda$-bracket 
$$[\cdot_\lambda \cdot] : R \otimes R \to \CC[\lambda]\otimes \mathcal{T}(R)$$
satisfying
\begin{enumerate}[(i)]
\item grading condition : $ [R_{\zeta_1 \ \lambda}  R_{\zeta_2} ]  \subset \mathcal{T}(R)_{\zeta_1+\zeta_2}, $
\item sesquilinearity :  $[\partial a _\lambda b]=-\lambda[a_\lambda b], \ [a_\lambda \partial b] = (\lambda +\partial)[a_\lambda b] ,$
\item skewsymmetry : $[a_\lambda b] = -p(a,b) [b_{-\lambda-\partial} a],$
\item Jacobi identity : $[a_\lambda[b_\mu c]]-p(a,b)[b_\mu [a_\lambda c]]-[[a_\lambda b]_{\lambda+\mu}  c] \in \CC[\lambda, \mu] \otimes \mathcal{M}(R),$
where $\mathcal{M}(R)$  is the left ideal of $\mathcal{T}(R)$ generated by $ a \otimes b \otimes C -p(a,b) b \otimes a \otimes C - \left( \int_{-\partial}^0  [a_\lambda b] d\lambda \right)\otimes  C$,  
\end{enumerate}
for $a,b,c \in R$ and $C \in \mathcal{T}(R).$
\end{defn}

\begin{ex}
Let $p\in \g$ commute with $\n$, let $k,c \in \CC$ and let $R=Cur^{cp}_k(\g)$ be the free $\CC[\partial]$-module $\CC[\partial]\otimes \g$ endowed with the nonlinear $\lambda$-bracket
$$[a_\lambda b] = [a,b]+\lambda k(a,b) + c(p,[a,b]), \ \ \text{ for } a,b\in \g.$$
Take $\Gamma_+=\ZZ_{\geq 0}$ and let $\gamma(a)=1$ for any $a \in R$. Then $R$ satisfies the grading condition. By computations, one can check that  $R$ is a nonlinear Lie conformal algebra. In addition, let $S(R)$ be the symmetric algebra  of the vector space $R$ over $\CC$. Define the $\lambda$-bracket on  $S(R)$ by the $\lambda$-bracket on $R$ and Leibniz rules. Then one can check that $S(R)$ is a Poisson vertex algebra.
\end{ex}

\subsection{Integrable systems} \label{Subsec:IS}
Let  $u_i=u_i(x,t)$, $i \in I=\{1, \cdots, l\}$,  be smooth functions with the coordinate $x\in S^1$ and time $t \in \RR$ and  let 
$u^{(n)}=(u_i^{(n)})_{i\in I}= \left(\frac{\partial^n}{\partial x^n} u_i\right)_{i\in I}$. An infinite dimensional evolution equation is a system of partial differential equations of the form
\begin{equation}\label{EE}
\frac{du_i}{dt}= P_i(u, u', u^{(2)}, \cdots), \ i \in  I,
\end{equation}
where $P_i$ are in the algebra of differential polynomials
\begin{equation} \label{DA}
A= \CC[u_i^{(n)}| i\in I , n=0,1, \cdots]=\CC_{\text{diff}}[\ u_i\ |\ i\in I \ ].
\end{equation}
Here, the derivation $\partial$ on $A$ is defined by
\begin{equation}
\partial =\sum_{n \in \ZZ_{\geq0}} \sum_{i=1}^l  u_i^{(n+1)} \frac{\partial}{\partial u_i^{(n)}}, \ \text{ or } \ \partial u_i^{(n)}= u_i^{(n+1)} . 
\end{equation}
We note that the commutator between $\partial$ and $\frac{\partial}{\partial u_i}$ acts on $f\in A$ as follows:
\begin{equation}
\left[ \partial, \frac{\partial}{\partial u_i^{(n)}}\right] f = \left( \partial \frac{\partial}{\partial u_i^{(n)}} - \frac{\partial}{\partial u_i^{(n)}} \partial \right) f = - \frac{\partial f}{\partial u_i^{(n-1)}}.
\end{equation}
Also, we say the total derivative order of $f\in A$ is $n$ 
if  the following two conditions hold:\\
\begin{equation*}  \text{
(i) there is $i \in I$ such that 
$\frac{\partial f}{\partial u_i^{(n)}}\neq 0$
\ \ (ii)  for any $j \in I$, the derivative $\frac{\partial f}{\partial u_j^{(n+1)}}=0$ }.
\end{equation*}

\vskip 5mm

Now we introduce integrals of motion in algebraic way and Hamiltonian systems using PVAs. 

\begin{defn} [local functional] \label{Def:LocalFunc} 
A local functional $\int f$ is the image of $f \in A$ under the projection map
\begin{equation}
\int: A \to A/\partial A.
\end{equation}
\end{defn}
The space of local functionals $\int  A$ is the universal space for the algebras satisfying the integration by parts. 
\begin{defn}  [integral of motion] \label{Def:IntMotion}
A local functional $\int f$ is called an integral of motion if
\begin{equation} \label{Eqn:0124_1.6}
\int \frac{df}{dt} =0.
\end{equation}
\end{defn}

\begin{rem}
When we need to clarify that $x$ is the coordinate, we use $\partial_x$, $f(x)$, and $\int f(x) dx$ instead of $\partial$, $f$, and $\int f$.
\end{rem}

Equivalently, using the chain rule and the integration by parts, a functional $\int f$ is called an integral of motion if
\begin{equation}\label{Eqn:2.8_010914}
\int \sum_{i\in I, n \in \ZZ_{\geq0}} \frac{\partial f}{\partial u_i^{(n)}}  \frac{du_i^{(n)}}{dt} = \int  \sum_{i\in I} \left(\sum_{n \in \ZZ_{\geq0}} (-\partial)^n\frac{\partial f}{\partial u_i^{(n)}}\right)  \frac{du_i}{dt}=0.
\end{equation}
To make the formula (\ref{Eqn:2.8_010914}) simpler, we recall the following definition.

\begin{defn}[fractional derivative]\label{FracDerivative}  The functional derivative of the differential polynomial $f \in A$ is defined as follows:
\begin{equation}
\frac{\delta f}{\delta u_i}= \sum_{n \in \ZZ_{\geq0}} (-\partial)^n \frac{\partial f}{\partial u_i^{(n)}}, \text{ and }  \frac{\delta f}{\delta u} = \left( \frac{\delta f}{\delta u_i} \right)_{i\in I}.
\end{equation}
\end{defn}
It is easy to see that (\ref{Eqn:2.8_010914}) is equivalent to 
\begin{equation}
\int \frac{\delta f}{\delta u} \cdot P=0,
\end{equation}
where $P=(P_i)_{i \in I}$ is same as in (\ref{EE}) and the dot product between $a=(a_i)_{i \in I}$ and $b= (b_i)_{i \in I}$ is $ \sum_{i \in I } a_i b_i$. Moreover, we remark that
\begin{equation} \label{Eqn:1.5}
\frac{\delta }{\delta u}\circ \partial =0.
\end{equation}

\vskip 5mm

In order to discuss Hamiltonian systems, 
let us consider an $l\times l$ matrix\\  $H(\partial)= H_{ij}(u,u', \cdots; \partial)_{i,j\in I}$, where entries are finite order  differential operators in $\partial$ with coefficients in $A$. 
Let $\{\cdot, \cdot\}_H$ be a local Poisson brackets on $A$ associated to $H(\partial)$, i.e. 
\begin{equation*}
\{u_i(x), u_j(y)\}_H =H_{ji} (u(y), u'(y), \cdots; \partial_y) \delta(x-y).
\end{equation*}
Here the distribution $\delta(x-y)$ is defined as follows. ( See \cite{FT}, Introduction in \cite{BDK} for details.)

\begin{defn} \label{Def:delta}
Let $x, y \in S^1$. The distribution $\delta(x-y)$ satisfies $\int_{S^1} \phi(x)  \delta(x-y) dx = \phi(y)$  for any smooth function $\phi\in A.$
\end{defn}

\begin{rem} \label{Rmk:delta}
One can check that  
$\delta(x-y)=\delta(y-x).$ Also, we let the distribution $\partial_x^n \delta(x-y)$ be defined by the equation $\int \phi(x) \partial_x^n \delta(x-y) dx = (- \partial_y)^n \phi(y)$. Then $\partial_x^n \delta(x-y)= (- \partial_y)^n \delta(x-y)$ follows by simple computations.
\end{rem}

Using Leibniz rules of the local Poisson bracket, we have
\begin{equation}\label{Eqn:1.6}
\{f(x), g(y)\} = \sum_{i,j \in I ,m,n \in \ZZ_{\geq0}} \frac{\partial f(x) }{\partial u_i^{(m)}} \frac{\partial g (y)}{\partial u_j^{(n)}} \partial_x^m \partial_y^n \{u_i(x), u_j(y)\}.
\end{equation}
By taking the double integral $\int\int dx dy$ of (\ref{Eqn:1.6}), we get
\begin{equation*}
\left\{ \int f(x) dx, \int g(y) dy \right\}=  \sum_{i,j \in I} \int \frac{\delta g(y)}{ \delta u_j} H_{ji}(u(y), u'(y), \cdots, \partial_y)  \frac{\delta f(y)}{ \delta u_i} dy
\end{equation*}
and, by taking $\int dx$ of both sides of  (\ref{Eqn:1.6}), we obtain
\begin{equation*}
\left\{ \int f(x) dx,  g(y) \right\}=  \sum_{i,j \in I, n \in \ZZ_{\geq0}}  \frac{\partial g(y)}{ \partial u_j^{(n)}} \partial_y^nH_{ji}(u(y), u'(y), \cdots, \partial_y)  \frac{\delta f(y)}{ \delta u_i}. 
\end{equation*}

\begin{defn} \label{Def:bracket}
Let $H(\partial)= H_{ij}(u,u', \cdots; \partial)_{i,j\in I}$ be an $l\times l$ matrix, where entries are finite order  differential operator in $\partial$ with coefficients in $A$. The bracket $\{ \cdot, \cdot\}_H: \int A \otimes \int A \to \int A$ between two local functionals is defined by
\begin{equation} \label{Eqn:1.7}
\textstyle \left\{ \int f, \int g\right\}_H = \int \frac{\delta g}{\delta u} \left( H( \partial) \frac{\delta f}{\delta u}  \right),
\end{equation}
where $H(\partial):= H(u_i, u_i', u_i^{(2)}, \cdots ; \partial)$,
 and the bracket $\{ \cdot, \cdot\}_H: \int A \otimes A \to  A$ between a local functional and a function is defined by
\begin{equation} \label{Eqn:1.8}
\textstyle \left\{ \int f,  g\right\}_H: =  \sum_{i,j \in I, n \in \ZZ_{\geq0}}  \frac{\partial g}{ \partial u_j^{(n)}} \partial^nH_{ji}( \partial)  \frac{\delta f}{ \delta u_i},
\end{equation}
where $ H(\partial):= (H_{ij}(\partial))_{i,j \in I}. $
\end{defn}

Note that the brackets in Definition \ref{Def:bracket} are well-defined by (\ref{Eqn:1.5}). Also,  If we let $H_{ji}(\lambda)=\{u_{i \ \lambda \  } u_j \}_H$, by (\ref{Eqn:1.7}), (\ref{Eqn:1.8}) and (\ref{master formula}), we have
\begin{equation} \label{Eqn:1.10}
\textstyle \left\{\ \int f\   ,\  g\ \right\}_H = \left\{\ f_{ \ \lambda \ } g\ \right\}_H|_{\lambda=0}, \ \ \ \ \left\{\ \int f\ , \ \int g\ \right\}_H = \int \left\{\ \int f\ ,\  g\ \right\}_H.
\end{equation}

\begin{defn}
\begin{enumerate}[(i)]
\item
An $l\times l$ matrix-valued differential operator $H(\partial)$ is called a Poisson structure if 
\begin{equation}
H_{ji}(\lambda):=\{u_{i\ \lambda} u_j\}_H
\end{equation}
gives a PVA structure on $A$.  
\item
Let $H(\partial)$ be a Poisson structure. An evolution equations of the form
\begin{equation}
\frac{du}{dt} = H(\partial) \frac{\delta h}{\delta u}
\end{equation}
is called a Hamiltonian system.  Equivalently, in terms of a $\lambda$-bracket on $A$ 
\begin{equation}  \label{Def:IntSys}
\frac{du}{dt}= \{h\ _\lambda \ u\}_H|_{\lambda=0}
\end{equation}
is called a Hamiltonian system.
\item
Local functionals $\int h, \int h' \in A$ are called to be in involution if $\{ \int h , \int h'\}_H= 0$ and the Hamiltonian system (\ref{Def:IntSys})
is called an integrable system if there are infinitely many linearly independent local functionals $\int h_0, \cdots, \int h_n$, $n=1,2, \cdots$, which are in involution. 
\end{enumerate}
\end{defn}

One can see that if $\int h_n$ for $n \in \ZZ_{\geq0}$ are linearly independent and they are all in involution, then the hierarchy of Hamiltonian equations
$$\left( \frac{du}{dt} = \{ h_{n\ \lambda \ } u\}_H|_{\lambda=0}\right)_{n\in \ZZ_{\geq 0}},$$    
consists of integrable equations. Hence the following observation called Lenard scheme plays a key role in finding integrals of motion.

\begin{Lenard}
Let $H(\partial)$ and $K(\partial)$ be Poisson structures such that 
$$H_{ji}(\lambda):= \{u_{i\ \lambda \ } u_j\}_H, \ \ K_{ji}(\lambda ):= \{u_{i\ \lambda \ } u_j\}_K.$$
If  there are local functions $\int h_0 , \int h_1, \cdots, \int h_n$  satisfying 
\begin{equation} \label{Eqn:Len}
\begin{aligned}
& \{h_{0\ \lambda\ }u_i\}_K|_{\lambda=0}=0, \\
& \{h_{t \ \lambda\ } u_i \}_H|_{\lambda=0}= \{h_{t+1 \ \lambda\ } u_i\}_K|_{\lambda=0}, \ \ t=0, \cdots, n-1,
\end{aligned}
\end{equation}
for $i \in I$, then  $\int h_0, \cdots, \int h_n$ are all in involution with respect to both brackets. 
\end{Lenard}

\begin{thm} \cite{BDK} Let $H(\partial)$ and $K(\partial)$ be a bi-Hamiltonian pair, i.e. $(H+K)(\partial)$ is also a Hamiltonian operator on $A^l$. If the following conditions hold:
\begin{enumerate}[(i)] 
\item $K(\partial)$ is non-degenerate,
\item $H(\partial) \frac{\delta h_0}{\delta u} = K(\partial) \frac{\delta h_1}{\delta u} $,
\item $ \left(\text{span}_\CC\left\{ \frac{\delta h_0}{\delta u}, \frac{\delta h_1}{\delta u}\right\} \right)^\perp\subset K(\partial)$,
\end{enumerate}
then there exists a sequence of local functionals $\int h_i$, $i \in \ZZ_{\geq 0}$, satisfying (\ref{Eqn:Len}).
\end{thm}

\begin{ex} [KdV hierarchy] (See \cite{BDK} for details) \\
Let $R= \CC_{\text{diff}}[u]$ be a PVA with two $\lambda$-brackets $\{\cdot_\lambda\cdot\}_H$ and $\{\cdot_\lambda \cdot\}_K$ such that 
$$\{u_\lambda u\}_H=(\partial+2\lambda)u + c \lambda^3, \ \ \{u_\lambda u\}_K= \lambda,\ \ c\in \CC$$
which means that $H(\partial)= u'+2u\partial +c\partial^3$ and $K(\partial)=\partial.$ Then we have  
\begin{equation}
\begin{aligned}
&  \{u_\lambda u\}_K|_{\lambda=0}=0,\ \ \{u_\lambda u\}_H|_{\lambda=0}= \{\frac{1}{2}u^2_\lambda u\}_K|_{\lambda=0}=\partial u, \\ 
&\{\frac{1}{2}u^2_\lambda u\})_H|_{\lambda=0} = \{\frac{1}{2} u^3+\frac{1}{2}cuu''_\lambda u\}_K|_{\lambda=0} = 3uu'+cu''' .
\end{aligned}
\end{equation}
Hence $h_0=u$, $h_1=\frac{1}{2} u^2$, $h_2= \frac{1}{2} u^3+\frac{1}{2} cuu''$ satisfy (\ref{Eqn:Len}). Recursively, one can find $h_i$ for any $i\in \ZZ_{\geq0}$ and the equations $\frac{du}{dt}= \{h_{i\ \lambda\ }u_i\}_K|_{\lambda=0}$ are all integrable systems. Especially, $\frac{du}{dt}= \{h_{2\ \lambda\ }u_i\}_K|_{\lambda=0}= 3uu'+cu'''$ is the KdV equation.  
\end{ex}

\section{Two equivalent definitions of classical affine $\WW$-algebras}\label{Sec:3}

Let $\g$ be a finite dimensional simple Lie algebra over $\CC$ with an $\sll_2$ triple $(e,h,f).$ The Lie algebra $\g$ can be decomposed into direct sum of eigenspaces of $\ad \frac{h}{2}$, $\g= \bigoplus_{i \in \frac{\ZZ}{2} } \g(i)$, where
$\g(i)=\{g \in \g | [\frac{h}{2},g]=ig\}$. Let 
$$\n=\bigoplus_{i \geq \frac{1}{2} }\g(i) \subset  \m=\bigoplus_{i \geq 1}\g(i)$$
 be two nilpotent subalgebras of $\g$  and let 
  $$\g_f=\{g \in \g| [g,f]=0\}.$$ We fix a symmetric bilinear form $(\cdot, \cdot)$ such that $(e,f)=1$ and $(h,h)=2.$\\
 
\subsection{The first definition of classical affine $\WW$-algebras}\label{Subsec:3.1}
The following three nonlinear Lie conformal superalgebras are needed to
define classical affine $\WW$-algebras. 

\begin{enumerate}
\item
Let $p\in \g$ commute with $\n$, let $k,c \in \CC$ and let $Cur_k^{cp}
(\g)$ be the free $\CC[\partial]$-module $\CC[\partial] \otimes \g$ endowed with the nonlinear $\lambda$-bracket
\begin{equation*}
[a_\lambda b]= [a,b] + \lambda k (a,b)+c(p,[a,b]), \text{ for } \ a,b \in \g.
\end{equation*}
\item
Let $\Pi$ be the parity reversing map and let $\phi_\n$ and
$\phi^{\n^*}$ be purely odd vector superspaces isomorphic to $\Pi\n$ and
$\Pi\n^*$, respectively. The charged free fermion nonlinear Lie
conformal superalgebra $F_{ch}$ is the free $\CC[\partial]$-module
$\CC[\partial] \otimes (\phi_\n \oplus \phi^{\n^*})$ endowed with the nonlinear
$\lambda$-bracket defined by
\begin{equation*}
[\phi_{a \ \lambda} \ \phi^{\theta}] = \theta(a), \ \ a\in \n, \
\theta\in \n^*,
\end{equation*}
\begin{equation*}
[\phi_{a \ \lambda} \ \phi_b]=[\phi^{\theta_1 }_{ \ \ \lambda}
 \ \phi^{\theta_2}]=0, \ \ \ a, b \in \n , \ \ \theta_1, \theta_2 \in
\n^*.
\end{equation*}
\item
Let $\Phi_{\n/\m}$ be a vector space isomorphic to
$\g\left(\frac{1}{2}\right)$ with the skew symmetric bilinear form $\lt
\Phi_{[a]}, \Phi_{[b]}\rt = (f,[a,b]).$ The Neutral free fermion
nonlinear Lie conformal algebra $F_{ne}$ is the free $\CC[\partial]$-module 
$\CC[\partial]\otimes \Phi_{\n/\m}$ endowed with the nonlinear $\lambda$-bracket
\begin{equation*}
[\Phi_{[a] \ \lambda} \Phi_{[b]}] = \lt \Phi_{[a]}, \Phi_{[b]}\rt =
(f,[a,b]), \ \ \ [a],[b] \in \n/\m.
\end{equation*}
\end{enumerate}

We denote the direct sum of the three nonlinear Lie conformal superalgebras $Cur_k^{cp}(\g)$, $F_{ch}$ and $F_{ne}$ by
\begin{equation} \label{Eqn:R_k^{cp}}
R_k^{cp}(\g,f):= \CC[\partial]\otimes
(\g\oplus \phi_\n \oplus \phi^{\n^*} \oplus \Phi_{\n/\m}).
\end{equation}

Let $S(R_k^{cp}(\g,f))$ be the Poisson vertex algebra endowed with the $\lambda$-bracket such that
\begin{align*}
\{ a _\lambda b\}=[a_\lambda b], \ \ a, b \in R_k^{cp}(\g,f),
\end{align*}
and extended by left and right Leibniz
rules.

Let $X: \n \to R_k^{cp}(\g, f)$ be a linear map such that $ X(a)=a +(f,a)+ \Phi_{[a]}, \ a\in \n$, and take the odd element
\begin{equation}\label{Eqn:3.2}
\begin{aligned}
 d &= \sum_{\alpha \in S} \phi^{v^\alpha} (u_\alpha+(f,u_\alpha)+
\Phi_{[u_\alpha]})+\frac{1}{2} \sum_{\alpha, \beta\in S}
\phi^{v^\alpha}\phi^{v^\beta} \phi_{[u_\beta, u_\alpha]}\\
&=\sum_{\alpha \in S} \phi^{v^\alpha} X_{u_\alpha} + \frac{1}{2}
\sum_{\alpha, \beta \in S} \phi^{v^\alpha} \phi^{v^\beta}
\phi_{[u_\beta, u_\alpha]} 
\end{aligned}
\end{equation}
  in $S(R_k^{cp}(\g,f))$, where
$\{u_\alpha| \alpha\in \mathcal{S}\}$ and $\{v^\alpha| \alpha \in \mathcal{S}\}$
are dual bases of $\n$ and $\n^*$. We also note that 
\begin{align*}
\{X_{a\ \lambda} \ X_b\}=X_{[a,b]}.
\end{align*}

\begin{prop} \label{Prop:3.1}
Let the operator $d_{(0)}$ on $S(R_k^{cp}(\g,f))$ be defined by $d_{(0)}(a)= \{d_\lambda a\}|_{\lambda=0}$. Then $d_{(0)}$ is a differential on the complex
$S(R_k^{cp}(\g,f))$.
\end{prop}
\begin{proof}
Since $d$ is an odd element in $S(R_k^{cp}(\g,f))$, we have $\{d_\lambda\{d_\mu
A\}\}=-\{d_\lambda \{d_\mu A\}\}+\{\{d_\lambda d\}_{\lambda+\mu}
A\}$. This implies that if $\{d _\lambda d\}=0$, then $ d_{(0)}^2 A=
\frac{1}{2}\{\{d _\lambda
d\}_{\lambda+\mu} A\}|_{\lambda=\mu=0}=0$. Hence it suffices to show that $\{d_\lambda d\}=0$. \\
Let us compute $\{d_\lambda d\}$ using Jacobi identity and
the skewsymmetry.
The $\lambda$-brackets between the summands of $d$ are as follows:
\begin{enumerate}[(i)]
\item
\begin{equation*}
\sum_{\alpha, \beta \in S}\{ \phi^{v^{\alpha}} X_{u_\alpha \ \lambda}
 \phi^{v^{\beta}} X_{u_\beta} \}= \sum_{\alpha,\beta\in S} \phi^{v^\beta}\phi^{v^\alpha}
X_{[u_\beta,u_\alpha]},
\end{equation*}

\item
\begin{align*}
&\sum_{\alpha, \beta, \gamma \in S}   \{\phi^{v^\alpha}X_{u_\alpha \
\lambda}\phi^{v^\beta} \phi^{v^\gamma} \phi_{[u_\gamma,
u_\beta]}\} \nonumber \\
& =  \sum_{\alpha, \beta, \gamma \in S} \phi^{v^\beta}
\phi^{v^\gamma} v^\alpha( [u_\gamma, u_\beta])X_{u_\alpha}\nonumber = \sum_{\alpha, \beta, \gamma \in S} \phi^{v^\beta} \phi^{v^\gamma}
X_{[u_\gamma, u_\beta]},
\end{align*}

\item
\begin{align*}
&\sum_{\alpha, \beta, \gamma, \delta \in S}
\{\phi^{v^\alpha}\phi^{v^\beta} \phi_{[u_\beta, u_\alpha] \ \lambda}
\phi^{v^\gamma}\phi^{v^\delta} \phi_{[u_\delta,
u_\gamma]}\}\nonumber\\
&=4\sum_{\alpha, \beta, \gamma, \delta \in
S}v^\delta([u_\beta,u_\alpha])\phi^{v^\gamma}\phi^{v^\alpha}\phi^{v^\beta}
\phi_{[u_\gamma, u_\delta]}=4\sum_{\alpha, \beta, \gamma \in
S}\phi^{v^\gamma}\phi^{v^\alpha}\phi^{v^\beta} \phi_{[u_\gamma,
[u_\beta,u_\alpha]]}=0.
\end{align*}
\end{enumerate}
The last equality in $(iii)$ follows from: 
\begin{equation*}
\phi^{v^\gamma}\phi^{v^\alpha}\phi^{v^\beta} \phi_{[u_\gamma,
[u_\beta,u_\alpha]]}+\phi^{v^\alpha}\phi^{v^\beta}\phi^{v^\gamma}
\phi_{[u_\alpha,
[u_\gamma,u_\beta]]}+\phi^{v^\beta}\phi^{v^\gamma}\phi^{v^\alpha}
\phi_{[u_\beta,
[u_\alpha,u_\gamma]]}=0.
\end{equation*}
Hence we have
$$\{d_\lambda d\}=(i)+(ii)+\frac{1}{4}(iii)=\sum_{\alpha, \beta, \gamma \in S}\phi^{v^\gamma}\phi^{v^\alpha}\phi^{v^\beta} \phi_{[u_\gamma,
[u_\beta,u_\alpha]]}=0$$
and $d_{(0)}$ is a differential on the complex $S(R_k^{cp}(\g,f))$.
\end{proof}

\begin{defn}[classical affine $\WW$-algebra (1)] \label{Def:W1aff}
Let $C_1=S(R_k^{cp}(\g,f))$ and $d_1 = d_{(0)}$, where $d$ is defined in (\ref{Eqn:3.2}). The cohomology $H(C_1, d_1)$ is the classical affine $\WW$-algebra $\WW_1(\g, k,f)$ associated to  $\g$, $f$ and $k\in \CC$ .
\end{defn}

\begin{rem} \label{Note:3.3}
\begin{enumerate}[(i)]
\item

Consider the PVA isomorphism
\begin{equation}\label{Eqn:3.2_2}
\psi^{cp}: S(R_k^0(\g,f) )\to S( R_k^{cp}(\g,f))  
\end{equation}such that 
$\psi^{cp}(a)= \left\{
\begin{array}{ll}
a+[cp, a] & \text{if}\ a\in \g, \\
a & \text{if}\ a \in \phi_n\oplus \phi^{n^*} \oplus \Phi_{\n/\m}.
\end{array}\right.$
 Then $\psi^{cp}$ induces the well-defined PVA isomorphism $\overline{\psi^{cp}}: H(S(R_k^0(\g,f)), d_1) \to H(S(R_k^{cp}(\g,f)), d_1).$ Hence the PVA structure on $H(S(R_k^{cp}(\g,f)), d_1)$ is independent on $c$ and $p$.

\item Let $\{u_\alpha| \alpha \in \overline{\mathcal{S}}\}$ and $\{u^\alpha|\alpha \in \overline{\mathcal{S}}\}$ be dual bases of $\g$, where $u_\alpha \in \g(j_\alpha)$. Let $\{z_1, \cdots ,z_{2s}\}$ and  $\{z_1^*, \cdots ,z_{2s}^*\}$ be bases of $ \g\left(\frac{1}{2}\right)$ such that $(f, [z_i, z_j^*])=\delta_{i,j}$. Then there is an energy-momentum field $\bar{L}$ of $\WW_1(\g,k,f)$ which is the projection of the element    
\begin{equation} \label{Eqn:3.3}
L=\frac{1}{2k}\sum_{\alpha \in \overline{S}} u_\alpha u^\alpha +\partial x -\sum_{\alpha \in S} j_\alpha\phi^{v^\alpha} \partial\phi_{u_\alpha}+\sum_{\alpha \in S} (1-j_\alpha)\partial \phi^{v^\alpha} \phi_{u_\alpha}+\frac{1}{2}\sum_{i=1}^{2s}\partial \Phi_{[z_i^*]}\Phi_{[z_i]}
\end{equation}
 of $S(R_k^{cp}(\g,f))$. By direct computations, we have $\{\bar{L}_\lambda \bar{L}\}=(\partial+2\lambda) \bar{L} -\frac{1}{2} k \lambda^3.$
\end{enumerate}
\end{rem}

\subsection{The second definition of classical affine $\WW$-algebras}\label{Subsec:3.2}
Let $S(\CC[\partial]\otimes\g)$ be the symmetric algebra generated by
$\CC[\partial]\otimes\g$ endowed with the $\lambda$-bracket induced from the $\lambda$-bracket on $Cur_k^{cp}(\g)$ and Leibniz
rules. Consider the differential algebra ideal 
\begin{equation}\label{Eqn:3.4}
J = \lt m-\chi(m) | m \in \m , \chi(m)=(f,m) \rt
\end{equation} 
of $S(\CC[\partial]\otimes\g)$. The $\lambda$-adjoint action of $\n$ on $S(\CC[\partial]\otimes \g)/J$ is defined by 
$$(\ad_\lambda n) (a+J) = \{n_\lambda a\} + J[\lambda].$$
For given $n \in \n$, $x \in \m$ and $A \in S(\CC[\partial]\otimes \g)$, we have
\begin{equation*}
\{n _\lambda A(m-\chi(m))\} =A\{n_\lambda
m-\chi(m)\}+\{n_\lambda A\}(m-\chi(m)) \in J[\lambda].
\end{equation*}
Hence the $\ad_\lambda \n$ action is well-defined and we can define the vector space:
$$ \WW_2(\g,k,f)=\left(S(\CC[\partial]\otimes\g)/ \lt m+ \chi(m) : m\in \m\rt
\right)^{\ad_\lambda \n}.$$
The following proposition shows that the multiplication and
the Poisson $\lambda$-bracket on $\WW_2(\g,f,k)$ can be defined by $(a+J)\cdot (b+J) =ab+J$ and $\{a+J
_\lambda  b+J\}=\{a_\lambda b\}+J[\lambda]$.

\begin{prop} The multiplication and the Poisson $\lambda$-bracket are
well-defined on $\WW_2(\g,k,f)$.
\end{prop}
\begin{proof}
Let $a+J, b+J \in \WW_2(\g, k, f)$ and $n \in \n$. Then we have
 $$(\ad_\lambda \ n) \ ((a+J)\cdot
(b+J))= (b+J)((\ad_\lambda \ n) \ (a+J))+(a+J)((\ad_\lambda \ n)
(b+J))\in 0+J[\lambda]$$ by the definition of $\WW_2(\g, k, f)$. Hence the
multiplication is well-defined on $\WW_2(\g, k, f)$. 

Moreover, one can check that $ \{a+A(m_1-\chi(m_1)) _\lambda
b+B(m_2-\chi(m_2))\} \in \{a_\lambda b\}+J[\lambda]$, for any $A, B \in S(\CC[\partial]\otimes \g)$,
and 
$(\ad_\lambda n) \{a_\mu b\}+J[\mu] \subset 0+J[\lambda,\mu]$, by Jacobi identity.  Hence the Poisson $\lambda$-bracket is also well-defined on $\WW_2(\g,k,f)$.
\end{proof}

\begin{defn}[classical affine $\WW$-algebra (2)] \label{Def:W2aff}
The classical affine $\WW$-algebra associated to $\g$, $f$ and $k$
is the PVA $$\WW_2(\g,f,k)=\left(S(\CC[\partial]\otimes\g)/ \lt m+ \chi(m) |  m\in
\m\rt \right)^{\ad_\lambda \n}.$$
\end{defn}

\begin{rem}
Let  ($\WW^{(cp)}(\g, f,k), \{\cdot_\lambda \cdot \}^{(cp)}$), where $c\in \CC$ and $p \in \mathfrak{z}(\n)$,  be the classical affine $\WW$-algebra associated to  $\g$ and $f$ endowed with the Poisson bracket induced from the bracket $\{ \cdot_\lambda \cdot\}^{(cp)} $ on $S(\CC[\partial]\otimes \g)$ such that  
$$\{a_\lambda b\}^{(cp)} = [a,b]+c(p,[a,b])+\lambda k(a,b), \ \ a, b \in \g.$$
We have the PVA isomorphism
\begin{equation}\label{Eqn:equiv}
 \psi^{(cp)}:(S(\CC[\partial]\otimes \g), \{\cdot_\lambda \cdot\}^{(0)}) \simeq  (S(\CC[\partial]\otimes \g), \{\cdot_\lambda \cdot\}^{(cp)})
\end{equation}
defined by $\psi^{(cp)}(a) = a+[cp,a]$,  $a \in \g$. Moreover, $\psi^{(cp)}$ induces the well-defined PVA isomorphism between $\WW^{(0)}(\g,f,k) $ and $\WW^{(cp)}(\g,f,k).$ Hence the Poisson algebra structure on $\WW(\g,k,f)$ is independent on $c$ and $p$. 
\end{rem}

\subsection{Equivalence of the two definitions of classical affine $\WW$-algebras}\label{Subsec:3.3}
In this section, we assume $p=0$.
Given $a \in \g$, let us denote 
\begin{equation}\label{Eqn:3.5}
 J_{a}:= a+ \sum_{\alpha \in S} \phi^{v^\alpha}
\phi_{\{u_\alpha,a\}}.
\end{equation}
 Then
\begin{equation*}
S(\CC[\partial]\otimes (\g \oplus \phi_\n \oplus \phi^{\n^*} \oplus
\Phi_{\n/\m})) = S(\CC[\partial]\otimes (J_\g \oplus \phi_\n \oplus
\phi^{\n^*} \oplus \Phi_{\n/\m})).
\end{equation*}
The Poisson $\lambda$-brackets between $J_a$ and elements in $R_k^0=J_\g \oplus \phi_\n \oplus
\phi^{\n^*} \oplus \Phi_{\n/\m}$ are as follows:

\begin{align}
&\{J_a \ _\lambda\ \phi_n \} = \{a+ \sum_{\alpha \in S}
\phi^{v^\alpha} \phi_{[u_\alpha,a]} \ _\lambda\ \phi_n\} = -
\sum_{\alpha \in S} v^\alpha(n) \phi_{[u_\alpha,
a]}=\phi_{[a,n]},\nonumber\\ 
&\{J_a \ _\lambda\ \phi^\theta\}  =\{a+ \sum_{\alpha \in S}
\phi^{v^\alpha}\phi_{[u_\alpha, a]} \ _\lambda\  \phi^\theta\}=
\sum_{\alpha\in S} \phi^{v^\alpha} \{\phi_{[u_\alpha,a]} \ _\lambda\
\phi^\theta \} =\phi^{a \cdot \theta},\nonumber\\
&\{J_a \ _\lambda\ \Phi_{[n]}\}  =0, \nonumber\\
\label{Eqn:3.6}
&\{J_a \ _\lambda\ J_b\} = J_{[a,b]}  +  k \lambda (a,b)  + \sum_{\alpha \in S}
\phi^{v^\alpha} (-\phi_{[\pi_\leq[u_\alpha, a],
b]}+\phi_{[\pi_\leq{[u_\alpha, b],a]}}),
\end{align}
where $\pi_\leq : \g\to \bigoplus_{\leq 0} \g(i)$ is the projection map and  $a, b \in \g$, $n \in \n$, $\theta\in \n^*$. 
Equation (\ref{Eqn:3.6}) shows that if $a$ and $b$ are both in
$\bigoplus_{i\geq 0}\g(i)$ or both in $\bigoplus_{i\leq 0}\g(i)$ then
\begin{equation}\label{Eqn:3.7}
\{J_a \ _\lambda\ J_b\} = J_{[a,b]}+\lambda k (a,b).
\end{equation}
Let  $J_\leq=\{J_a | a \in \bigoplus_{i\leq 0} \g(i)\}$ and let
$r_+= \phi_\n \oplus d_{(0)}\phi_\n,$ $r_-=J_\leq  \oplus
\phi^{\n^*} \oplus \Phi_{\n/\m},$ $R_+=\CC[\partial]\otimes r_+ ,$ and $R_-=\CC[\partial]\otimes r_-.$  Since
\begin{equation}\label{Eqn:3.8}
 \{d \ _\lambda J_a\}  = \sum_{\alpha\in S} \left(\phi^{v^\alpha}
(J_{\pi_\leq [u_\alpha, a]} -  \Phi_{[u_\alpha, a]}- (u_\alpha ,
[a,f])) + k (\partial+\lambda)(u_\alpha, a) \phi^{v^\alpha} \right)
\end{equation}
and $d_{(0)} (\phi_a)= J_a +(a,f)+\Phi_{[a]}$, we have
$$ d_{(0)} (S(R_+)) \subset S(R_+) \text{\ and\ } d_{(0)} (S(R_-)) \subset S(R_-).$$
Hence we can define cohomologies $H(S(R_+), d_+)$ and $H(S(R_-),d_-)$, where $d_+= d_{(0)}|_{S(R_+)}$ and  $d_-= d_{(0)}|_{S(R_-)}$.  By K$\ddot{\text{u}}$nneth lemma, we have $H(C_1, d_{(0)})= H(S(R_+), d_+)\otimes
H(S(R_-),d_-) $. Moreover, it is easy to see that   $\ker(d_+)= \text{im}(d_+)= d_{(0)} \phi_\n$  which implies that $H(S(R_+),d_+)=S(H(R_+,d_+))=\CC$ and 
\begin{equation}\label{Eqn:3.9}
H(C_1,d_1)=H(S(R_-),d_-)=S(H(R_-,d_-)).
\end{equation}

\vskip 5mm

Let us define the degree and the charge on $R_-$ as follows:
\begin{equation}\label{Eqn:3.10}
\begin{aligned}
 \deg (J_a)=& -j+ \frac{1}{2},\ \  \deg (\phi^{v^\alpha})=j_\alpha-\frac{1}{2},\ a \in \g(j),\ u_\alpha \in \g(j_\alpha),\\
 &\deg (\Phi_{[n]})=0, \ \  \deg (\partial)=0,\  n \in \n, \\
 \text{charge}(J_a)=\text{charge}(\Phi_{[n]})&=\text{charge}(\partial)=0, \text{charge} (\phi^\theta)=1, \ a \in \bigoplus_{i \leq 0} \g(i), n \in \n, \theta \in \n^*.
\end{aligned}
\end{equation}
We denote by $\{F_i\}_{i \in \frac{\ZZ}{2}}$ the increasing filtration  with respect to the degree on $S(R_-)$: 
\begin{equation}\label{Eqn:filt_2}
\cdots \subset F_{p-\frac{1}{2}}(S(R_-)) \subset F_{p}(S(R_-)) \subset F_{p+\frac{1}{2}}(S(R_-)) \subset \cdots,
\end{equation}
where $F_p(S(R_-)) = \bigoplus_{i \leq p} S(R_-)(i)$ and $S(R_-)(i) = \{A \in S(R_-)|\deg(A) =i\}$. \\

Then we have the graded superalgebra  $\gr S(R_-):=\bigoplus_{i\in \frac{\ZZ}{2}} F_{i+\frac{1}{2}}(S(R_-))/F_i(S(R_-))$ of $S(R_-)$ and the graded differential  $\gr d_-: \gr S(R_-) \to \gr S(R_-)$ which is induced from the differential $d_-$  on $S(R_-)$. Since $\ \text{im}(\gr d_-|_{S(R_-)})=\CC[\partial]\otimes \phi^{\n^*}$ and $\ \ker (\gr d_-|_{S(R_-)})=\CC[\partial] \otimes\phi^{\n^*} \oplus\CC[\partial] \otimes  \lt J_a | a \in \g_f  \rt ,$ by K$\ddot{\text{u}}$nneth lemma, we obtain 
\begin{equation}\label{Eqn:3.11_2}
H(\gr S(R_-), \gr d_-) =H^0(\gr S(R_-), \gr d_-)= S(\CC[\partial]\otimes \lt J_a| a\in \g_f\rt).
\end{equation}


\begin{prop} \label{Prop:3.7}
Let $\{u_i| i \in I\} \text{ be a basis of } \g_f$ consisting of eigenvectors of $\ad \frac{h}{2}$ and assume that $u_i \in  \g\left(\frac{n_i}{2}\right) $. Then there exists  $A_{u_i}\in F_{-\frac{n_i}{2}}(S(R_-))$ with charge $0$ such that $H(S(R_-),d_-)$ is generated by 
$$ \{J_{u_i}+ A_{u_i}\}_{i \in I}$$
as a differential algebra.
\end{prop}
\begin{proof}
Let $L$ be the energy-momentum field introduced in  (\ref{Eqn:3.3}). The complex $S(R_-)$ is locally finite,
since, for each $i \in \frac{\ZZ}{2}$, the subspace 
$$L_i:=\{a\in S(R_-) | \{L_\lambda a\}=(\partial+i\lambda)a+o(\lambda) \}$$ 
 is finite dimensional and $d_-|_{L_i}\subset L_i$. Hence  we have 
$$ H(\gr S(R_-), \gr d_-) \simeq \gr H(S(R_-), d_-).$$
Recall that  $H(\gr S(R_-), \gr d_-)=  S(\CC[\partial]\otimes \lt J_a| a\in \g_f\rt)$ and $\deg (J_{u_i})= -\frac{n_i}{2}+\frac{1}{2}$. Thus the cohomology $H(S(R_-), d_-)$ is generated by the set of the form  $\{J_{u_i}+ A_{u_i}\}_{i\in I}$, where $A_{u_i}\in F_{-\frac{n_i}{2}}(S(R_-))$. 
Moreover, by equation (\ref{Eqn:3.11_2}),  any element in $H(S(R_-),d_-)$ has a representative of charge $0$. Hence we can choose $A_{u_i}$ with charge $0$. 

\end{proof}

The next theorem is our main purpose of this section.

\begin{thm} \label{Thm:3.11}
Let $i: S(R_-) \to S(\CC[\partial]\otimes \g)/\lt m-\chi(m)| m \in \m \rt$ be an associative 
superalgebra homomorphism defined by
\begin{equation}\label{Eqn:3.11}
\partial^k J_a  \mapsto \partial^k a, \ \partial^k \Phi_b  \mapsto - \partial^k b, \ \partial^k \phi^\theta  \mapsto 0
\end{equation}
where  $k\in \ZZ_{\geq 0}$, $a \in \bigoplus_{i\leq 0}\g(i)$,
$b \in \g\left(\frac{1}{2}\right)$ and $\theta \in \n^*$. Then
\begin{enumerate}
\item the homomorpshism $i$ induces an associative algebra isomorphism 
\begin{equation} \label{Eqn:3.12}
j:H(S(R_-), d_-) \to (S(\CC[\partial]\otimes \g)/\lt m-\chi(m): m \in \m \rt)^{\ad_\lambda \n}.
\end{equation}
\item the isomorphism (\ref{Eqn:3.12}) is a Poisson vertex algebra isomorphism.
\end{enumerate}
Hence $\WW_1(\g,f,k)$ and $\WW_2(\g,f,k)$ are isomorphic as PVAs and we denote the classical affine $\WW$-algebra associated to $\g$, $f$ and $k$ by $\WW(\g,f,k).$
\end{thm}

\begin{proof}
In this proof, we use the following notations.
\begin{notation}\label{Notation:3.9}
We denote
$$A_i:= a_{i,1} \cdots a_{i,k_i}, N_i:= n_{i,1} \cdots n_{
i,l_i}, G:= g_1 g_2 \cdots g_m,$$
$$\partial^{R_i} A_i:=\partial^{r_{i,1}}a_{i,1}
 \cdots \partial^{r_{i,k_i}}a_{i,k_i}, \partial^{S_i}N_i:=\partial^{s_{i,1}}n_{i,1}\cdots
\partial^{s_{i,l_i}}n_{i,l_i}, \partial^{T}G:= \partial^{t_1} g_1 \cdots \partial^{t_m} g_m$$
for $a_{i,j}\in \bigoplus_{i\leq 0}\g(i), \  n_{i,j}\in \g\left(\frac{1}{2}\right),\ g_i \in \bigoplus_{\alpha \leq 1} \g_\alpha$ and $R_i=(r_{i,1}, \cdots, r_{i,k_i}), S_i=(s_{i,1}, \cdots, s_{i,l_i}), T= (t_1, \cdots, t_m).$ Also, let
$$J_{A_i}:= J_{a_{i,1}} \cdots J_{a_{i,k_i}}, \ \partial^{R_i} J_{A_i}:=\partial^{r_{i,1}}J_{a_{i,1}}\cdots \partial^{r_{i,k_i}}J_{a_{i,k_i}},$$
$$  \Phi_{N_i}:=\Phi_{[n_{i,1}]}\cdots \Phi_{[n_{i,l_i}]}, \ \partial^{S_i}
\Phi_{N_i}:=\partial^{s_{i,1}}\Phi_{[n_{i,1}]}\cdots
\partial^{s_{i,l_i}}\Phi_{[n_{i,l_i}]},$$
$$ K_{g_i}:= J_{g_i}-\Phi_{[g_i]}-(g_i,f), \ K_G:= K_{g_1} \cdots K_{g_m}, \ \partial^{T}K_G:= \partial^{t_1} K_{g_1} \cdots \partial^{t_m} K_{g_m}.$$
\end{notation}

\vskip 5mm 

 Since $i(\phi^{\n^*})=0$, we have $i(d_-(S(R_-))=0$. Hence the homomorphism $i$ induces the well-defined associative algebra homomorphism 
 $$\overline{i}:H(S(R_-), d_-) \to S(\CC[\partial]\otimes \g)/\lt m-\chi(m)| m \in
 \m\rt.$$ 
 By Proposition \ref{Prop:3.7}, any element in $H(S_-, d_-)$ has a
unique representative of the form $\sum_{i\in I}\partial^{R_i}J_{A_i} \partial^{S_i}\Phi_{N_i}$.  Hence $\overline{i}$ is injective.

Let us show that the image under  $\overline{i}$ is
$(S(\CC[\partial]\otimes \g)/\lt m-\chi(m)| m \in \m \rt)^{\ad_\lambda \n}$, i.e.
\begin{equation}\label{Eqn:3.13}
d_-\left(\sum_{i\in I} \partial^{R_i} J_{A_i} \partial^{S_i}
\Phi_{N_i}\right)=0 \text{ if and only if }\{\n _\lambda \sum_{i\in
I}\partial^{R_i} A_i
\partial^{S_i} N_i \}=0.
\end{equation}
We notice that $K_g=J_g$ if $g \in \bigoplus_{i\leq 0}\g(i) $ and $K_g=-\Phi_{[g]}$ if $g \in \g\left(\frac{1}{2}\right)$.
Hence any element in $H(C_1, d_1)$ can be written as a polynomial in
$\partial^t K_g$ where $g \in \bigoplus_{i\leq 1} \g(i)$ and $t\in
\ZZ_{\geq 0}$.  By (\ref{Eqn:3.8}), we obtain 
\begin{equation}\label{Eqn:3.14}
d_-(K_g)= \sum_{\alpha\in S} \phi^{v^\alpha} K_{\{u_\alpha, g\}}
+k\sum_{\alpha\in S}
\partial \phi^{v^\alpha} (u_\alpha, a).
\end{equation}
Since $d$ is a derivation, by equation (\ref{Eqn:3.14}),  we have
\begin{equation}\label{Eqn:3.15}
\begin{aligned}
 d_-(\partial^T K_G) = &\sum_{j=1}^m
\partial^{t_j}\left(\sum_{\alpha \in S} \phi^{v^\alpha}K_{[u_\alpha,
g_j]} + k \sum_{\alpha \in S}\partial
\phi^{v^\alpha}(u_\alpha,g_j)\right) \\
&
\ \ \ \ \ \ \ \ \cdot \partial^{t_1}K_{g_1} \cdots \partial^{t_{j-1}}K_{g_{j-1}}\partial^{t_{j+1}}K_{g_{j+1}}\cdots\partial^{t_m}K_{g_m}. \\
\end{aligned}
\end{equation}
Here,
\begin{align}
&\partial^{t_j}\left(\sum_{\alpha \in S} \phi^{v^\alpha}K_{[u_\alpha, g_j]}
+ k  \sum_{\alpha \in S} \partial \phi^{v^\alpha} (u_\alpha,
g)\right) \nonumber \\
&  =\sum_{\alpha \in S} \sum_{p=0}^{t_j} { t_j \choose p}
\partial^p \phi^{v^\alpha} \partial^{t_j-p} K_{[u_\alpha, g_j]} + k
\sum_{\alpha \in S} \partial^{t_j+1} \phi^{v^\alpha} (u_\alpha,
g_j).\label{Eqn:3.16}
\end{align}
On the other hand,  $\partial^T G=\overline{i}\left(\partial^T K_G \right)$ and we have 
\begin{equation}\label{Eqn:3.17} 
\{u_\alpha \ _\lambda\ \partial^T G\} = \sum_{j=1}^t
(\partial+\lambda)^{t_j} \left([u_\alpha, g_j]+ k\lambda(u_\alpha,
g_j)\right)\partial^{t_1}g_1 \cdots
\partial^{t_{j-1}}g_{j-1}\partial^{t_{j+1}}g_{j+1}\cdots\partial^{t_m}g_m.
\end{equation}
Here, 
\begin{equation}\label{Eqn:3.18}
(\partial+\lambda)^{t_j}([u_\alpha, g_j]+k \lambda(u_\alpha, t_j))=\sum_{p=0}^{t_j} {t_j \choose p} \lambda^p
\partial^{t_j-p}[u_\alpha, g_j]+k\lambda^{t_j+1}(u_\alpha, g_j).
\end{equation}
Let us assume that $X\in H(R_-, d_-)$ and  $Z_{i,\alpha} \in S[\partial^j K_a| j\in \ZZ_{\geq 0}, a\in
\bigoplus_{i\leq 1} \g(i)]$  satisfy the equation: $
d_-(X)=\sum_{\alpha \in S, i\in \ZZ_{\geq0}} Z_{i,\alpha} \partial^i
\phi^{v^\alpha}.
$
 Then, by (\ref{Eqn:3.15})-(\ref{Eqn:3.18}), the image $z_{i,\alpha}$ of $Z_{i,\alpha}$ under $\overline{i}$ satisfies the equation
$
 \{u_\alpha \ _\lambda\ \overline{i}(X)\}= \sum_{i\in \ZZ_{\geq0}}
z_{i,\alpha} \lambda^i.
$
Hence we conclude that
[$ d_-(X)=0 \Leftrightarrow Z_{i,\alpha}=0 \text{  for any  }i, \alpha \Leftrightarrow
 z_{i, \alpha}=0 \text{  for any  } i, \alpha \Leftrightarrow \{\n \ _\lambda\
\overline{i}(X)\}=0 $]. 

To prove the second part of the theorem,  let us pick two elements 
$$\sum_{i\in I} \partial^{R_i} J_{A_i} \partial^{S_i}\Phi_{N_i}, \sum_{i'\in
I'}\partial^{R_{i'}} J_{A_{i'}} \partial^{S_{i'}}\Phi_{N_{i'}}\in H(S(R_-), d_-).$$
Then the images of them under $j$ are
 $\sum_{i\in I}(-1)^{l_i}
\partial^{R_i} A_i \partial^{S_i}N_i,$ and $\sum_{i'\in I'}(-1)^{l_{i'}}\partial^{R_{i'}}
A_{i'} \partial^{S_{i'}}N_{i'}$ in $(S(\CC[\partial]\otimes \g)/\lt
m-\chi(m)\rt)^{\ad_\lambda \n}$. We want to show that
\begin{equation}\label{Eqn:3.21}
\begin{aligned}
j&\left(\left\{ \sum_{i\in I}
\partial^{R_i} J_{A_i} \partial^{S_i}\Phi_{N_i} \ _\lambda\ \sum_{i'\in
I'}\partial^{R_{i'}} J_{A_{i'}}  \partial^{S_{i'}}\Phi_{N_{i'}}
\right\}\right)\\
 &= \left\{ \sum_{i\in I} (-1)^{l_i}
\partial^{R_i} A_i \partial^{S_i}N_i \  _\lambda\ \sum_{i'\in I'}(-1)^{l_{i'}}\partial^{R_{i'}}
A_{i'} \partial^{S_{i'}}N_{i'}\right\}.
\end{aligned}
\end{equation}
Since $\{J_a \ _\lambda \Phi_{[n]}\}=0$, we have
\begin{equation}\label{Eqn:3.22}
\begin{aligned}
& \left\{ \sum_{i\in I} \partial^{R_i} J_{A_i}
\partial^{S_i}\Phi_{N_i} \ _\lambda\ \sum_{i'\in I'}\partial^{R_{i'}}
J_{A_{i'}} \partial^{S_{i'}}\Phi_{N_{i'}}
\right\}  \\
& \ \ \ \ \ =  \sum_{i\in I, i'\in I} \partial^{S_{i'}}\Phi_{N_{i'}} \{
\partial^{R_i}J_{A_i} \
_{\lambda+\partial}\partial^{R_{i'}}J_{A_{i'}}\}_{\to}
\partial^{S_i} \Phi_{N_i}\\
&\ \ \ \ \ +  \sum_{i\in I, i'\in I'} \partial^{R_{i'}}J_{A_{i'}} \{
\partial^{S_i} \Phi_{N_i} \ _{\lambda+\partial}\ \partial^{S_{i'}} \Phi_{N_{i'}}\}_{\to}
\partial^{R_i} J_{A_i}.
\end{aligned}
\end{equation}
Using the following two equalities:
\begin{align*}
 & j(\{J_a \ _\lambda J_b\})=i(J_{[a,b]}
+\lambda k(a,b))= \{a \ _\lambda\ b\} \text{ for } a,b \in
\bigoplus_{i\leq 0} \g(i);\\
&j\{\Phi_{n_1} \ _\lambda\ \Phi_{n_2}\}=(f,
[n_1,n_2]) =-\{n_{1 \ \lambda\ }n_2\} \text{ in } S(\CC[\partial]\otimes \g)/\lt m
-\chi(m)\rt, \text{ for } n_1,n_2 \in \g\left(\frac{1}{2}\right);
\end{align*}
 we  get the equation
\begin{equation}\label{Eqn:3.23}
\begin{aligned}
 j \left(\left\{ \sum_{i\in I} \right.\right.&\left.\left.
\partial^{R_i} J_{A_i} \partial^{S_i}\Phi_{N_i} \ _\lambda\ \sum_{i'\in
I'}\partial^{R_{i'}} J_{A_{i'}}
\partial^{S_{i'}}\Phi_{N_{i'}} \right\}\right) \\
  =(-1)^{l_i+l_{i'}}&  \left( \sum_{i\in I, i'\in I} \right.\partial^{S_{i'}}N_{i'} \{
\partial^{R_i}A_i \ _{\lambda+\partial}\ \partial^{R_{i'}}A_{i'}\}_{\to}
\partial^{S_i} N_i \\
&- \left.\sum_{i\in I, i'\in I'} \partial^{R_{i'}}A_{i'} \{
\partial^{S_i}N_i \ _{\lambda+\partial}\ \partial^{S_{i'}} N_{i'}\}_{\to}
\partial^{R_i} A_i\right).
\end{aligned}
\end{equation}
On the other hand, the RHS of (\ref{Eqn:3.21}) can be rewritten as follows.
\begin{equation}\label{Eqn:3.24}
\begin{aligned}
& \{\sum_{i\in I}\partial^{R_i}A_i \partial^{S_i}N_i \ _\lambda
\sum_{i'\in I'}
\partial^{R_{i'}} A_{i'} \partial^{S_{i'}}N_{i'}\}
=  \sum_{i'\in I'} \partial^{S_{i'}}N_{i'} \{\sum_{i\in I}
\partial^{R_i} A_i \partial^{S_i} N_i \ _\lambda \partial^{R_{i'}}
A_{i'}\} \\
= & \sum_{i\in I, i'\in I'} \partial^{S_{i'}} N_{i'}
\{\partial^{R_i}A_i \ _{\lambda+\partial}
\partial^{R_{i'}}A_{i'}\} _\to \partial^{S_i}N_i - \sum_{i\in I, i'\in I'} \partial^{R_{i'}}A_{i'}
\{\partial^{S_i}N_i \ _{\lambda+\partial} \partial^{S_{i'}} N_{i'}
\} _\to \partial^{R_i} A_i.
\end{aligned}
\end{equation}
Here, we used the fact that $\{\n \ _\lambda \sum_{i\in I}\partial^{R_i}A_i
\partial^{S_i} N_i\}=\{\n \ _\lambda \sum_{i'\in I'}
\partial^{R_{i'}} A_{i'} \partial^{S_{i'}}N_{i'}\}=0.$
Hence
 (\ref{Eqn:3.21}) follows from (\ref{Eqn:3.23}) and (\ref{Eqn:3.24}) and hence the map $j$ defined in (\ref{Eqn:3.12}) is a
Poisson vertex algebra isomorphism between the two definitions of
classical affine W-algebras.
\end{proof}

\section{ Two equivalent definitions of classical affine fractional $\WW$-algebras } \label{Sec:7}

We introduce notations (\ref{Eqn:7.3_3})-(\ref{Eqn:7.14_2}) which are used in the rest of the paper. \\

Let us denote $\hat{\g}:= \g[[z]] z \oplus \g[-z]$. The Lie bracket and the symmetric bilinear form on $ \hat{\g}$ are defined by
\begin{equation}\label{Eqn:7.3_3}
[g z^i, h z^j]:= [g,h] z^{i+j},  \ \ 
(g z^i, h z^j) := (g,h)\delta_{i+j,0}, \ \ g, h \in \g.
\end{equation}
Also, we consider two different gradations $\gr_1$ and $\gr_2$ on $\hat{\g}$ given by eigenvalues of the operations $\partial_1= z \partial_z$ and $\partial_2= (d+1) z \partial_z+\ad \frac{h}{2}$, where $d$ is the largest eigenvalue of $\ad \frac{h}{2}:\g\to \g$ . 
Then the two gradations have the following properties:
\begin{equation}\label{Eqn:7.2_2}
\gr_1 (z)=1, \ \gr_1(\g)=0,\  \gr_2(z)=d+1,  \ \gr_2(g)= i \text{ for } g\in \g_i.
\end{equation}
If $n_1 \in \ZZ$ and $n_2 \in \frac{\ZZ}{2}$, then we denote subspaces of $\hat{\g}$ as follows:
\begin{equation}\label{Eqn:7.3_2}
\begin{aligned}
&\hat{\g}^{n_1}:= \g z^{n_1}, \ \ \hat{\g}_{n_2}:=\{g\in \hat{\g} | \gr_2(g)=n_2 g \}, \ \ \hat{\g}_{n_2}^{n_1} := \hat{\g}^{n_1} \cap \hat{\g}_{n_2},\\
&\hat{\g}_{<i} := \bigoplus_{n_2 <i} \hat{\g}_{n_2}, \  \hat{\g}_{>i} := \bigoplus_{n_2 >i}  \hat{\g}_{n_2},\ \ \hat{\g}^{\leq i}:=\bigoplus_{n_1 \leq i} \hat{\g}^{n_1},\ \  \hat{\g}_{n_2},\ \ \hat{\g}^{\geq i}:=\bigoplus_{n_1 \geq i} \hat{\g}^{n_1}.
\end{aligned}
\end{equation}
In particular, we have $p z^j \in \hat{\g}^{j}_{(d+1)j+d}$ and $ f  z^j \in \hat{\g}^j_{(d+1)j-1}$.

For a given integer $m \geq 0$, there are two quotient spaces and two subspaces of $\hat{\g}$ which play important roles in the following sections. The four vector spaces are denoted by:
\begin{equation}\label{Eqn:7.4_2}
\begin{aligned}
 \hat{\g}^+_{[d,m]}:= \hat{\g}^{\geq 0}/\hat{\g}_{>(d+1)m+1}, &\ \ \hat{\g}^-_{[d,m]} :=  \hat{\g}^{\leq 0}/\hat{\g}_{<-(d+1)m-1},\\
 \hat{\g}^{\geq 0}_{<(d+1)m+1}:= \hat{\g}^{\geq 0}\cap \hat{\g}_{<(d+1)m+1}, &\ \ \hat{\g}^{\leq 0}_{>-(d+1)m-1} :=  \hat{\g}^{\leq 0}\cap \hat{\g}_{>-(d+1)m-1}.
\end{aligned}
\end{equation}
We notice that two subspaces $ \hat{\g}^{\geq 0}_{<(d+1)m+1}$ and $ \hat{\g}^{\leq 0}_{>-(d+1)m-1} $ are dual with respect to the bilinear form $(\cdot, \cdot)$ on $\hat{\g}$. Also, there are vector space isomorphisms $
 \iota: \hat{\g}^{\geq 0}_{<(d+1)m+1} \oplus \hat{\g}_{(d+1)m+1}  \simeq \hat{\g}^+_{[d,m]}  \text{ and }  \iota^-:  \hat{\g}^{\leq 0}_{<-(d+1)m-1}\oplus \hat{\g}_{-(d+1)m-1}  \simeq \hat{\g}^-_{[d,m]} $ such that $\iota(a)=\bar{a}$ and $\iota^-(b)=\bar{b}$.  


Let us denote 
\begin{equation} \label{Eqn:7.3}
 \Lambda_m = -fz^{-m}-pz^{-m-1}\in \hat{\g}_{-(d+1)m-1}
\end{equation}
and let $ \chi $ be an element in $ (\hat{\g}_{(d+1)m+1})^*$ defined by 
\begin{equation}
 \chi(a) = (\Lambda_m, a),\ a \in \hat{\g}_{(d+1)m+1} . 
\end{equation}

In this section, we mainly deal with a differential algebra
\begin{equation}\label{Eqn:7.4}
 \mathcal{V}_{[d,m]}:=S(\CC[\partial] \otimes  \hat{\g}^+_{[d,m]})/I,
\end{equation}
where
$ I \text{ is the differential algebra ideal of   } S(\CC[\partial]\otimes \hat{\g}^+_{[d,m]})  \text{ generated by } 
\{ \iota(a)- \chi(a) | a \in \hat{\g}_{(d+1)m+1} \} .$ Then $S(\CC[\partial]\otimes  \hat{\g}^{\geq 0}_{<(d+1)m+1})$ and $ \mathcal{V}_{[d,m]}$ are isomorphic as differential algebras. Indeed, 
the differential algebra isomorphism
\begin{equation} \label{Eqn:7.2}
\widetilde{\iota}:S(\CC[\partial]\otimes  \hat{\g}^{\geq 0}_{<(d+1)m+1}) \simeq  \mathcal{V}_{[d,m]}
\end{equation} 
induced by $\iota$. Also, the Lie algebra $\hat{\g}\otimes \mathcal{V}_{[d,m]}$ is endowed with the bracket and the bilinear form   defined by 
\begin{equation} 
[a \otimes f, b\otimes g]= [a,b]\otimes fg \text{ and } (a\otimes f, b \otimes g) = (a,b)fg.
\end{equation}

Let $\{\ u_i\ |\ i\in \bar{S}\ \}$ and $\{\ \widetilde{u}_i\ |\ i \in \bar{S}\ \}$ be dual bases of $\g$ with respect to the bilinear form $(\cdot, \cdot)$ and assume that the basis elements $u_i$ and $\widetilde{u}_i$ are eigenvectors of $\ad \frac{h}{2}$.
We denote $
u_i^j := u_i z^j \in \hat{\g}$ and $\widetilde{u}_i^j := \widetilde{u}_i z^j \in \hat{\g}
$
for any $j \in \ZZ$. Then $\{\ u_i^j\ |\ i \in \bar{S}, \ j \in \ZZ\ \}$ and $\{\ \widetilde{u}_i^{-j}\ |\ i \in \bar{S}, \ j \in \ZZ\ \}$ are dual bases of $\hat{g}$. Also, the following two sets
\begin{equation}\label{Eqn:7.7}
\begin{aligned}
&\mathcal{B}:=\left\{\ u_i^j\  | \  i \in \bar{S}, \  j\in \ZZ\ \right\}\cap  \hat{\g}^{\geq 0}_{<(d+1)m+1} \text{ and } \mathcal{B}^-:=\left\{\ \widetilde{u}_i^{-j}\ | \ i \in \bar{S}, \  j\in \ZZ \ \right\}\cap \hat{\g}^{\leq 0}_{>-(d+1)m-1}
\end{aligned}
\end{equation}
are bases of $  \hat{\g}^{\geq 0}_{<(d+1)m+1} $ and $ \hat{\g}^{\leq 0}_{>-(d+1)m-1}$ and these two bases are dual with respect to $(\cdot, \cdot)$. 

The basis $\bar{\mathcal{B}}$ of $ \hat{\g}^{\geq 0}_{<(d+1)m+1}\oplus \hat{\g}$ is a extended basis of  $\mathcal{B}$ such that
\begin{equation}\label{Eqn:7.13_2}
\bar{\mathcal{B}}:=\mathcal{B}\cup \mathcal{B}_m , \text{ where } \mathcal{B}_m  \text{ is a basis of $\hat{\g}_{(d+1)m+1}$. }
\end{equation}
The index sets $\mathcal{I}$, $\mathcal{I}^-$, and $\bar{\mathcal{I}}$ are defined by
\begin{equation}\label{Eqn:7.14_2}
 \mathcal{I}=\{\ (i,j)\ |\ u_i^j \in \mathcal{B}\ \}, \  \   \mathcal{I}^-=\{\ (i,j)\ |\ u_i^j \in \mathcal{B}^-\ \}, \ \   \bar{\mathcal{I}}=\{\ (i,j)\ |\ u_i^j \in \bar{\mathcal{B}}\ \}.
\end{equation}

Using the isomorphism $\widetilde{\iota}$, we often write an element in $\mathcal{V}_{[d,m]}$ as a polynomial in $\CC[\ \partial^n\mathcal{B}\ |\ n \geq 0\ ],$ by an abuse of notation.

\subsection{ First definition of classical affine fractional $\WW$-algebras }\label{Subsec:7.1}
In this section, we review the construction of classical affine fractional $\WW$-algebras introduced in  \cite{BDHM,DHM}.
Let $\mathcal{F}$ be the algebra of complex valued smooth functions on $S^1$ and let $\{\ \partial_x^n u_i^j(x) \in \mathcal{F}\ |\ n \in \ZZ_{\geq 0} , \ i\in \bar{S}, j\in \ZZ\ \}$ be an algebraically independent set. We have a differential algebra homomorphism
$$ \mu: \CC_{\text{diff}}[\ \mathcal{B}\ ] \to \mathcal{F},$$
such that $\mu:u_i^{j (n)} \mapsto \frac{\partial_x^n u_i^j (x)}{\partial x^n}$ if $u_i^j \in \mathcal{B}$ and $\mu: u_i^{j(n)} \mapsto \delta_{n,0}\cdot (\Lambda_m, u_i^j)$ if $u_i^j \in \mathcal{B}_m.$ Then $\mathcal{V}_{[d,m]}$ is isomorphic to the image $\mu(\mathcal{V}_{[d,m]}).$ 
Hence functions in $\mu(\mathcal{V}_{[d,m]})$ can be identified with elements in $\mathcal{V}_{[d,m]}.$ 

When we need to clarify that a function $f\in \mathcal{V}_{[d,m]}$ depends on a variable $x\in S^1$,  we denote $f$ by $f(x)$. Especially, if there are more than one independent variables in a formula, we indicate the variable of each functional.  \\

\begin{defn}\label{Def:7.1}
For given
\begin{equation}\label{Eqn:7.9}
q_m \in \hat{\g}^{\leq 0}_{>-(d+1)m-1} \otimes \mathcal{V}_{[d,m]} \text{ and } k \in \CC, 
\end{equation} 
an operator of the form
\begin{equation}\label{Eqn:7.8}
L_m(x) =k\partial + q_m(x) +\Lambda_m \otimes 1 \in \CC \partial \ltimes \hat{\g} \otimes  \mathcal{V}_{[d,m]}
\end{equation}
is called the Lax operator associated to $q_m$.
\end{defn}

The differential  $\partial$ acts on $ \hat{\g} \otimes \mathcal{V}_{[d,m]}$ by
\begin{equation}\label{Eqn:7.10}
\partial (a \otimes f (x) ): = a \otimes  (\partial f(x)).
\end{equation}
Then one can check that  $\partial  [a \otimes f(x) , b \otimes g(x)] =[ a \otimes \partial  f(x) , b \otimes g(x) ]+ [a \otimes f(x)  + b \otimes \partial  g(x)].$
The Lax operator $L_m= k\partial+ \sum_{(i,j) \in \mathcal{I}} \widetilde{u}_i^{-j} \otimes q_m^{i,j}+ \Lambda_m \otimes 1$ linearly acts on $\hat{\g} \otimes  \mathcal{V}_{[d,m]} $ by the adjoint action:
\begin{equation*}
 [L_m(x),a \otimes f(x)]=ka \otimes \partial  f(x)  + \sum_{(i,j)  \in \mathcal{I}} [\widetilde{u}_i^{-j}, a ] \otimes  q_m^{i,j}(x)f(x)  + [\Lambda_m, a ] \otimes f(x) .
\end{equation*}
Recall that we have a bilinear form $(\cdot, \cdot)$ on $\hat{\g}\otimes \mathcal{V}_{[d,m]}$ defined by  $(a \otimes f, b\otimes g) =(a,b)fg$.

\begin{defn} \label{Def:7.2}
Let $a, b\in \hat{\g}$, $x,y,w \in S^1$ and  $g_i \in  \mathcal{F}$ for $i=1,2,3,4$. A bilinear form
$(\cdot, \cdot)_{w}$
is defined by
\begin{equation} \label{Eqn:7.11}
\begin{aligned}
& \left(  a \otimes F(x,w)  , \  b \otimes G(y,w) \right)_{w}:=  \int_{S^1} F(x,w)G(y,w)  \ dw,
\end{aligned}
\end{equation}
where $F(x,w)= g_1(x) g_2(w) \partial_x^n \delta(x-w)$, or $g_1(x)g_2(w)$ and $G(y,w)= g_3(y) g_4(w) \partial_y^m \delta(y-w)$ or $ g_3(y)g_4(w).$
\end{defn}

Using the bilinear form, a basis element $u_i^j\in \mathcal{B}$ can be understood as a functional on $  \hat{\g}^{\leq 0}_{>-(d+1)m-1}\otimes \mathcal{V}_{[d,m]}$, i.e. 
\begin{equation}\label{Eqn:7.12}
u_i^j  (a \otimes f(x)) :=(u_{i}^{j} \otimes 1 ,a\otimes f(x))= (u_{i}^{j} \otimes \delta(x-w) ,a\otimes f(w))_w, 
\end{equation}
 where $ a \in \hat{\g}^{\leq 0}_{>-(d+1)m-1} \text{ and }  f \in \mathcal{V}_{[d,m]}.$ Moreover, we let 
\begin{equation}\label{Eqn:7.22_2}
\begin{aligned}
(\partial u_i^j)(a \otimes f) = \partial (u_i^j(a \otimes f)), \ \ (g_1 g_2)(a \otimes f) =\left( g_1  (a \otimes f)\right) \cdot \left(  g_2(a \otimes f)\right), \ \  c  (a \otimes f)=c,
\end{aligned}
\end{equation}
for $g_1, g_2 \in \mathcal{V}_{[d,m]}$ and $c \in \CC$. Hence every element in $\mathcal{V}_{[d,m]}$ is a functional on $\hat{\g}\otimes \mathcal{V}_{[d,m]}$.

\begin{rem}
For $f,g \in \mathcal{V}_{[d,m]}$ and $a \in \hat{\g}$, the functional multiplication is defined by $ g\cdot (a \otimes f)= a \otimes gf.$
\end{rem}

\begin{lem}\label{Def:7.3}
For any $ Q_m \in \hat{\g}^{\leq 0}_{>-(d+1)m-1} \otimes \mathcal{V}_{[d,m]}$ and $q_m^{i,j}
\in \mathcal{V}_{[d,m]}$, the operator
$$ L_m(Q_m):= k\partial+ \sum_{(i,j) \in \mathcal{I}} \widetilde{u}_{i}^{-j} \otimes (q_m^{i,j} (Q_m)) + \Lambda_m \otimes 1$$
is a Lax operator.
\end{lem}

Lemma \ref{Def:7.3} follows from the fact that $q_m^{i,j}$ is a functional on  $ \hat{\g}^{\leq 0}_{>-(d+1)m-1} \otimes \mathcal{V}_{[d,m]}$.  Also, the operator 
\begin{equation}\label{Eqn:7.13}
\begin{aligned}
\mathcal{L}_m & := k\partial+ \sum_{(i,j) \in \mathcal{I}} \widetilde{u}_{i}^{-j} \otimes u_i^j + \Lambda_m \otimes 1= k\partial+ \sum_{(i,j)\in \bar{\mathcal{I}}} \widetilde{u}_{i}^{-j} \otimes u_i^j  \ \ \ \in \ \ \CC\partial \ltimes \hat{\g}\otimes \mathcal{V}_{[d,m]}
\end{aligned}
\end{equation}
is called the universal Lax operator. Then a Lax operator $L_m$ has the form
\begin{equation}\label{Eqn:7.14}
L_m=k\partial+q_m  +\Lambda_m \otimes 1 = \mathcal{L}_m(q_m).
\end{equation}
   
Now we introduce the gauge equivalence class on $\hat{g}\otimes \mathcal{V}_{[d,m]}$.

\begin{defn}\label{Def:7.4}
Let $S = a\otimes f \in \hat{\g} \otimes \mathcal{V}_{[d,m]} \text{ and } b\otimes g \in \hat{\g} \otimes \mathcal{V}_{[d,m]}$. The adjoint action of the Lie algebra  $\hat{\g} \otimes \mathcal{V}_{[d,m]}$ on  the space $ \CC\partial \ltimes \hat{\g} \otimes \mathcal{V}_{[d,m]} $ is defined as follows:
$$ (\ad S) (\partial) = -a \otimes\partial f, \text{ and }  (\ad S)(b \otimes g) = [a,b] \otimes fg.$$
\end{defn}

\begin{defn}\label{Def:7.5}
Let $ S\in \n \otimes  \mathcal{V}_{[d,m]}$ and  consider the map 
\begin{equation}\label{Eqn:7.15}
L_m=k\partial+q_m+\Lambda_m \otimes 1 \mapsto \widetilde{L_m}:=e^{\ad S}(L_m)=k\partial+\widetilde{q}_m+\Lambda_m \otimes 1
\end{equation}
between Lax operators. Then the  map $G_S: \hat{\g} \otimes \mathcal{V}_{[d,m]}\to  \hat{\g} \otimes \mathcal{V}_{[d,m]}$ such that $ G_S:q_m \mapsto \widetilde{q}_m$
is called the gauge transformation by $S$. Since gauge transformations define an equivalence relation, we  write  
$q_m \sim \widetilde{q}_m$ or $q_m \sim_S \widetilde{q}_m$
if $\widetilde{q}_m = G_S(q_m)$. 
\end{defn}

Any element in the differential algebra $\mathcal{V}_{[d,m]}$ is a map from $  \hat{\g}^{\leq 0}_{>-(d+1)m-1} \otimes \mathcal{V}_{[d,m]}$ to $\mathcal{V}_{[d,m]}$, as described in (\ref{Eqn:7.12}). Instead of  the whole algebra $\mathcal{V}_{[d,m]}$, we focus on the subset of $\mathcal{V}_{[d,m]}$ consisting of functionals which are well-defined on the gauge equivalence classes of $  \hat{\g}^{\leq 0}_{>-(d+1)m-1} \otimes \mathcal{V}_{[d,m]}$.

\begin{defn}\label{Def:7.6}
A functional $a:  \hat{\g}^{\leq 0}_{>-(d+1)m-1} \otimes \mathcal{V}_{[d,m]} \to \mathcal{V}_{[d,m]} $ is said to be gauge invariant if $a(q_m)=a(\widetilde{q}_m)$ whenever $q_m$ and $\widetilde{q}_m$ are gauge equivalent.
\end{defn}

\begin{defn}\label{Def:7.7}
The subset of $\mathcal{V}_{[d,m]}$ consisting of gauge invariant functionals on $ \hat{\g}^{\leq 0}_{>-(d+1)m-1} \otimes \mathcal{V}_{[d,m]}$ 
is called the $m$-th affine fractional $\WW$-algebra associated to $\g$, $\Lambda_m$ and $k$ and we write this algebra as
\begin{equation}\label{Eqn:7.16}
\WW_1(\g,\Lambda_m,k).
\end{equation}
\end{defn}

Affine fractional $\WW$-algebras are well-defined differential algebras since if $\phi$ and $\psi$ are gauge invariant then $\phi+ \psi$, $\phi \psi$ and $\partial \phi$ are also gauge invariant. In the rest of this section, we introduce two local Poisson brackets on each affine fractional $\WW$-algebra.

\begin{rem}
In Section \ref{Subsec:7.2}, we introduce another definition of affine fractional $\WW$-algebras. The subscript $1$ in (\ref{Eqn:7.16}) is used until we prove the equivalence of two definitions.
\end{rem}

\begin{defn}\label{Def:7.8}
Let $\phi, \psi \in \WW_1(\g, \Lambda_m, k)$. Then the two local brackets $\{ \cdot ,\cdot \}_1$ and $\{\cdot, \cdot \}_2$ on $\WW_1(\g, \Lambda_m, k)$ are defined by
\begin{equation}\label{Eqn:7.17}
 \begin{aligned}
 & \{\phi(x), \psi(y)\}_1  \\
& = -\left(\sum_{\begin{subarray}{l} (i,j) \in \mathcal{I},\\ 
\ \ n\geq 0 \end{subarray}} u_i^j \otimes  \frac{\partial \phi(x)}{\partial u_i^{j(n)}}\partial_x^n \delta(x-w), \left[\sum_{\begin{subarray}{l}(p,q) \in \mathcal{I} ,\\ \ \ l\geq 0\end{subarray}} u_p^{q+1}  \otimes  \frac{\partial\psi(y)}{\partial u_p^{q(l)}}\partial_y^l \delta(y-w) , \mathcal{L}_m(w)\right] \right)_w, 
\end{aligned}
\end{equation}

\begin{equation}\label{Eqn:7.18}
\begin{aligned}
& \{\phi(x), \psi(y)\}_2  \\   
&= \left(\sum_{\begin{subarray}{l} (i,0) \in \mathcal{I}\\
   \ \  n\geq 0\end{subarray}}     u_i^0 \otimes  \frac{\partial \phi(x)}{\partial u_i^{0(n)}}\partial_x^n \delta(x-w), \left[\sum_{\begin{subarray}{l}(p,0) \in \mathcal{I},\\ \ \  l\geq 0\end{subarray}} u_p^{0(l)}  \otimes  \frac{\partial\psi(y)}{\partial u_p^{0(l)}}\partial_y^l \delta(y-w) , \mathcal{L}_m(w)\right] \right)_w \\
&  -\left(\sum_{\begin{subarray}{l}(i,j)\in \mathcal{I} ,j>0,\\ \ \ n\geq 0\end{subarray}} u_i^j \otimes  \frac{\partial \phi(x)}{\partial u_i^{j(n)}}\partial_x^n \delta(x-w), \left[\sum_{\begin{subarray}(p,q)\in \mathcal{I},\\ q>0, l \geq 0\end{subarray}} u_p^{q}  \otimes  \frac{\partial\psi(y)}{\partial u_p^{q(l)}}\partial_y^l \delta(y-w) , \mathcal{L}_m(w)\right] \right)_w.
\end{aligned}
\end{equation}
\end{defn}

(Later, in Proposition \ref{Prop:7.11}, we show that $\{\cdot, \cdot \}_1$ and $\{\cdot, \cdot \}_2$ are well-defined Poisson local brackets.)

\begin{rem} \label{Note:7.10}

The local Poisson brackets $\{ \phi(x), \psi(y) \}_i$, $i=1,2$, have the form $\newline$ $\sum_{ n \geq0} \Phi_{i,n}(y) \partial_y^n \delta(x-y)$, for some $\Phi_{i, n} \in \mathcal{V}_{[d,m]}.$ By definition of  $\delta$-function, we have $\int \Phi_{i,n}(y) \partial_y^n \delta(x-y)  dx =(-\partial)^n \Phi_{i,n}(y)$. Hence Poisson brackets on $\mathcal{V}_{[d,m]}/\partial \mathcal{V}_{[d,m]} $ can be defined by
$$\int  \left\{ \phi,  \psi\right\}_i  : = \int \int\{\phi(x), \psi(y)\}_i dxdy.$$
Two brackets $\int \{ \cdot , \cdot \}_i$, $i=1,2$, are well-defined Poisson brackets on $\WW_1(\g, \Lambda_m, k)/\partial \WW_1(\g, \Lambda_m, k) $ if the local brackets $\{\cdot, \cdot \}_i$ are well-defined on $\WW_1(\g,\Lambda_m, k)$. Indeed, the skew-commuativity, Jacobi identity and Leibniz rules of $\{\cdot, \cdot\}_i$ gaurantee those of $\int \{ \cdot, \cdot \}_i$. 
\end{rem}

\begin{lem}\label{Prop:7.9}
The derivative of $\phi$ at $q_m(x)$ has the following property:
\begin{equation}\label{Eqn:7.19}
\frac{d}{d\epsilon} \phi(q_m(x)+\epsilon r(x) )|_{\epsilon=0}=\sum_{\begin{subarray}{l}(i,j) \in \mathcal{I},\\ \ \  n \geq 0\end{subarray}} \frac{\partial \phi(q_m(x))}{ \partial {u}_i^{j \ (n)}}  \ \ (  \ u_i^{j} \otimes \partial_x^n \delta(x-w) \ , r(w) \ )_w .
\end{equation}
\end{lem}

\begin{proof}
By Taylor expansion, we can write $\phi(q_m(x)+\epsilon r(x))$ as a power series in $\epsilon$. Substituting $\epsilon=0$ to the series, we obtain the following formula:
\begin{align*}
\frac{d}{d\epsilon} \phi(q_m(x)+ \epsilon r(x) ) |_{\epsilon=0} 
&= \left. \frac{d}{d\epsilon} \right|_{\epsilon=0} \left(  \phi(q_m(x)) +\epsilon \sum_{\begin{subarray}{l}(i,j) \in \mathcal{I}, \\ \ \ n \in \ZZ_{\geq 0}\end{subarray}} \frac{\partial \phi}{\partial u_i^{j\ (n)}} \partial_x^n (u_i^j\otimes 1 , r(x)) +o(\epsilon^2) \right) \\
& =\sum_{(i,j)\in \mathcal{I}, n \geq 0} \frac{\partial \phi(q_m(x))}{ \partial {u}_i^{j \ (n)}}  \ \ (  \ u_i^{j} \otimes \partial_x^n \delta(x-w) \ , r(w) \ )_w .
\end{align*}
\end{proof}

\begin{lem}\label{Prop:7.10}
Let $L_m(x) = \partial+ q_m(x) + \Lambda_m \otimes 1$ and let $e^{\ad S(x)} L_m(x) = \widetilde{L}_m(x) = \partial+ \widetilde{q}_m(x) + \Lambda_m \otimes 1.$ Then, for $\phi, \psi \in \WW_1(\g, \Lambda_m, k)$, we have 
\begin{equation}\label{Eqn:4.26_011014}
\sum_{\begin{subarray}{l} (i,j) \in \mathcal{I},\\ \ \  n \geq 0\end{subarray}} u_i^j\otimes \frac{\partial \phi(\widetilde{q}_m(x))}{\partial u_i^{ j (n)}} \partial_x^n \delta(x-w) = \sum_{\begin{subarray}{l}(i,j) \in \mathcal{I},\\ \ \  n \geq 0\end{subarray}} \frac{\partial \phi(q_m(x))}{\partial u_i^{ j (n)}} e^{\ad S(w)} \left(u_i^j \otimes  \partial_x^n \delta(x-w)\right).
\end{equation}
\end{lem}

\begin{proof}
By Lemma \ref{Prop:7.9}, we have two equations:
\begin{equation*}
\left. \frac{d}{d \epsilon} \phi(\widetilde{q}_m (x) + \epsilon r (x) ) \right|_{\epsilon=0}=\sum_{\begin{subarray}{l}(i,j) \in \mathcal{I},\\ \ \  n \geq 0\end{subarray}}\frac{\partial \phi(\widetilde{q}_m(x))}{\partial u_i^{ j (n)}} \left(  \ u_i^j \otimes \partial_x^n \delta(x-w) \ ,\  r(w)\ \right)_w;
\end{equation*}
\begin{equation*}
  \left.\frac{d}{d\epsilon}\phi(q_m+ e^{-\ad S} \epsilon r)\right|_{\epsilon=0} =  \sum_{\begin{subarray}{l}(i,j) \in \mathcal{I},\\ \ \  n \geq 0\end{subarray}}\frac{\partial \phi(q_m(x))}{\partial u_i^{ j (n)}} \left(  e^{\ad S(w)} \left(u_i^j \otimes \partial_x^n \delta(x-w)\right) \  , \ r(w)\right)_w.
\end{equation*}
Since $q_m+ e^{-\ad S}\epsilon r  \sim_S \widetilde{q}_m + \epsilon r$, we obtain (\ref{Eqn:4.26_011014}).
\end{proof}

To simplify notations, let us denote
\begin{align*}
& d_{q_m(x)} \phi(w) := \sum_{(i,j)\in \mathcal{I},n\geq 0}  u_i^j \otimes \frac{\partial \phi(q_m(x))}{\partial u_i^{ j (n)}} \partial_x^n \delta(x-w),\\
& d_{q_m(x)} \phi(w)^0 := \sum_{(i,0) \in \mathcal{I} ,n\geq 0}  u_i^0 \otimes \frac{\partial \phi(q_m(x))}{\partial u_i^{ 0 (n)}} \partial_x^n \delta(x-w),\\
& d_{q_m(x)} \phi(w)^> := \sum_{(i,j)\in \mathcal{I } ,j>0,n\geq 0}  u_i^j \otimes \frac{\partial \phi(q_m(x))}{\partial u_i^{ j (n)}} \partial_x^n \delta(x-w).
\end{align*}
Then the two brackets $\{\cdot, \cdot\}_1$ and $\{\cdot, \cdot\}_2$ are 
\begin{align*}
 &\{ \phi(x), \psi(y)\}_1  =- (d_{q_m(x)} \phi(w), [zd_{q_m(y)} \psi(w), \mathcal{L}_m (w)])_w; \\
 & \{ \phi(x), \psi(y)\}_2  = (d_{q_m(x)} \phi(w)^0, [d_{q_m(y)} \psi(w)^0, \mathcal{L}_m (w)])_w - (d_{q_m(x)} \phi(w)^>, [d_{q_m(y)} \psi(w)^>, \mathcal{L}_m (w)])_w . 
\end{align*}

\begin{prop}\label{Prop:7.11}
Two local  Poisson brackets  $\{ \cdot ,\cdot \}_1$ and $\{ \cdot , \cdot \}_2$ are well-defined on $\WW_1(\g, \Lambda_m, k)$. 
\end{prop}

\begin{proof}
Recall that  $q_m \sim_S \widetilde{q}_m$ if and only if 
$ e^{\ad S} \mathcal{L}_m (q_m) = \mathcal{L}_m (\widetilde{q}_m)$
and, by Lemma \ref{Prop:7.10},  $d_{\widetilde{q}_m(x)} \phi (w) = e^{\ad S(w)} d_{q_m(x)}\phi(w)$.
Since the bilinear form  $(\cdot, \cdot)_w$ is invariant under the adjoint action, we have 
\begin{align*}
\{ \phi(x), \psi(y) \}_1(\widetilde{q}_m) & = - (e^{\ad S(w)}d_{q_m(x)} \phi(w), [e^{\ad S(w)} zd_{q_m(y)} \psi(w),e^{\ad S(w)} \mathcal{L}_m (w)])_w \\
& = - (e^{\ad S(w)}d_{q_m(x)} \phi(w), e^{\ad S(w)}[ zd_{q_m(y)} \psi(w), \mathcal{L}_m (w)])_w \\
& =-( d_{q_m(x)} \phi(w), [ zd_{q_m(y)} \psi(w), \mathcal{L}_m (w)])_w= \{ \phi(x), \psi(y) \}_1(q_m).
\end{align*}
Since $S \in \hat{\g}^0 \otimes \mathcal{V}_{[d,m]}$, the same procedure works for the second bracket, i.e.
$$\{ \phi(x), \psi(y) \}_2(\widetilde{q}_m)= \{ \phi(x), \psi(y) \}_2(q_m).$$
Hence $\{ \cdot, \cdot \}_1$ and $\{\cdot, \cdot\}_2$ are well-defined on $\WW_1(\g, \Lambda_m, k)$.\\
To see the skew-symmetry, Jacobi identity and Leibniz rules, let us rewrite formulas (\ref{Eqn:7.17}) and (\ref{Eqn:7.18}). The first bracket is as follows:
\begin{align*}
& \{\phi(x), \psi(y)\}_1  \\
&   =- \sum_{i,j,n} \sum_{p,q,l}  \frac{\partial \phi(x)}{\partial u_i^{j(n)}}\partial_x^n \frac{\partial\psi(y)}{\partial u_p^{q(l)}} \partial_y^l  [u_i^j, u_p^{p+1}](y)  \delta(x-y).
\end{align*}
Similarly the second bracket of $\phi(x)$ and $\psi(y)$ can be written as follows:
\begin{align*}
 & \{\phi(x), \psi(y)\}_2    \\
 & =  \sum_{u_i^0 \in \g,n\geq 0} \ \sum_{u_p^0 \in \g,l\geq 0}  \frac{\partial \phi(x)}{\partial u_i^{0(n)}}\partial_x^n \frac{\partial\psi(y)}{\partial u_p^{0(l)}}\partial_y^l\left( k\partial_y   (u_i^0, u_p^0)    + [u_i^0, u_p^0](y)  \right) \delta(x-y) \\
 & - \sum_{u_i^j \in S,j>0,n\geq 0}\  \sum_{u_p^q \in S,q>0,l\geq 0}  \frac{\partial \phi(x)}{\partial u_i^{j(n)}}\partial_x^n \frac{\partial\psi(y)}{\partial u_p^{q(l)}}\partial_y^l [u_i^j, u_p^{p}](y)  \delta(x-y).
 \end{align*}
Then Leibniz rules hold obviously. Skew-symmetries and Jacobi identities of two brackets follow from those of $\hat{\g}$ and the properties of $\delta$-function.
\end{proof}

 The following theorem follows from Proposition \ref{Prop:7.11} and Remark \ref{Note:7.10}.

\begin{thm}
\begin{enumerate}[(i)]
\item
The two brackets $\{\cdot, \cdot\}_i$, $i=1,2$, are well-defined local Poisson brackets on the affine fractional $\WW$-algebra $\WW(\g, \Lambda_m, k).$ 
\item
The two brackets  $\int \{\cdot, \cdot\}_i$ on  the algebra  $\WW(\g, \Lambda_m, k)/\partial \WW(\g, \Lambda_m, k)$, which are induced from  $\{\cdot, \cdot\}_i$, are well-defined Poisson brackets.
\end{enumerate}
\end{thm}

\subsection{Second definition of classical affine fractional $\WW$-algebras}\label{Subsec:7.2}
In this section, we construct a PVA called a classical affine fractional $\WW$-algebra which is equivalent to the fractional $\WW$-algebra explained in the previous section. 

Let $\lambda$-adjoint action $\ad_\lambda \n$ on $\mathcal{V}_{[d,m]}$ be defined as follows:
\begin{align*}
&(\ad_\lambda n) (a z^i)  = \{n_\lambda a z^i \}, \ \text{ where } \{n _\lambda a z^i\} = [n,a]+ \delta_{i,0} k\lambda(n,a),\\
&(\ad_\lambda n) (\partial A)  =(\partial+\lambda) (\ad_\lambda n) (A), \ (\ad_\lambda n) (AB)  =B(\ad_\lambda n) (A) + A (\ad_\lambda n) (B),
\end{align*}
for $a \in\g, \ n \in \n,$ and $ A, B \in \mathcal{V}_{[d,m]}.$ 

\begin{defn}
The $m$-th affine fractional $\WW$-algebra associated to $\g$, $\Lambda_m$ and $k$ is an associative differential algebra 
\begin{equation}\label{Eqn:7.20}
\WW_2(\g,\Lambda_m,k):={ \mathcal{V}_{[d,m]}}^{\ad_\lambda \n}
\end{equation}
endowed with the product $(a+I)(b+I)=ab+I,$ where $I$ is the differential algebra ideal of $S(\CC[\partial]\otimes \hat{\g}^+_{[d,m]})$ such that $\mathcal{V}_{[d,m]}=S(\CC[\partial]\otimes \hat{\g}^+_{[d,m]})/I.$
\end{defn}

\begin{prop}\label{Prop:7.13}
The associative differential algebra $\WW_2(\g,\Lambda_m,k)$ is well-defined.
\end{prop}
\begin{proof}
We need to show that $(\ad_\lambda \n) (I) \subset I[\lambda]$. Let  us pick $n \in \n$, $A \in S(\CC[\partial] \otimes \hat{\g}_{(m)})$ and $B =  a-\chi(a)$, $a \in \hat{\g}_{(d+1)m+1}$. Then, by Leibniz rule,
$$ \{n_\lambda AB\}= B (\ad_\lambda n) (A) + A (\ad_\lambda n)(B).$$
 The first term of the RHS is clearly in $I$ and the second term is also  in $I$ since $(\ad_\lambda n) (B)=0$. Hence $\ad_\lambda \n$-action is well-defined on $S(\CC[\partial]\otimes \hat{\g}_{(m)})/I$. Also, since $\lambda$-adjoint action satisfies Leibniz rule, the product on  $\WW_2(\g, \Lambda_m, k)$ is well-defined and, since $(\ad_\lambda n) (\partial A)= (\partial+\lambda)( (\ad_\lambda n) (A))$,   we conclude that $\WW_2(\g, \Lambda_m, k)$ is a well-defined differential associative algebra. 
\end{proof}

We want to show that that $\WW_1(\g, \Lambda_m, k) $ and $ \WW_2(\g, \Lambda_m, k)$  are isomorphic as differential associative  algebras.


\begin{prop}\label{Prop:7.14}
Let $Q_m= \sum_{(i,j) \in \bar{\mathcal{I}}}  \widetilde{u}_i^{-j} \otimes  u_i^{j} \in \hat\g^-_{({m})}\otimes \mathcal{V}_{[d,m]}
$. Then the universal Lax operator $\mathcal{L}_m = k\partial+ Q_m.$ Also, for a given $S \in \n \otimes \mathcal{V}_{[d,m]}$, we denote  $$\mathcal{L}_m(\epsilon)= k\partial + Q_m(\epsilon)= e^{\ad \ \epsilon S} \mathcal{L}_m.$$
If $S=a \otimes r$, then the derivative of $\phi \in \mathcal{V}_{[d,m]}$ has the following property:
\begin{equation}\label{Eqn:7.23}
\frac{d}{d\epsilon} \phi(Q_m (\epsilon))|_{\epsilon=0} = -\{a_\partial \phi \}_{k \to}r.
\end{equation}
\end{prop}
\begin{proof}
Since $ \mathcal{L}_m(\epsilon) = \mathcal{L}_m+ \ad\epsilon S (\mathcal{L}_m)+o(\epsilon^2),$ we can write
$Q_m(\epsilon)$ as a power series of $\epsilon$:
$$Q_m(\epsilon)= Q_m + \ad\epsilon S ( \mathcal{L}_m) + o(\epsilon^2).$$
By Taylor expansion, we have
\begin{equation}\label{Eqn:0328_6.30}
\phi(Q_m(\epsilon))= \phi(Q_m)+ \epsilon \sum_{\begin{subarray}{l}(i,j) \in \mathcal{I}, n \geq 0 \\ \ \ (\alpha,\beta) \in \mathcal{I}\end{subarray}}\frac{\partial \phi(Q_m)}{\partial u_i^{j \ (n)}} \partial^n (-k a \otimes \partial r +[a \otimes r, \widetilde{u}_\alpha^{-\beta} \otimes u_\alpha^{\beta}], u_i^{j} \otimes \delta(x-w))_w + o(\epsilon^2).
\end{equation}
 The second term of the RHS can be rewritten as
\begin{align*}
& ( [\ a \otimes r, k \partial + \sum_{(\alpha, \beta)\in\mathcal{I}}  \widetilde{u}_\alpha^{-\beta} \otimes u_\alpha^{\beta}\ ]\  , u_i ^{j}\otimes \delta(x-w))_w\\ 
&= (-k a\otimes \partial r + \sum_{(\alpha, \beta)\in \mathcal{I}} [a, \widetilde{u}_{\alpha}^{-\beta}] \otimes r u_\alpha^{\beta}\ ,\ u_i^{j} \otimes \delta(x-w))_w   =-[a,u_i^{j}]r-k(a, u_i^{j}) \partial r.
\end{align*}
We proved (\ref{Eqn:7.23}), since 
$$ \{a _\partial \phi \}_{k \to} r = \sum_{\begin{subarray}{l}(i, j)\in \mathcal{I}, \\ \ \ n\geq0 \end{subarray}} \frac{\partial \phi}{\partial u_i^{j  (n) }}\partial^n \{a_\partial u_i^{j} \}_{\to} r = \sum_{\begin{subarray}{l}(i, j)\in \mathcal{I},\\ \ \ n\geq 0 \end{subarray}}\frac{\partial \phi}{\partial u_i^{j  (n) }} \partial^n ( [a,u_i^{j}]r+ k(a, u_i^{j}) \partial r).$$
\end{proof}

\begin{cor}\label{Cor:7.15}
Two differential algebras
$\WW_1(\g, \Lambda_m, k)$ and  $\WW_2(\g, \Lambda_m, k)$ are isomorphic. Hence we denote the $m$-th affine fractional $\WW$-algebra associated to $\g$, $\Lambda_m$ and $k$ by $\WW(\g, \Lambda_m, k).$
\end{cor}
\begin{proof}
Let $\ Q_m = \sum_{(i,j) \in \bar{\mathcal{I}}} \widetilde{u}_i^{-j} \otimes u_i^j\ $ and let $\ \mathcal{L}_m (\epsilon) = k \partial + Q_m(\epsilon)= e^{\ad \epsilon S} \mathcal{L}_m,\ $ where $\ S= a\otimes f \in \n \otimes \mathcal{V}_{[d,m]}.\ $
By the definition of $\WW_1(\g,\Lambda_m, k)$,  if a functional $\phi$ is in $ \WW_1(\g, \Lambda, k)$ then $\ \frac{d}{d\epsilon}  \phi(Q_m(\epsilon)) |_{\epsilon=0}=0$.
Notice that $$ \left[ \{ n _\lambda\phi \}=0 \text{ for any } n  \in \n\right] \text{   \  if and only if \  }\left[\{n_\partial \phi\}_{\to} r=0 \text{ for any }  a \otimes r  \in \n \otimes  \mathcal{V}_{[d,m]}\right].$$
 Thus, by Proposition \ref{Prop:7.14}, 
 if $\ \frac{d}{d\epsilon}  \phi(Q_m(\epsilon)) |_{\epsilon=0}=0$, then a functional $\phi$ is in $\WW_2(\g, \Lambda_m, k)$. Hence $\WW_1(\g, \Lambda,k) \subset \WW_2(\g, \Lambda,k).$\\
Conversely, assume that $\phi \in \WW_2(\g, \Lambda_m, k)$. Then, by Proposition \ref{Prop:7.14}, $ \left. \frac{d}{d\epsilon}\phi(Q_m(\epsilon)) \right|_{\epsilon=0}=0$.  Let 
 $L_m:= k\partial+ q_m = \mathcal{L}(q_m)= k\partial + Q_m(q_m) $ be a  Lax operator and  let $ L_m(\epsilon) :=e^{\ad\ \epsilon S}L_m=k \partial + q_m(\epsilon) $.  If $S= a \otimes r \in \n \otimes \mathcal{V}_{[d,m]}$, then  
 \begin{equation}
\left.\frac{d \phi(q_m(\epsilon))}{d\epsilon}\right|_{\epsilon=0} = \left.( [\ a \otimes 1, k \partial + \sum_{\alpha, \beta}  \widetilde{u}_\alpha^{-\beta} \otimes u_\alpha^{\beta}\ ]\  , u_i ^{j}\otimes \delta(x-w))_w \right|_{u_i^j = (\widetilde{u}_i^{-j}, q_m)} \cdot r.
 \end{equation}
 Hence if 
 $$\left.\frac{d \phi(Q_m(\epsilon))}{d\epsilon}\right|_{\epsilon=0} = ( [\ a \otimes 1, k \partial + \sum_{\alpha, \beta}  \widetilde{u}_\alpha^{-\beta} \otimes u_\alpha^{\beta}\ ]\  , u_i ^{j}\otimes \delta(x-w))_w r=0$$
for any  $a\otimes r \in \n \otimes \mathcal{V}_{[d,m]}$, then $\left.\frac{d \phi(q_m(\epsilon))}{d\epsilon}\right|_{\epsilon=0} =0$ for any $a\otimes r \in \n \otimes \mathcal{V}_{[d,m]}$ and the functional $\phi$ is gauge invariant, i.e. $\phi \in \WW_1(\g, \Lambda_m, k).$ 
\end{proof}

Let us define Poisson $\lambda$-brackets on $\WW(\g, \Lambda_m, k)$ using the local brackets $\{ \cdot, \cdot\}_i$, $i=1,2$.
Suppose the local Poisson brackets are
\begin{equation}\label{Eqn:7.24}
 \{ \phi(x), \psi(y)\}_i=: \sum_{j\geq 0} \frac{1}{j!}(\phi_{(j)} \psi)_i(y)\partial_y^j \delta(x-y), \  \phi, \psi \in \WW(\g, \Lambda_m, k), \ i=1, 2.
\end{equation}
Then we can define corresponding Poisson $\lambda$-brackets by 
\begin{equation}\label{Eqn:7.25}
\begin{aligned}
\{\phi_\lambda \psi \}_i (y) :=  \int e^{\lambda(x-y)} \{ \phi(x), \psi(y)\}_i dx  =\sum_{j\geq 0} \frac{\lambda^j}{j!} (\phi_{(j)} \psi)_i(y)\ , \ i=1,2.
\end{aligned}
\end{equation}
The algebra $\WW(\g, \Lambda_m, k)$ is a well-defined PVA endowed with the $\lambda$-brackets $\{\cdot_\lambda \cdot\}_i$, $i=1,2$. (see Proposition \ref{Prop:7.16})

\begin{lem}\label{Lem:7.16}
For any $\phi, \psi \in \WW(\g, \Lambda_m, k)$, we have
\begin{equation}\label{Eqn:7.26}
\{\partial_x \phi(x), \psi(y)\}_i =\partial_x \{ \phi(x), \psi(y)\}_i \text{ and } \{ \phi(x), \partial_y\psi(y)\}_i =\partial_y \{ \phi(x), \psi(y)\}_i.
\end{equation}
\end{lem}
\begin{proof}
By the product rule of derivatives, we obtain 
\begin{equation}\label{Eqn:prod}
\begin{aligned} 
\sum_{n \geq 0}\frac{\partial }{\partial u_i^{j(n)}} (\partial_x\phi(x)) \partial_x^n \delta(x-w) &  =\sum_{n \geq 0} \left[\left(\partial_x\frac{\partial}{\partial u_i^{j(n)}} \phi(x) \right) \partial_x^n \delta(x-w)  + \frac{ \partial }{\partial u_i^{j(n-1)}} \phi(x) \partial_x^n \delta(x-w) \right]\\
& = \sum_{n \geq 0} \partial_x \left( \frac{\partial}{\partial u_i^{j(n)}} \phi \partial_x^n\delta(x-w) \right).
\end{aligned}
\end{equation}
Applying equation (\ref{Eqn:prod}) to the brackets (\ref{Eqn:7.17}) and (\ref{Eqn:7.18}), we get the equation (\ref{Eqn:7.26}). 
\end{proof}

\begin{prop} \label{Prop:7.16}
Two $\lambda$-brackets (\ref{Eqn:7.25}) are well-defined on $\WW(\g, \Lambda_m, k).$
\end{prop}
\begin{proof}
For any $\phi, \psi \in \WW(\g,\Lambda_m, k)$,
we already proved that $\{\phi(x), \psi(y)\}_i $ is gauge invariant. Hence by (\ref{Eqn:7.25}),  we have $\{\phi_\lambda \psi\}_i \in \WW(\g, \Lambda_m, k)[\lambda]$. 

To prove sesquilinearities, we need Lemma \ref{Lem:7.16}. The local bracket between $\partial \phi$ and $\psi$ is:
\begin{equation}
\{\partial \phi(x), \psi(y)\}_i= \sum_{j\geq 0} \partial_x \frac{1}{j!}(\phi_{(j)} \psi)(y)\partial_y^j \delta(x-y) =\sum_{j \geq 0}  -\frac{1}{j!}(\phi_{(j)} \psi)(y)\partial_y^{j+1} \delta(x-y)
\end{equation} 
Also, we have the bracket between $\phi$ and $\partial \psi$:
\begin{equation}
\begin{aligned}
\{\phi(x), \partial \psi(y)\}_i  =& \sum_{j\geq 0} \partial_y \frac{1}{j!}(\phi_{(j)} \psi)(y)\partial_y^j \delta(x-y)  \\
=& \sum_{j \geq 0} \frac{1}{j!}\partial_y((\phi_{(j)} \psi)(y))\partial_y^j \delta(x-y)+\frac{1}{j!}(\phi_{(j)} \psi)(y)\partial_y^{j+1} \delta(x-y).
\end{aligned}
\end{equation}
Then the sesquilinearities 
\begin{align*}
\{ \partial \phi_\lambda \psi \}_i   = -\lambda \{ \phi_\lambda \psi\}_i, 
\{ \phi_\lambda \partial \psi \} _i = (\partial+ \lambda) \{ \phi_\lambda \psi\}_i
\end{align*}
follow from the equation (\ref{Eqn:7.25}).

Leibniz rules, skew-symmetries and Jacobi identities of Poisson $\lambda$-brackets follow from those of local Poisson brackets. Hence $\WW(\g, \Lambda_m, k)$ endowed with the Poisson $\lambda$-bracket $\{\cdot_\lambda\cdot \}_i$ is a well-defined PVA. 
\end{proof}

Recall that
\begin{equation}\label{Eqn:7.27}
\{ \phi(x), \psi(y) \}_1 =  - \sum_{\begin{subarray}{l}(i,j)\in \mathcal{I}\\ \ \  n\geq0\end{subarray}}\sum_{\begin{subarray}{l}(\alpha, \beta) \in \mathcal{I},\\ \ \  l\geq0\end{subarray}} \frac{\partial \phi}{\partial u_i^{j(n)}}(x)\partial_x^n\frac{\partial \psi}{\partial u_\alpha^{\beta(l)}}(y)\partial_y^l  [ u _i^{j}, u_\alpha^{\beta+1}](y) \delta(x-y)
\end{equation}
and
\begin{equation}\label{Eqn:7.28}
\begin{aligned}
 \{\phi(x), \psi(y)\}_2 &   
= \sum_{\begin{subarray}{l}(i,0)\in \mathcal{I}, \\ \ \  n\geq 0\end{subarray}}\sum_{\begin{subarray}{l}(\alpha,0)\in\mathcal{I} ,\\ \ \ l\geq 0\end{subarray}}   \frac{\partial \phi}{\partial u_i^{0(n)}}(x)\partial_x^n\frac{\partial \psi}{\partial u_\alpha^{0(l)}}(y)\partial_y^l\left( k (u_i^0, u_\alpha^0) \partial_y +[u_i^0, u_\alpha^0]  (y)\right)\delta(x-y) \\
&-  \sum_{\begin{subarray}{l}(i, j)\in \mathcal{I}, \\  j>0,n\geq 0\end{subarray}} \sum_{\begin{subarray}{l}(\alpha, \beta) \in \mathcal{I}, \\ \beta>0,l\geq 0\end{subarray}} \frac{\partial \phi}{\partial u_i^{j(n)}}(x)\partial_x^n\frac{\partial \psi}{\partial u_\alpha^{\beta(l)}}(y)\partial_y^l[u_i^{j}, u_\alpha^{\beta} ](y)  \delta(x-y).
\end{aligned}
\end{equation}
for $\phi, \psi \in \mathcal{V}_{[d,m]}$. These two local Poisson brackets are also well-defined on  $\WW_1(\g, \Lambda,k)$. The brackets between basis elements in $\mathcal{B}$ are
\begin{equation} \label{Eqn:7.29}
 \{u_i^j (x), u_p^q(y)\}_1= -z[u_i^j, u_p^q](y)  \delta(x-y)
\end{equation}
and
\begin{equation}\label{Eqn:7.30}
\{u_i^j (x), u_p^q(y)\}_2=
\left\{\begin{array}{ll}
k(u_i, u_p) \partial_y \delta(x-y) + [u_i,u_p](y)\delta(x-y) & \text{if }  j=q=0,\\
-[u_i^j,u_p^q](y)\delta(x-y) & \text{if }  j\ne 0, q\ne0\\
0 & \text{otherwise}. \\
\end{array}
\right.
\end{equation}
Hence the Poisson $\lambda$-brackets on $\mathcal{V}_{[d,m]}$  and $\WW(\g, \Lambda_m, k)$ are as follows:
\begin{defn}\label{Def:7.17}
The differential algebra $\mathcal{V}_{[d,m]}$ has two Poisson $\lambda$-brackets:
\begin{equation}\label{Eqn:7.31}
\{a z^i_\lambda b z^j\}_1 =-[a,b] z^{i+j+1},
\end{equation}
and
\begin{equation}\label{Eqn:7.32}
\{a z^i_\lambda b z^j\}_2= \left\{\begin{array}{ll}
k(a, b) \lambda+ [a,b] & \text{if }  i=j=0,\\
0 & \text{if } i=0, j\ne 0, \text{ or } i\ne 0, j=0, \\
-[a z^i, b z^j] & \text{if }  i\ne 0, j\ne0,
\end{array}
\right.
\end{equation}
where $a z^i$ and $b z^j$ are basis elements in $\mathcal{B}$. Also, two brackets (\ref{Eqn:7.31}) and (\ref{Eqn:7.32}) induce well-defined Poisson $\lambda$-brackets on $\WW(\g, \Lambda_m, k)\subset \mathcal{V}_{[d,m]}.$
\end{defn}

Furthermore, the following proposition shows that the two Poisson $\lambda$-brackets are compatible, which means that $\{\cdot_\lambda \cdot\}_\alpha:= \{\cdot_\lambda \cdot\}_1 + \alpha \{\cdot_\lambda \cdot \}_2$ is a Poisson $\lambda$-bracket for any $\alpha \in \CC.$

\begin{prop}\label{Prop:7.18}
The Poisson $\lambda$-brackets $\{\cdot _\lambda \cdot \}_1$ and $\{\cdot _\lambda \cdot \}_2$ are compatible.
\end{prop}
\begin{proof}

Let us define the $\lambda$ brackets with the parameter $\alpha$ by:
$$\{\cdot _\lambda \cdot \}_\alpha = \alpha \{\cdot _\lambda \cdot \}_1 + \{ \cdot _\lambda \cdot\}_2.$$
The skewsymmetry, sesquilinearities and Leibniz rules of $\{\cdot _\lambda \cdot \}_\alpha$ are trivial by those of $\{\cdot_\lambda\cdot\}_1$ and $\{\cdot_\lambda\cdot\}_2.$ Thus the only thing to check is Jacobi identity. Due to the sesquilinearities and Leibniz rules, it suffices to show that:
\begin{equation}\label{Eqn:7.33}
\{a z^i \ _\lambda\ \{b z^j \ _\mu\ c z^k \}_\alpha \}_\alpha - \{b z^j \ _\mu\ \{a z^i \ _\lambda\ c z^k \}_\alpha \}_\alpha
=\{\{ a z^i \ _\lambda\ b z^j \}_{\alpha \ \lambda+\mu\ } c z^k\}_\alpha,
\end{equation}
when $a,b,c \in \g$ and $i,j,k\in \ZZ_{\geq 0}.$
If $i=j=k=0$, then
\begin{align*}
&\{a z^i \ _\lambda\ \{b z^j \ _\mu\ c z^k \}_\alpha \}_\alpha  = (\alpha^2 z^2-\alpha z+1) [a,[b,c]]+ k \lambda (a,[b,c]), \\
&\{b z^j \ _\mu\ \{a z^i \ _\lambda\ c z^k \}_\alpha \}_\alpha = (\alpha^2 z^2-\alpha z+1) [b,[a,c]]- k \mu (b,[a,c]),\\
&\{\{ a z^i \ _\lambda\ b z^j \}_{\alpha \ \lambda+\mu\ } c z^k\}_\alpha = (\alpha^2 z^2-\alpha z +1)[[a,b],c]+k(\lambda+\mu)([a,b],c),
\end{align*}
Similarly, one can check (\ref{Eqn:7.33}) in other cases.  
Hence Jacobi identity holds and the compatibility of $\{\cdot _\lambda \cdot \}_1$ and $\{\cdot _\lambda \cdot \}_2$ is proved .
\end{proof}

\section{Generating elements of a classical affine fractional $\WW$-algebra and Poisson $\lambda$-brackets  between them} \label{Sec:8}

\subsection{Generating elements of  classical affine fractional $\WW$-algebras}\label{Subsec:8.1}

Recall the universal Lax operator 
\begin{equation} \label{Eqn:0225_7.1}
\mathcal{L}_m= k\partial+Q_m+\Lambda_m \otimes 1 = k\partial + \sum_{(i,j)\in \mathcal{I}} \widetilde{u}_i^{-j} \otimes u_i^j+ \Lambda_m \otimes 1 \in \CC\partial \ltimes \hat{\g} \otimes \mathcal{V}_{[d,m]}.
\end{equation} 
We notice the map
$ \ad f: \g(i) \to \g(i-1),\  i>0$
is injective. Hence if $\mathfrak{b}=\bigoplus_{i \geq -\frac{1}{2}} \g(i)$, then the adjoint map
$\ad f: \n\cdot z^{-m} \to \mathfrak{b} \cdot z^{-m}$
is also injective. Let us choose an $\ad h$-invariant complementary subspace $V_m$ in $\mathfrak{b} z^{-m}$ such that
$$ \mathfrak{b}\cdot z^{-m} = V_m \oplus [f, \n \cdot z^{-m}]$$
and let us denote  $\CC_{\text{diff}}[V]= \CC_{\text{diff}}[b_i|i \in J]$ when $\{b_i\}_{i\in J}$ is a basis of the vector space $V$ .

\begin{prop} \label{Prop:8.1}
Let $\overline{\CC_{\text{diff}}[\hat{\g}^m]}$ be the subspace $\CC_{\text{diff}}[\hat{\g}^m]+I$ of $\mathcal{V}_{[d,m]}$. Then there exist unique $S \in \n \otimes \overline{\CC_{\text{diff}}[\hat{\g}^m]}\subset\n \otimes \mathcal{V}_{[d,m]}$ and unique $Q_m^{can} \in (V_m  \oplus \bigoplus_{-m< i   \leq 0}\hat{\g}^i ) \otimes \mathcal{V}_{[d,m]}$ such that
\begin{equation}\label{Eqn:8.1}
e^{adS} \mathcal{L}_m = \mathcal{L}_m^{can} = k \partial + Q_m^{can} + \Lambda_m \otimes 1.
\end{equation}
\end{prop}

\begin{proof}
Equation (\ref{Eqn:8.1}) can be rewritten as
\begin{equation}\label{Eqn:8.2}
Q_m^{can}  +[\Lambda_m \otimes 1,S]  = Q_m +[S, k\partial+Q_m] + \frac{1}{2} [S,[S, \mathcal{L}_m]] + \frac{1}{6} [S,[S,[S,\mathcal{L}_m]]] + \cdots.
\end{equation}
Since $\Lambda_m = -f z^{-m} -p ^{-m-1}$ and $p \in \ker \ad\n$, equation (\ref{Eqn:8.2}) is same as
\begin{equation}\label{Eqn:8.3}
\begin{aligned}
& Q_m^{can}  +[-f z^{-m}\otimes 1,S]  = Q_m +[S, k\partial+Q_m] \\
& + \frac{1}{2} [S,[S, k\partial + Q_m - f z^{-m}\otimes 1]] + \frac{1}{6} [S,[S,[S, k\partial + Q_m -f z^{-m} \otimes 1]]] + \cdots.
\end{aligned}
\end{equation}
Let $S=\sum_{i > 0} S_i$, $Q_m = \sum_{i >-(j+1)m-1, k\geq0} Q_{i,k}$ and $Q_m^{can} = \sum_{i >-(j+1)m-1,k\geq0} Q_{i,k}^{can}$,  where $S_i \in \hat{\g}^0_i \otimes \mathcal{V}_{[d,m]}$,  $Q_{i,k} \in \hat{\g}_i^{-k} \otimes \mathcal{V}_{[d,m]}$ and $Q^{can}_{i,k} \in \hat{\g}_i^{-k} \otimes \mathcal{V}_{[d,m]}$.  By the injectivity of $\ad f: \n z^{-m} \to \mathcal{B} z^{-m}$ and the equation
\begin{equation} \label{Eqn:8.4}
 Q_{-(j+1)m-\frac{1}{2}, m}^{can} + [f z^{-m} \otimes 1, S_{\frac{1}{2}}] = Q_{-(j+1)m-\frac{1}{2},m}\in \hat{g}^{-m}_{-(d+1)m-\frac{1}{2}} \otimes \overline{\CC_{\text{diff}}[\hat{\g}^m]}
\end{equation}
which follows from equation (\ref{Eqn:8.3}),
we can find $Q_{-(j+1)m-\frac{1}{2},m}^{can}\in (\hat{\g}_{-(j+1)m-\frac{1}{2}}^{-m}\cap V_m) \otimes\overline{\CC_{\text{diff}}[\hat{\g}^m]}$ and $S_{\frac{1}{2}}\in \hat{\g}_{\frac{1}{2}}^0 \otimes \overline{\CC_{\text{diff}}[\hat{\g}^m]}$ uniquely. Similarly, since  $Q_{-(j+1)m,m}^{can}$ and $S_{1}$ are determined by $Q_{-(j+1)m-\frac{1}{2},m}^{can}$, $S_{\frac{1}{2}}$, and $Q_m$, they should be in $\hat{\g}_{-(j+1)m}^{-m} \otimes \overline{\CC_{\text{diff}}[\hat{\g}^m]}$ and $\hat{\g}^0_1 \otimes \overline{\CC_{\text{diff}}[\hat{\g}^m]}$, respectively. By the induction on the $\gr_2$-grading, the whole $S$ and the $(\hat{\g}^{-m}\otimes \mathcal{V}_{[d,m]})$-part of $Q_m^{can}$  can be found uniquely. Once we find $S$, the $(\hat{\g}^{-k}\otimes \mathcal{V}_{[d,m]})$-part of $Q_m^{can}$, for any $k <m$, automatically come out from the equation $e^{\ad S} \mathcal{L}_m = \mathcal{L}_m^{can}$.
\end{proof}

\begin{prop} \label{Prop:8.2}
Let   $\{\widetilde{w}_i^{-m}| i \in \mathcal{J}_m\}$ be a basis of $V_m$, where $\widetilde{w}_i$ is an  eigenvector of $\ad \ h$ and $\widetilde{w}_i^{-m}:= \widetilde{w}_i z^{-m}$ and let $\{\widetilde{u}_i^{-j}|(i,j)\in \mathcal{J}\}:=\mathcal{B}^- \cap \bigoplus_{l<m} \hat{\g}^{-l}$  be a basis of $\bigoplus_{-m<i  \leq 0}\hat{\g}^i $, where $\mathcal{B}^-$ is defined in (\ref{Eqn:7.7}). Also, let us denote the dual elements of $\widetilde{w}_l^{-m}$ and $\widetilde{u}_i^{-j}$ by $w_l^m$ and $u_i^j$, respectively. If we have
$$Q_m^{can}= \sum_{i\in \mathcal{J}_m} \widetilde{w}_i^{-m} \otimes \gamma_{w_i^{m}} + \sum_{(i,j) \in \mathcal{J}} \widetilde{u}_i^{-j} \otimes \gamma_{u_i^j},$$ then $\WW(\g, \Lambda_m, k)$ is freely generated by $\gamma_{w_i^m}$ and $\gamma_{u_i^j}$ as a differential algebra. Moreover, the generating elements have the form
\begin{equation} \label{Eqn:8.5}
\begin{aligned}
\gamma_{w_i^{m}}& = w_i^{m} + A_{ w_i^{m} },\ \gamma_{u_i^j} = u_i^{j} + B_{u_i^{j}},
\end{aligned}
\end{equation}
where $ A_{ w_i^{m} }\in \CC_{\text{diff}}[\ v\ |  \gr_1(v)=m,\ \gr_2(v)>\gr_2(w_i^m) ]$ and $ B_{u_i^{j}}\in \CC_{\text{diff}}[\ v\ | \gr_1(v)=m \text{ or } \gr_1(u_i^j), \ \gr_2(v)>\gr_2(u_i^j) ]$.
\end{prop}

\begin{proof}
Equation (\ref{Eqn:8.5}) follows from equation (\ref{Eqn:8.3}) and the proof of Proposition \ref{Prop:8.1}. Hence the only thing to prove is that $\WW(\g, \Lambda_m, k)=\CC_{\text{diff}}[\ \gamma_{ w_l^m }(Q_m), \gamma_{u_i^j} (Q_m)\ |\  (i,j) \in \mathcal{J},\ l \in \mathcal{J}_m\ ]$. Suppose
$\Phi \in \mathcal{V}_{[d,m]} $
is a gauge invariant functional. Then
$ \Phi(Q_m)= \Phi(Q_m^{can})$.
Since $ w_l^m (Q_m^{can}) = \gamma_{w_l^m}(Q_m),  u_i^j (Q_m^{can}) = \gamma_{u_i^j}(Q_m),$ we have
\begin{align*}
 \Phi(Q_m) = \Phi(Q_m^{can})  & \in \CC_{\text{diff}}[\ u_i^j (Q_m^{can}),\ w_l^m (Q_m^{can})\ |\  (i,j) \in \mathcal{J},\ l \in \mathcal{J}_m\ ]\\
 & =  \CC_{\text{diff}}[\ \gamma_{ w_l^m }(Q_m),\ \gamma_{u_i^j} (Q_m)\ |\  (i,j) \in \mathcal{J},\ l \in \mathcal{J}_m\ ].
\end{align*}
Hence we conclude that $\WW(\g, \Lambda_m, k) \subset  \CC_{\text{diff}}[\ \gamma_{ w_l^m }(Q_m),\ \gamma_{u_i^j} (Q_m)\ |\  (i,j) \in \mathcal{J},\ l \in \mathcal{J}_m\ ]$ .

On the other hand, Suppose that $Q_m \sim Q_m'$. Then  $Q_m ^{(can)}= Q_m'^{(can)}$. Since
$ \gamma_{u_i^j}(Q_m) = u_i^j (Q_m^{can})= u_i^j (Q_m'^{can})= \gamma_{u_i^j}(Q_m')$ and 
$ \gamma_{w_i^m}(Q_m) = w_i^m (Q_m^{can})= w_i^m (Q_m'^{can})= \gamma_{w_i^m}(Q_m')$,
 we have
 $ \CC_{\text{diff}}[\ \gamma_{ w_l^m }(Q_m),\ \gamma_{u_i^j} (Q_m)\ |\  (i,j) \in \mathcal{J},\ l \in \mathcal{J}_m\ ] \subset \WW(\g, \Lambda_m, k). $
\end{proof}

\subsection{ Examples} \label{Subsec:8.2}
\begin{ex} Let $\g=\sll_2$.
The universal Lax operator for the $m$-th fractional $\WW$-algebra $\WW(\sll_2, -fz^{-m}-ez^{-m-1}, 1)$ is 
\begin{align*}
 \mathcal{L}_m & = \partial + \sum_{j=0}^{m-1}(ez^{-j} \otimes f z^{j} + h z^{-j} \otimes x z^{j} + f z^{-j} \otimes e z^{j}) \\ & + ez^{-m} \otimes f z^{m} + h z^{-m} \otimes x z^{m} + ( f z^{-m} + e z^{-m-1}) \otimes - 1.
\end{align*}
The gauge transformation of $\mathcal{L}_m$ by $ S= e \otimes A \in \n \otimes \mathcal{V}_{[d,m]}$ is
\begin{equation*}
 e^{\ad S}\mathcal{L}_m   =: \mathcal{L}_m^{can}= (I\otimes 1 + e \otimes A ) (\partial+ Q_m + \Lambda_m \otimes 1) ( I \otimes 1 - e \otimes A).
\end{equation*}

By Proposition \ref{Prop:8.1}, there exists a unique $q_m^{can}\in  (\CC fz^{-m}\oplus \bigoplus_{-m<i\leq 0} \hat{g}^i ) \otimes \mathcal{V}_{[d,m]}$ which is gauge equivalent to $q_m$. We can write $q_m=  \sum_{j=0}^{m}  ez^{-j} \otimes fz^j (q_m) + h z^{-j} \otimes xz^j (q_m) + fz^{-j} \otimes ez^j(q_m)
$ in a matrix form:
\begin{equation*}
q_m= \sum_{j=0}^m\left[
\begin{array}{cc}
x z^j (q_m) & f z^j(q_m) \\
e z^j (q_m)& -x z^j(q_m)\\
\end{array}
\right] z^{-j},\text{ where } \ ez^m (q_m)=0.
\end{equation*}

 Letting $e^{\ad S} \mathcal{L}_m = \partial + \widetilde{q}_m+\Lambda\otimes 1$, the gauge equivalent element  $\widetilde{q}_m$ to $q_m$ can be written in a matrix form,
\begin{align*}
\widetilde{q}_m&= \sum_{j=0}^{m-1}\left[
\begin{array}{cc}
x z^j (q_m)+ A (ez^j)(q_m)& f z^j (q_m)-2A (xz^j)(q_m) -A^2 (ez^j)(q_m) - \delta_{j,0} A'\\
e z^j (q_m)& -x z^j(q_m)-A(ez^j)(q_m)\\
\end{array}
\right] z^{-j} \\
&+ \left[
\begin{array}{cc}
x z^m (q_m)- A (q_m)& f z^j (q_m)-2A (xz^m)(q_m) +A^2 (q_m) \\
0 & -x z^j(q_m)+A(q_m)\\
\end{array}\right].
\end{align*}
By substituting $A= x z^m$, we obtain  
\begin{equation} \label{Eqn:8.9}
q^{can}_m + \Lambda_m=
\sum_{j=0}^{m+1}
\left[ \begin{array}{cc}
\gamma_{x z^j}(q_m) & \gamma_{f z^j}(q_m) \\
\gamma_{e z^j}(q_m) & -\gamma_{x z^j}(q_m)\\
\end{array} \right]z^{-j},
\end{equation}
where
\begin{equation}\label{Eqn:8.10}
\begin{aligned}
&\gamma_f = f -2 (xz^{m} ) x - (x z^{m})^2 e -  \partial(x z^m); 
 \ \ \ \gamma_{fz^i}= fz^i - 2(x z^{m})(x z^i) - (x z^m)^2 (e z^i); \\
& \gamma_{x z^i} = x z^j + (x z^m) (e z^j); 
\ \  \ \ \ \gamma_{e z^j}  = e z^j; 
\ \ \ \ \ \gamma_{f z^m}  = f z^m -(x z^m)^2;  \\
&\gamma_{x z^m} =\gamma_{x z^{m+1}}= \gamma_{e^{m+1}}=0;
\ \ \ \ \ \ \gamma_{e z^m}=\gamma_{f z^{m+1}} =-1.
\end{aligned}
\end{equation}

Applying Proposition \ref{Prop:8.2}, we get a generating set  
$\{\ \gamma_X\ |\ X=f, fz^i, x z^j, e z^j, f z^m, \  i=1, \cdots, m, \ j =0, \cdots, m-1 \}$ of $\WW(\sll_2, \Lambda_m=-f z^{-m}-e z^{-m-1}, 1).$

The first $\lambda$-brackets   between the elements in the generating set are as follows:
\begin{equation} \label{Eqn:8.11}
\begin{aligned}
&\{\gamma_{f \ \lambda} \gamma_f \}_1  = - 2 \lambda, \ \ \ \{\gamma_{f \ \lambda} \gamma_{f z^j}\}_1  = - 2 \gamma_{x z^j}, \\
&\{\gamma_{f \ \lambda} \gamma_{x z^i}\}_1  =  - \gamma_{f z^{i+1}} + \gamma_{e z^i},\ \ \{\gamma_{f \ \lambda} \gamma_{e z^i}\}_1  = 2 \gamma_{x z^{i+1}},
\end{aligned}
\end{equation}
where $i \geq 0$ and $j \geq 1$. Otherwise,
\begin{equation} \label{Eqn:8.12}
\{\gamma_{g_1 z^i  \ \lambda} \gamma_{g_2 z^j}\}_1  = -\gamma_{\{g_1,g_2\} z^{i+j+1}},
\end{equation}
where $g_1, g_2 \in \sll_2$ and $i, j \geq 0$.\\
The second $\lambda$-brackets are defined as follows:
\begin{equation}\label{Eqn:8.13}
\{ \gamma_{g_1 \ \lambda} \gamma_{g_2}\}_2  =\gamma_{\{g_1, g_2\}}+  \lambda(g_1, g_2) \\
\end{equation}
and
\begin{equation}\label{Eqn:8.14}
\begin{aligned}
&\{ \gamma_{f z_\lambda} \gamma_f \}_2  =-2 \gamma_x - \lambda, \ \ \ \{ \gamma_{f z_\lambda} \gamma_x \}_2  =\gamma_e, \ \ \ \{ \gamma_{f z_\lambda} \gamma_e \}_2  =0, \ \ \ \{ \gamma_{f z_\lambda} \gamma_{f z}\}_2  = 0,\\  &  \{ \gamma_{f z_\lambda} \gamma_{f z^j}\}_2 = -2 \gamma_{x z^j}, \ \ \ \{ \gamma_{f z_\lambda} \gamma_{x z^i}\}_2 = -\gamma_{f z^{1+i}}+\gamma_{e z^i}, \ \ \ \{ \gamma_{f z_\lambda} \gamma_{e z^i}\}_2 = 2 \gamma_{x z^{1+i}}, \\
\end{aligned}
\end{equation}
where $g_1,g_2 \in \sll_2$, $j \geq 2$ and $i,l \geq1$. Otherwise,
\begin{equation*}
\{ \gamma_{g_1 z^i \ \lambda} \gamma_{g_2 z^l} \}_2  = -\gamma_{\{g_1 z^i, g_2 z^l\} }.
\end{equation*}
\end{ex}

\vskip 5mm

\begin{ex}
Let $\g=\sll_n$ and $\Lambda_m= -e_{n1} z^{-m} -e_{1n} z^{-m-1}$.
Consider the dual bases 
\[\mathcal{B}=\{\ E^k_{ij},\ E^k_{ll}\ |  i \ne j, l=1, \cdots, n-1, k=0, \cdots, m\} \backslash \{e_{1n} z^{m}\} \] 
and 
\[\mathcal{B}^-=\{\ \widetilde{E}_{ij}^{-k},\ \widetilde{E}_{ll} z^{-k}\ |\ i \ne j, l=1, \cdots, n-1, k=0, \cdots, m\} \backslash \{e_{n1} z^{-m}\}\]  
of $ \hat{\g}^{\geq 0}_{<2m+1}$ and $ \hat{\g}^{\leq 0}_{>-2m-1}$, where $E^k_{ij}:= e_{ji} z^{k}$, $E^k_{ll}:=\frac{n-1}{n}e_{ll}z^k-\frac{1}{n}e_{nn}z^k$, $\widetilde{E}_{ij}^{-k}:= e_{ij}z^{-k}$ and $\widetilde{E}_{ll} z^{-k}:=(e_{ll}-e_{nn})z^{-k}$. The universal Lax operator  $\mathcal{L}_m = \partial+  \sum_{E_{ij}^{k}\in \mathcal{B}}  \widetilde{E}_{ij}^{-k} \otimes E_{ij}^k +\Lambda_m \otimes 1 \in \CC\partial \ltimes \hat{\sll}_n\otimes \mathcal{V}_{[d,m]}$ can be written in the matrix form as below:
\begin{align} \label{Eqn:8.15}
\mathcal{L}_m = \partial+ \sum_{k=0}^{m-1} z^{-k}\left(
E_{i,j}^k \right)_{i,j=1, \cdots, n} - z^{-m}(e_{n1}  - z^{-1}e_{1n} ), \ \text{ where } E_{nn}^k= -\sum_{i=1}^{n-1} E_{ii}^k.
\end{align}

Then  an element in the associated Lie group to the Lie algebra $\n\otimes \mathcal{V}_{[d,m]}$ is an upper diagonal matrix. \\
Suppose
\begin{equation}\label{Eqn:8.16}
S=\left[
\begin{array}{cccc}
1& s_{12}& \cdots & s_{1n} \\
0& 1 &0 & s_{2n} \\
\vdots& \ddots & \ddots & \vdots \\
0& \cdots &  0 & 1
\end{array}\right]\in \n\otimes \mathcal{V}_{[d,m]}
\end{equation}
and  $q_m\sim_S \bar{q}_m$, i.e. $ \bar{q}_m=S (q+\Lambda_m) S^{-1} -\Lambda_m + S\partial(S^{-1})$.  In the matrix form, $\bar{q}_m =\sum_{k=0}^{m} z^{-k} ( a^k_{ij})_{i,j=1}^n$, where
\begin{equation} \label{Eqn:8.17}
\begin{aligned}
{a}^q_{n1} & = E^q_{n1}, \ \ \ {a}^r_ {i1} = E^r_{i1}+ s_{in}E^r_{n1}, \ \ \ {a}^r_{11}  = E^r_{11}+\sum_{k=2}^n s_{ik}E^r_{k1}, \ \ \ {a}^r_ {nj}  = E^r_{nj}-E^r_{n1} s_{1j},  \\
{a}^r_{ij} & = E^r_{ij}+s_{in}E^r_{nj}-(E^r_{i1}+s_{in}E^r_{n1}) s_{1j},  \ \ \ {a}^p_{1j}  = E^p_{1j}-\sum_{k=2}^{n}s_{1k}E^p_{kj}-(E^p_{11}+\sum_{k=2}^n s_{1k}E^p_{k1})s_{1j} , \\
{a}^0_{1j} & = E^ 0_{1j}-\sum_{k=2}^{n}s_{1k}E^0_{kj}-(E^0_{11}+\sum_{k=2}^n s_{1k}E^0_{k1})s_{1j}-ks'_{1j} , \\
{a}^p_{in}& = E^p_{in}+s_{in} E^p_{nn} +(E^p_{i1}+\sum_{k=2}^n s_{ik}E^p_{k1})t_{1n} -\sum_{j=2}^{n-1} (E^p_{ij}+\sum_{k=2}^{n} s_{ik}E^p_{kj})s_{jn}  , \\
{a}^0_{in}& = E^p_{in}+s_{in} E^p_{nn} + (E^p_{i1}+\sum_{k=2}^n s_{ik}E^p_{k1})t_{1n} -\sum_{j=2}^{n-1} (E^p_{ij}+\sum_{k=2}^{n} s_{ik}E^p_{kj})s_{jn} -ks'_{in}, \\
{a}^p_{1n} & = E^p_{1n}+\sum_{k=2}^n s_{1k} E^p_{kn} + (E^p_{11}+\sum_{k=2}^n s_{1k}E^p_{k1})_{1n}t -\sum_{j=2}^{n-1} (E^p_{1j}+\sum_{k=2}^{n} s_{1k}E^p_{kj})s_{jn} ,\\
{a}^0_{1n} & = E^p_{1n}+\sum_{k=2}^n s_{1k} E^p_{kn} + (E^p_{11}+\sum_{k=2}^n s_{1k}E^p_{k1})t_{1n} -\sum_{j=2}^{n-1} (E^p_{1j}+\sum_{k=2}^{n} s_{1k}E^p_{kj})s_{jn} + kv_{1n} , \\
{a}^r_{nn}& =- \sum_{i=1}^{n-1} {a}^r_{ii} , \ \  E^m_{n1} =-1.\\
\end{aligned}
\end{equation}
Here  $t_{1n}=-s_{1n}+s_{12}s_{2n}+ \cdots + s_{1n-1 }s_{n-1 n}$ and $v_{1n}=-\partial s_{1n}+ \partial(s_{12})s_{2n} + \cdots + \partial(s_{1 n-1})s_{n-1 n}$ and
$p=1, \cdots, m$, $q=0, \cdots, m-1$, $r=0, \cdots, m$, $i,j=2, \cdots n-1$. If the entries of $S$ in (\ref{Eqn:8.16}) are picked as follows:

\begin{equation}\label{Eqn:8.18}
 s_{in}= \left\{
 \begin{array}{cl}
 -\frac{E_{i1}^m}{E_{n1}^m} & \text{ if } E_{n1}^m\neq 0\\
 0 & \text{ otherwise}
 \end{array}\right. , \ \  
  s_{1j}= \left\{
 \begin{array}{cl}
 \frac{E_{nj}^m}{E_{n1}^m} & \text{ if } E_{n1}^m\neq 0\\
 0 & \text{ otherwise.}
 \end{array}\right.
 \end{equation}
Then we get the generators $ \gamma_{ij}^k (q_m) := E_{ij}^k (\bar{q}_m)$ of the differential algebra $\WW(\sll_n, \Lambda_m, 1)$. 

\end{ex}

\subsection{ More results on generating elements of a classical affine fractional $\WW$-algebra associated to a minimal nilpotent} \label{Subsec:8.3}
Suppose $f$ is a minimal nilpotent and $m \in \ZZ_{>0}$. Under this assumption, we can describe generating elements of $\WW(\g, \Lambda_m, k)$  and $\lambda$-brackets between them. The following lemma is useful to find a generating set of $\WW(\g,\Lambda_m, k)$.

\begin{lem} \label{Lem:8.7}
Let $X:=\{\ w_l^m\ |\ l \in \mathcal{J}_m\ \}$ and $Y:=\{\ u_i^j\ |\ (i,j)\in \mathcal{J}\ \}$ be defined as in Proposition \ref{Prop:8.2} and let 
$\psi: \g_f \cdot z^m \oplus \bigoplus _{i=0}^{m-1} \g \cdot z^i\to \WW(\g, \Lambda_m,k)$ be a linear map such that
$\psi_{v}:= \psi(v)= v + C_v$,
where $v\in X \cup Y$ and $C_v\in \CC[\ u, u', u'', \cdots\ |\ u \in X \cup Y, \ \gr_2(u)>\gr_2(v)\ ]$. Then $\Psi:= \{\ \psi_v\ |\ v \in X\cup Y\ \} $ freely generates $\WW(\g, \Lambda_m, k).$
\end{lem}
\begin{proof}
Since we know that the set $\{\ \gamma_v\ |\ v \in X\cup Y\ \}$ defined in (\ref{Eqn:8.5}) freely generates $\WW(\g, \Lambda_m, k)$, it suffices to show that any $\gamma_v$ is in the differential algebra generated by $\Psi$. We know that $\gamma_v-\psi_v$ is an element in $\WW(\g,\Lambda_m, k)$. Hence, by Proposition \ref{Prop:8.2} and the assumption on $C_v$, 
the element  $\gamma_{v} -\psi_{v}\in \CC_{\text{diff}}[\ \gamma_u|\ u \in X \cup Y, \ \gr_2(u)>\gr_2(v)\ ]$.  By the induction on $\gr_2$-grading, we have $\gamma_v -\psi_v \in \Psi$ so that $\gamma_v\in \Psi$.
\end{proof}

In order to state and prove the main theorems in this section (see Theorem\ref{Thm:8.8} and \ref{Thm:8.9}), take two bases  $\{z_i\}_{i=1}^{2s}$ and $\{z_i^*\}_{i=1}^{2s}$ of $\g\left(\frac{1}{2} \right)$  such that $[z_i, z_j^*] = \delta_{ij} e$ and let $z_{2s+1}:= x$ and $z_{2s+1}^*:= e.$  Also, we let $\{\cdot, \cdot\}$ be the  Poisson bracket on $\mathcal{V}_{[d,m]}$ induced from the Lie bracket on $\hat{\g}$.

\begin{thm} \label{Thm:8.8}
 Let 
$\tilde{\eta}_m:\hat{\g}^+_{[d,m]}\to \mathcal{V}_{[d,m]}$ be a linear map such that
\begin{equation}
\tilde{\eta}_m(a)  = \sum_{l=0}^4 \sum_{i_i, \cdots, i_l=1}^{2s+1} \frac{1}{l!} (z_{i_1} z^m) \cdots (z_{i_l} z^m) \{a, z_{i_1}^*,\cdots, z_{i^l}^*\},
\end{equation}
where $\{a_1, a_2, \cdots, a_t\}$ is the $(t-1)$-brackets $\{\{\cdots \{a_1, a_2\} , \cdots,  a_{t-1}\}, a_t \} $, and let $\eta_m:\hat{\g}^+_{[d,m]}\to \mathcal{V}_{[d,m]}$ be a linear map such that
\begin{equation}
\begin{aligned}
& \eta_m(a)= \tilde{\eta}_m(a)\ \ \text{ if }   a \in \hat{\g}^0_{\geq0}\oplus \left(\bigoplus_{t=1}^{m-1}\hat{\g}^t\right)\oplus \g_f z^m ,\\
& \eta_m(b) = \tilde{\eta}_m(b) -\sum_{i=1}^{2s} (z_i^*, w) k\partial(z_iz^m)\ \ \text{ if } b \in \hat{\g}^0_{-\frac{1}{2}}, \\
& \eta_m(f) = \tilde{\eta}_m(f) -k \partial(xz^m) -\frac{k}{2} \sum_{i=1}^{2s} \partial(z_i^* z^m) (z_iz^m).
\end{aligned}
\end{equation}
Then \[\CC_{\text{diff}}\left[\eta_m\left(\bigoplus_{i=0}^{m-1}\g z^i \oplus \g_f z^m\right)\right]=\WW(\g, \Lambda_m, k).\] 
\end{thm}

Equivalently, Theorem \ref{Thm:8.8} can be written as follows.

\begin{thm} 
Suppose that $u \in \g
\left(\frac{1}{2}\right), v \in \g(0), w \in \g\left(-\frac{1}{2} \right)$, $t=0, \cdots, m-1,$  and $ v_f \in \g_f(0),\ w_f \in \g_f\left(-\frac{1}{2} \right)$. Then the following elements freely generate $\WW(\g, \Lambda_m, k)$ as a differential algebra : 
\begin{align*}
 \eta_m (ez^t) & = ez^t; \\
 \eta_m (uz^t) & = u z^t+ \sum_{i=1}^{2s} (z_i z^m) \{ uz^t, z_i^*\}; \\
 \eta_m (v z^t) & = vz^t + \sum_{i=1}^{2s} (z_i z^m) \{ vz^t, z_i^*\}  + (x z^m)\{v,e\}+ \frac{1}{2} \sum_{i,j=1}^{2s} (z_j z^m)(z_i z^m) \{vz^t, z_i^*, z_j^*\};\\
 \eta_m (wz^t) & =  wz^t+ \sum_{i=1}^{2s} (z_i z^m) \{ wz^t, z_i^*\}  + (x z^m)\{wz^t,e\}+ \frac{1}{2} \sum_{i,j=1}^{2s} (z_j z^m)(z_i z^m) \{wz^t, z_i^*, z_j^*\} \\ & + \sum_{i=1}^{2s} (x z^m)(z_i z^m) \{wz^t, e , z_i^*\} + \frac{1}{6}\sum_{i,j,k=1}^{2s} (z_i z^m) (z_j z^m) (z_k z^m) \{ w^t, z_i^*, z_j^*, z_k^*\} \\
  & -\delta_{t,0}\left( \sum_{i=1}^{2s}(z_i^*, w) k \partial (z_iz^m) \right) ;\\
 \eta_m (fz^t) &=  fz^t+ \sum_{i=1}^{2s} (z_i z^m) \{ fez^t, z_i^*\}  - 2 (x z^m)(xz^t)+ \frac{1}{2} \sum_{i,j=1}^{2s} (z_j z^m)(z_i z^m) \{fz^t, z_i^*, z_j^*\} \\ & -2 \sum_{i=1}^{2s} (x z^m)(z_i z^m) \{xz^t , z_i^*\} - (xz^m)^2 (ez^t) + \frac{1}{6}\sum_{i,j,k=1}^{2s} (z_i z^m) (z_j z^m) (z_k z^m) \{ fz^t, z_i^*, z_j^*, z_k^*\} \\
 & - \sum_{i,j=1}^{2s} (x z^m) (z_i z^m) (z_j z^m)\{xz^t,z_i^*, z_j^*\} \\
 &+ \frac{1}{24}\sum_{i,j,k,l=1}^{2s} (z_i z^m) (z_j z^m) (z_k z^m)(z_l z^m) \{ fz^t, z_i^*, z_j^*, z_k^*, z_l^*\}\\
 &-\delta_{t,0}\left(k \partial (xz^m) +\frac{k}{2} \sum_{i=1}^{2s} \partial (z_i^*z^m) (z_i z^m)\right); \\
\end{align*}
and 
\begin{align*}
 \eta_m (v_f z^m) & = v_fz^m + \frac{1}{2}\sum_{i=1}^{2s} (z_i z^m) \{ v_fz^m, z_i^*\} ;\\
 \eta_m (w_f z^m) & =  w_f z^m+ \sum_{i=1}^{2s} (z_i z^m) \{ w_f z^m, z_i^*\} + \frac{1}{3} \sum_{i,j=1}^{2s} (z_j z^m)(z_i z^m) \{w_f z^m, z_i^*, z_j^*\}   ;\\
 \eta_m (f z^m) &=  fz^m+ \sum_{i=1}^{2s} (z_i z^m) \{ fz^m, z_i^*\}  -  (x z^m)^2+ \frac{1}{2} \sum_{i,j=1}^{2s} (z_j z^m)(z_i z^m) \{fz^m, z_i^*, z_j^*\} \\ &  + \frac{1}{6}\sum_{i,j,k=1}^{2s} (z_i z^m) (z_j z^m) (z_k z^m) \{ fz^m, z_i^*, z_j^*, z_k^*\} \\
 &+ \frac{1}{24}\sum_{i,j,k,l=1}^{2s} (z_i z^m) (z_j z^m) (z_k z^m)(z_l z^m) \{ fz^m, z_i^*, z_j^*, z_k^*, z_l^*\}.
 \end{align*}

\end{thm}

\begin{proof}
By Lemma \ref{Lem:8.7}, the only thing to check is that the $\ad_\lambda \n$-action trivially acts on $\eta_m\left(\bigoplus_{i=0}^{m-1}\g z^i \oplus \g_f z^m\right)$.  For example, in the case when $ w \in \g \left( \frac{1}{2} \right)$ and $t>0$, we have  
\begin{align*}
&\{ z^*_{l \ \lambda} \sum_{i=1}^{2s} (z_iz^m) \{w z^t, z_i^*\}\}  = \{w z^t, z_l^*\} + \sum_{i=1}^{2s} (z_i z^m) \{ z_l^* ,\{w z^t, z_i^*\}\} ; \\
&\{ z^*_{l \ \lambda} \sum_{i,j=1}{2s} (z_i z^m) (z_j z^m) \{w z^t, z_i^*, z_j^*\} \}= 2 \sum_{i=1}^{2s} (z_i z^m) \{w z^t, z_i^i, z_l^*\}  \\
& -\sum_{i,j=1}^{2s} (z_i z^m) (z_j z^m) \{w z^t, z_i^*, z_j^*, z_k^*\} -\{z_l^* , \{v,e\} (2xz^m)\} \\
&+ \{ z_l^*, \sum_{i=1}^{2s} \{z_i^*, \{v,e\}\} (z_i z^m) (2x z^m)\} + \sum_{i=1}^{2s} \{z_i^*, \{v,e\}\} (z_i z^m) (z_l^* z^m) ;\\
\end{align*} 
and 
\begin{align*}
& \{ z^*_{l \ \lambda} \sum_{i,j,k=1}^{2s} (z_i z^m) (z_j z^m) (z_k z^m) \{ v, z_i^*, z_j^*, z_k^*\}\} \\
& = 3 \sum_{i,j=1}^{2s} (z_i z^m) (z_j z^m) \{v, z_i^*, z_j^*, z_l^*\} + 3 \sum_{i=1}^{2s} \{v, z_i^*, e\} (z_i z^m) (z_l^* z^m) .
\end{align*}
Hence $  \{ z^*_{l \ \lambda}  \eta_m(w z^t) \}=0.$
The most complicated case to check is $\{z^*_{l \ \lambda} \ \eta_m(fz_k)\} =0.$ However, since
$\{\eta_m(v z^i) _\lambda  \eta_m(w z^j)\} = \eta_m((\{v,w\}, e) fz^{i+j})$,
by Jacobi identity, we have $\{z^*_{l,\ \lambda} \ \eta_m(fz_k)\} =0.$
\end{proof}

\begin{thm} \label{Thm:8.9}
Assume that $m >0$. The first $\lambda$-brackets between the generating elements are
\begin{align*}
& \{ \eta_m(f)  _\lambda   \eta_m(f) \}_1 = -2k\lambda; \\
& \{ \eta_m(f)  _\lambda  \eta_m(u z^i) \}_1 = -\eta_m([ f, u z^i]z ) +\eta_m([ u z^i, e ]),\ \text{ if }   \ u z^i  \neq f ; \\
& \{ \eta_m(a z^i)  _\lambda  \eta_m(b z^j)\}_1 = - \eta_m([a,b] z^{i+j+1} ),\ \  \text{ otherwise };
\end{align*}
and the second $\lambda$-brackets are
\begin{align*}
& \{ \eta_m(a) _\lambda  \eta_m(b) \}_2 = \eta_m([a, b])+ k \lambda (a,b) ; \\
& \{ \eta_m(fz) _\lambda  \eta_m(f) \}_2 = -2\eta_m(x)-k \lambda;\\
& \{ \eta_m(fz)  _\lambda  \eta_m(u) \}_2 = \eta_m([u, e]) ; \\
& \{ \eta_m(fz) _\lambda \eta_m(a z^i) \}_2 = -\eta_m([fz, a z^i]) + \eta_m([ a z^i, e ]); \\
& \{ \eta_m(a z^i) _\lambda\eta_m(b z^j)\}_2 =- \eta_m([a,b] z^{i+j} );
\end{align*}
for $a, b\in \g$, $\ u \in \bigoplus_{i=-\frac{1}{2}}^{1} \g(i)$, and $i,j>0.$
\end{thm}

We need the following remark and the lemma to show Theorem \ref{Thm:8.9}.

\begin{rem}\label{Rem:0326_7.6}
Recall that $\{\cdot_\lambda \cdot\}_1$ and $\{\cdot_\lambda \cdot\}_2$ are the $\lambda$-brackets on $\mathcal{V}_{[d,m]}$. We have
\begin{enumerate}[(i)]
\item $\{\ \n z^m\ _\lambda\ \mathcal{V}_{[d,m]}\ \}_1= \{\ \n z^m\  _\lambda\ \mathcal{V}_{[d,m]}\ \}_2=0,$
\item if $a \in \bigoplus_{i>-1}\g(i) \oplus\bigoplus_{t>0}\hat{\g}^t$ then $ \{ xz^m_\lambda\ a \}_1= 0$ and $\{ xz^m\ _\lambda \ f \}_1=  f z^{m+1} = -1,$
\item if $a \in \g \oplus\bigoplus_{i>-1} \g(i) z \oplus\bigoplus_{t>1}\hat{\g}^t$ then $\{xz^m\ _\lambda\ a\}_1= 0$ and $\{ xz^m\ _\lambda\  fz \}_2=  f z^{m+1} = -1 .$
\end{enumerate}
In $m=0$ case,  (i) and (iii) are not true since $\{z_{i\ \lambda } z_j^*\}_2 =-1$ and $\{x_\lambda z_i\} _2 = \frac{1}{2} z_i.$
\end{rem}

\begin{lem}\label{Lem:0326_7.7}
Let $v,w \in \mathcal{V}_{[d,m]}$. Then for any $k\geq 0$, the following equation holds:
\begin{equation} \label{Eqn:0326_7.17}
\begin{aligned}
&\sum_{i_1, \cdots, i_r=1}^{2s+1} \frac{1}{r!} \{\{v,w\}, z_{i_1}^*, z_{i_2}^*, \cdots, z_{i_r}^*\} \\
&= \sum_{l=0}^{r}\frac{1}{l!}\frac{1}{(r-l)!}\sum_{i_1, \cdots, i_r=1}^{2s+1} \{v, z_{i_1}^*, z_{i_2}^*, \cdots, z_{i_l}^*\}\{w, z_{i_{l+1}}^*, z_{i_{l+2}}^*, \cdots, z_{i_r}^*\}.
\end{aligned}
\end{equation}
\end{lem}
\begin{proof}
If $r=0$, equation (\ref{Eqn:0326_7.17}) holds obviously. Suppose that we have (\ref{Eqn:0326_7.17}) when $r=n$. Then by Jacobi identity, we have
\begin{equation} 
\begin{aligned}
&\sum_{i_1, \cdots, i_{n+1}=1}^{2s+1} \frac{1}{(n+1)!} \{\{\{v,w\}, z_{i_1}^*\}, z_{i_2}^*, \cdots, z_{i_{n+1}}^*\} \\
&= \frac{1}{n+1}\sum_{l=0}^{n}\frac{1}{l!}\frac{1}{(n-l)!}\sum_{i_1, \cdots, i_{n+1}=1}^{2s+1} \{v, z_{i_1}^*, z_{i_2}^*, \cdots, z_{i_l}^*, z_{i_{n+1}}^*\}\{w, z_{i_{l+1}}^*, z_{i_{l+2}}^*, \cdots, z_{i_n}^*\}\\
& + \frac{1}{n+1}\sum_{l=0}^{n}\frac{1}{l!}\frac{1}{(n-l)!}\sum_{i_1, \cdots, i_{n+1}=1}^{2s+1} \{v, z_{i_1}^*, z_{i_2}^*, \cdots, z_{i_l}^*\}\{w, z_{i_{l+1}}^*, z_{i_{l+2}}^*, \cdots, z_{i_n}^*, z_{i_{n+1}}^*\}\\
&= \frac{1}{n+1}\sum_{l=0}^{n}\frac{(l+1)+(n-l)}{(l+1)!(n-l)!}\sum_{i_1, \cdots, i_{n+1}=1}^{2s+1} \{v, z_{i_1}^*, z_{i_2}^*, \cdots, z_{i_l}^*, z_{i_{n+1}}^*\}\{w, z_{i_{l+1}}^*, z_{i_{l+2}}^*, \cdots, z_{i_n}^*\} \\
&+ \frac{1}{(n+1)!} \sum_{i_1, \cdots, i_{n+1}=1}^{2s+1}\{v,\{w,z_1^*, \cdots, z_{i_{n+1}}^*\}\}\\
&= \sum_{l=0}^{n+1}\frac{1}{l!}\frac{1}{((n+1)-l)!}\sum_{i_1, \cdots, i_{n+1}=1}^{2s+1} \{v, z_{i_1}^*, z_{i_2}^*, \cdots, z_{i_l}^*\}\{w, z_{i_{l+1}}^*, z_{i_{l+2}}^*, \cdots, z_{i_{n+1}}^*\}.
\end{aligned}
\end{equation}
Hence (\ref{Eqn:0326_7.17}) holds for any $r\geq 0.$
\end{proof}

$\bold{Proof\ of\ Theorem\ \ref{Thm:8.9}:}$

$\newline$

If $ a,b \in \bigoplus_{i>-1} \g(i)\oplus(\bigoplus_{t>0}\hat{\g}^t)$, by Remark \ref{Rem:0326_7.6} and Lemma \ref{Lem:0326_7.7}, we have
\begin{equation}
\begin{aligned}
& \{\eta_m(a)_\lambda \eta_m(b)\}_1= \{\tilde{\eta}_m(a)_\lambda \tilde{\eta}_m(b)\}_1 \\
& =\sum_{r=0}^{4}  \sum_{i_1, \cdots, i_r=1}^{2s+1} \frac{1}{r!}\left((z_{i_1} z^m) \cdots (z_{i_r} z^m)\right) (-\{\{a,b\}, z_{i_1}^*, z_{i_2}^*, \cdots, z_{i_r}^*\} z)\\
 &=- \tilde{\eta}_m([a,b]z)= -\eta_m([a,b]z).
\end{aligned}
\end{equation}
Also, we have 
\begin{equation}
\begin{aligned}
&\{\eta_m(f)_\lambda \eta_m(f)\}_1=- \{\tilde{\eta}_m(f), \tilde{\eta}_m(f)\}z+\{f_\lambda(-k\partial(xz^m))\}_1+\{(-k\partial(xz^m)_\lambda f\}_1\\
&\ \ \ \ \ \ \ \ \ =\{f_\lambda(-k\partial(xz^m))\}_1+\{(-k\partial(xz^m))_\lambda f\}_1=-2k\lambda,
\end{aligned}
\end{equation}
and, if $b \in \bigoplus_{i>-1}\g(i)\oplus(\bigoplus_{t>0}\hat{\g}^t)$, we obtain 
\begin{equation}
\begin{aligned}
& \\
&\{\eta_m(f)_\lambda \eta_m(b)\}_1\\
& =  -\tilde{\eta}_m([f,b]z) + \sum_{p=0}^2 \sum_{i_i, \cdots, i_p=1}^{2s}\frac{1}{p!} (z_{j_1} z^m) \cdots (z_{j_p}z^m)\{\{b, e\}, z_{i_1}^*,\cdots, z_{i_p}^*\} + \{b,e,e\}(xz^m)\\
&  =-\tilde{\eta}_m([f,b]z)+\tilde{\eta}_m([b,e])= -\eta_m([f,b]z)+\eta_m([b,e]).
\end{aligned}
\end{equation}

\vskip 3mm

Next, let us compute the second $\lambda$-brackets. If $a, b\in \g$, we have
\begin{equation}
\{\eta_m(a)_\lambda \eta_m(b)\}_2 = \{\tilde{\eta}_m(a)_\lambda \tilde{\eta}_m(b)\}_2
\end{equation}
 by Remark \ref{Rem:0326_7.6}, and 
 \begin{equation}
\begin{aligned}
& \{\tilde{\eta}_m(a)_\lambda \tilde{\eta}_m(b)\}_2 \\
 &= \tilde{\eta}_m([a,b]) + \sum_{r=0}^{4}  \sum_{l=0}^{r} \sum_{i_1, \cdots, i_r=1}^{2s+1} 
\left((z_{i_{l+1}} z^m) \cdots (z_{i_r} z^m)\right)\\
 & \ \ \ \ \ \cdot\left( k(\lambda+\partial)(\{a, z_{i_1}^*, z_{i_2}^*, \cdots, z_{i_l}^*\},\{b, z_{i_{l+1}}^*, z_{i_{l+2}}^, \cdots, z_{i_r}^*\})\left((z_{i_1} z^m) \cdots (z_{i_l} z^m)\right)\right)\\
 &\\
 & =  \tilde{\eta}_m([a,b]) +k(\lambda+\partial)(a,b) + X_0+\lambda X_1 =\eta_m([a,b]) +k\lambda(a,b) + Y_0+\lambda Y_1 ,\\
 & 
 \end{aligned}
 \end{equation}
where $X_0, X_1 ,Y_0, Y_1\in \CC_{\text{diff}}\left[\left(\g\left(\frac{1}{2}\right) \oplus \g(1)\right)\cdot z^m\right]$ with zero constant term. Since $ \{\eta_m(a)_\lambda \eta_m(b)\}_2$ and  $\eta_m([a,b]) +k\lambda(a,b)$ are in $\WW(\g, \Lambda_m, k)[\lambda]$, the element $Y_0$ and $Y_1$ should be in  $\WW(\g, \Lambda_m, k)$. However, we know that  $(\CC_{\text{diff}}\left[\left(\g\left(\frac{1}{2}\right) \oplus \g(1)\right)\cdot z^m\right]\cap \WW(\g, \Lambda_m, k))=\CC$. Hence $Y_0=Y_1=0$. \\

If $a, b\in \bigoplus_{i>-1}\g(i)z \oplus \bigoplus \hat{\g}^{\geq 2}$, then 
 \begin{equation}
\begin{aligned}
\\
& \{\eta_m(a)_\lambda \eta_m(b)\}_2=\{\tilde{\eta}_m(a)_\lambda \tilde{\eta}_m(b)\}_2 \\
 &=  \sum_{r=0}^{4}  \sum_{l=0}^{r} \sum_{i_1, \cdots, i_r=1}^{2s+1}  \frac{1}{l!}\frac{1}{(r-l)!}\left((z_{i_1} z^m) \cdots (z_{i_l} z^m)\right)\left((z_{i_{l+1}} z^m) \cdots (z_{i_r} z^m)\right)\\
 & \hskip 25mm \cdot(- \{ \ \{a, z_{i_1}^*, z_{i_2}^*, \cdots, z_{i_l}^*\}\ ,\ \{b, z_{i_{l+1}}^*, \cdots, z_{i_r}^*\}\ \})\\
 & = - \tilde{\eta}_m([a,b]) = -\eta_m([a,b]).
 \end{aligned}
 \end{equation}

The second $\lambda$-bracket between $\eta_m(fz)$ and $\eta_m(f)$ is
\begin{equation}
\begin{aligned}
&\{\eta_m(fz)_\lambda \eta_m(f)\}_2=   \sum_{p=0}^2\sum_{i_i, \cdots, i_p=1}^{2s}\frac{1}{p!} (z_{j_1} z^m) \cdots (z_{j_p}z^m)\{\{f, e\}, z_{i_1}^*,\cdots, z_{i_p}^*\}\left\{fz \ _\lambda \   (x z^m)\right\}_2 \\
&\hskip 30mm +  \{f,e,e\}\left\{fz \ _\lambda \   \frac{1}{2}(xz^m)^2 \right\}_2 -k\{fz_\lambda \partial(xz^m)\}_2 \\
& =  \sum_{p=0}^2\sum_{i_i, \cdots, i_p=1}^{2s}\frac{1}{p!} (z_{j_1} z^m) \cdots (z_{j_p}z^m)\{-2x, z_{i_1}^*,\cdots, z_{i_p}^*\} + \{-2x,e\}(xz^m)-k\lambda\\
&\\
& = -2\eta_m(x)-k\lambda.\\
&\\
\end{aligned}
\end{equation}

If $u \in \bigoplus_{i\geq-\frac{1}{2}}\g(i)$, then 
\begin{equation}
\begin{aligned}
&\{\eta_m(fz)_\lambda \eta_m(u)\}_2=   \sum_{p=0}^2\sum_{i_i, \cdots, i_p=1}^{2s}\frac{1}{p!} (z_{j_1} z^m) \cdots (z_{j_p}z^m)\{\{u, e\}, z_{i_1}^*,\cdots, z_{i_p}^*\}\left\{fz \ _\lambda \   (x z^m)\right\}_2 \\
& =\sum_{p=0}^2  \sum_{i_i, \cdots, i_p=1}^{2s}\frac{1}{p!} (z_{j_1} z^m) \cdots (z_{j_p}z^m)\{\{u,e\}, z_{i_1}^*,\cdots, z_{i_p}^*\} 
 = \eta_m([u,e]).
\end{aligned}
\end{equation}

If $ b\in  \bigoplus_{i>0} \hat{\g}^i$, then 
\begin{equation}
\begin{aligned}
 \{\eta_m(fz)_\lambda \eta_m(b)\}_2 &=  \sum_{r=0}^{4}  \sum_{l=0}^{r} \sum_{i_1, \cdots, i_r=1}^{2s+1}  \frac{1}{l!}\frac{1}{(r-l)!}\left((z_{i_1} z^m) \cdots (z_{i_l} z^m)\right)\left((z_{i_{l+1}} z^m) \cdots (z_{i_r} z^m)\right)\\
 & \hskip 25mm\cdot (- [\ \{fz, z_{i_1}^*, z_{i_2}^*, \cdots, z_{i_l}^*\}\ ,\ \{b, z_{i_{l+1}}^*, z_{i_{l+2}}^, \cdots, z_{i_r}^*\}\ ])\\
 & +  \sum_{p=0}^2\sum_{i_i, \cdots, i_p=1}^{2s}\frac{1}{p!} (z_{j_1} z^m) \cdots (z_{j_p}z^m)\{\{b, e\}, z_{i_1}^*,\cdots, z_{i_p}^*\}\left\{fz \ _\lambda \   (x z^m)\right\}_2 \\
 &  + \ \ \  \{b,e,e\}\left\{fz \ _\lambda \   \frac{1}{2}(xz^m)^2 \right\}_2\\
 &\\
 &= -\eta_m([fz, b]) + \eta_m([b,e]).
\end{aligned}
\end{equation}
So we proved Theorem  \ref{Thm:8.9}.

\section{Integrable systems related to classical affine fractional $\WW$-algebras} \label{Sec:9}

Assume that $\Lambda_m:=-f z^{-m}-p z^{-m-1}$ is a semisimple element in $\hat{\g}$. Then
\begin{equation} \label{Eqn:9.1}
\hat{\g} = \ker(\ad \Lambda_m) \oplus \text{im} (\ad \Lambda_m).
\end{equation}
In this case, the following property holds.

\begin{prop} \label{Prop:9.1}
Let $L_m=k\partial+q_m+\Lambda_m\otimes 1$ be a Lax operator defined as in Definition \ref{Def:7.3}. Then there exist unique $S \in \bigoplus_{i>0}\hat{\g}_{i} \otimes  \mathcal{V}_{[d,m]}$ and unique $h(q_m)\in (\ker \ad \Lambda_m \cap \hat{\g}) \otimes  \mathcal{V}_{[d,m]}$ such that  
\begin{equation} \label{Eqn:9.2}
L_{m,0}= e^{\ad S}(L_m)=k \partial +\Lambda_m \otimes1 + h(q_m).
\end{equation}
\end{prop}

\begin{proof}
Let $h(q_m)= \sum_{i > -(d+1)m-1} h_{i}(q_m)$, $ q_m = \sum_{i > -(d+1)m-1} q_{m,i} $ and $ S = \sum_{i> 0 } S_i$, where $h_{i}(q_m), q_{m,i},  S_i \in \hat{\g}_i \otimes  \mathcal{V}_{[d,m]}$. 
Then equation (\ref{Eqn:9.2}) can be rewritten as
\begin{equation} \label{Eqn:9.3}
h(q_m) +[\Lambda_m\otimes 1, S] = q_m + [S, k\partial+q_m] + \frac{1}{2} [S, [S, L_m]] + \frac{1}{6} [S, [S, [S, L_m]]] + \cdots .
\end{equation}
We project the both sides of equation (\ref{Eqn:9.3}) onto $\hat{\g}_{-(d+1)m-\frac{1}{2}} \otimes \mathcal{V}_{[d,m]}$ and get the formula:
$$ h_{-(d+1)m-\frac{1}{2}}(q_m) + [\Lambda_m, S_{\frac{1}{2}}]= 	q_{-(d+1)m-\frac{1}{2}}.$$
By (\ref{Eqn:9.1}), we can find  $h_{-(d+1)m-\frac{1}{2}}(q_m)$ and $S_{\frac{1}{2}}$ uniquely.  Similarly, equating the $(\hat{\g}_{-(d+1)m}\otimes \mathcal{V}_{[d,m]})$-part of (\ref{Eqn:9.3}), we obtain 
$h_{-(d+1)m}(q_m)$ and $S_{1}$. Inductively, $h_i$ and $S_j$ for any $i > -(d+1)m-1$ and $j>0$ are determined uniquely.
\end{proof}

Since $ \Lambda_m z^k \otimes 1$ is in $ (\ker(\ad \Lambda_m) \cap \hat{\g}_{-(d+1)m-1}) \otimes \CC$ , the space $ (\ker(\ad \Lambda_m) \cap \hat{\g}) \otimes \CC $ is nontrivial. Hence we can choose a nonzero element
\begin{equation}\label{Eqn:b}
 b \in (\ker(\ad \Lambda_m) \cap \hat{\g}) \otimes \CC. 
\end{equation}

Let $S\in \n \otimes  \mathcal{V}_{[d,m]}$ be from Proposition \ref{Prop:9.1} when $L_m$ is substituted by the universal Lax operator $\mathcal{L}_m$. Also, let $(e^{-\ad S(x)}b)^{>0}$ and $(e^{-\ad S(x)}b)^{\leq 0}$ be the  projections of $e^{-\ad S(x)}b$ onto $\hat{\g}^{>0}\otimes \mathcal{V}_{[d,m]}$ and $\hat{\g}^{\leq 0}\otimes \mathcal{V}_{[d,m]}$. Then the following two evolution equations
\begin{align}
\frac{\partial \phi(y)}{\partial t} &  = \int-\left((e^{-\ad S(x)}b)^{>0}\delta(x-w), \left[\frac{\delta \phi(y)}{\delta u} \delta(y-w), \mathcal{L}_m(w)(q_m)\right]\right)_wdx , \label{Eqn:9.4}\\
\frac{\partial \phi(y)}{\partial t} & = \int \left((e^{-\ad S(x)}b)^{\leq 0}\delta(x-w) ,\left[\frac{\delta \phi(y)}{\delta u} \delta(y-w) , \mathcal{L}_m(w)(q_m)\right] \right)_w dx, \label{Eqn:9.5}
\end{align}
where $\phi \in \WW(\g, \Lambda_m,k)$ and $\frac{\delta\phi}{\delta u}=\sum_{(i,j) \in \mathcal{I}} u_i^j \otimes   \frac{\delta \phi}{\delta u_i^j}$, are useful to find an integrable system associated to the algebra $\WW(\g, \Lambda_m, k)$.

\begin{prop} \label{Prop:9.2}
Two equations (\ref{Eqn:9.4}) and (\ref{Eqn:9.5}) are the same evolution equation.
\end{prop}

\begin{proof}
We notice that the bilinear form $(\cdot , \cdot )$ on $\hat{g}\otimes \mathcal{V}_{[d,m]}$ is  $\ad (\hat{g}\otimes \mathcal{V}_{[d,m]})$-invariant, i.e. $(a_1\otimes f_1, [a_2\otimes f_2, a_3\otimes f_3])= ([a_1\otimes f_1,a_2\otimes f_2], a_3\otimes f_3).$ Also, we have $\int (a\otimes f, [b\otimes g, \partial_x])dx= - \int (a, b) f \partial_x g \ dx = \int (b,a)  g\partial_x f \ dx=- \int (b\otimes g, [ a \otimes f, \partial_x])dx.$ By subtracting equation (\ref{Eqn:9.4}) from (\ref{Eqn:9.5}), we obtain
\begin{equation*}
 (\ref{Eqn:9.5})-(\ref{Eqn:9.4}) =- \int \left(\frac{\delta \phi(y)}{\delta u} \delta(y-w), \left[  e^{-\ad S(x)}b \delta(x-w) ,  \mathcal{L}_m(w) \right] \right)_w dx. 
\end{equation*}
Moreover, since $e^{\ad S(x)}\left[ e^{-S(x)} b \delta(x-w), \mathcal{L}_m (w)  \right] = \left[  b \delta(x-w), e^{\ad S(x)}\mathcal{L}_m (w)  \right] $ and the invariance of the bilinear form holds, we have 
\begin{equation*}
 (\ref{Eqn:9.5})-(\ref{Eqn:9.4})=- \int \left(e^{-\ad S(y)}\frac{\delta \phi(y)}{\delta u} \delta(y-w), \left[ b \delta(x-w)  ,  e^{\ad S(w)} \mathcal{L}_m(w) \right]\right)_w dx =0.
\end{equation*}
\end{proof}

Given $b$ in (\ref{Eqn:b}), let  $H_b: \hat{\g}^{\leq 0}_{>-(d+1)m-1} \otimes \mathcal{V}_{[d,m]}\to \mathcal{V}_{[d,m]}$ be a functional defined by 
\begin{equation} \label{Eqn:9.6_0228}
 H_b(q_m(x)):= (b , h(q_m(x))).
\end{equation}
Then $H_b$ has the following property.

\begin{prop} \label{Thm:9.3}
Let $e^{-\ad S} b^{\geq0}_{<(d+1)m+1}$ be the projection of $e^{-\ad S} b$ onto $\hat{\g}^{\geq0}_{<(d+1)m+1}.$ Then 
$$ \frac{\delta}{\delta u} H_b = e^{-\ad S}b^{\geq0}_{<(d+1)m+1}.$$

\end{prop}

\begin{proof}
For $ r \in  \hat{\g}^{\leq 0}_{>-(d+1)m-1} \otimes \mathcal{V}_{[d,m]}$, 
let $L_m(\epsilon) = k\partial + (q_m(x) + \epsilon r) + \Lambda_m \otimes 1$. By Proposition \ref{Prop:9.1}, there exist $h(q_m(x)+\epsilon r)\in \ker \ad\Lambda_m  \otimes \mathcal{V}_{[d,m]}$ and $ S\in \text{im} (\ad\Lambda_m) \otimes \mathcal{V}_{[d,m]}$ satisfying the equation:
\begin{equation*}
L_{m,0}(\epsilon) = k \partial + h(q_m(x)+ \epsilon r) + \Lambda_m\otimes 1 = e^{\ad S(\epsilon)} L_m(\epsilon).
\end{equation*}
Then
\begin{equation}\label{Eqn:9.6_0228}
\frac{d}{d \epsilon} L_{m,0}(\epsilon) = e^{\ad S(\epsilon)} r  + \left[ \frac{\partial S(\epsilon)}{\partial \epsilon}, e^{\ad S(\epsilon)}L_m(\epsilon)\right]= e^{\ad S(\epsilon)} r  + \left[ \frac{\partial S(\epsilon)}{\partial \epsilon}, L_{m,0}(\epsilon)\right].
\end{equation}
Using Taylor expansion and formula (\ref{Eqn:9.6_0228}), we have
\begin{equation}\label{Eqn:9.7_0228}
\begin{aligned}
\left(  \frac{\delta}{\delta u} H_b(q_m(x)) , r  \right)& = \frac{d}{d \epsilon} (b  , h(q_m(x)+\epsilon r))|_{\epsilon=0} = \frac{d}{d \epsilon} (b , L_{m,0}(y)(\epsilon))|_{\epsilon=0} \\
& = (e^{-\ad S(x)} b, r) +  \left(b, -\partial_x \left.\frac{d S(x)(\epsilon)}{d\epsilon}\right|_{\epsilon=0} \right).
\end{aligned}
\end{equation}
Applying $\int \ dx$ to (\ref{Eqn:9.7_0228}), we obtain
\begin{align}\label{Eqn:9.8_0228}
 \int \left(  \frac{\delta}{\delta u} H_b(q_m(x))  , r  \right) dx = \int  (e^{-\ad S(x)} b, r) dx.
\end{align}
Since (\ref{Eqn:9.8_0228}) holds for any  $r\in  \hat{\g}^{\leq 0}_{>-(d+1)m-1} \otimes \mathcal{V}_{[d,m]}$, we conclude that  $\frac{\delta H_b}{\delta u} =  e^{-\ad S} b^{\geq0}_{<(d+1)m+1}.$
\end{proof}

By Proposition \ref{Thm:9.3}, we have
$$ \int \left( \frac{\delta H_{z^{-1} b}(q_m(x))}{\delta  u} \delta(x-w), z r(y) \right)_w dx 
= \int ( (e^{-\ad S(x)}  b)^{> 0} \delta(x-w), r(y))_w dx.$$
Hence (\ref{Eqn:9.4}) is equivalent to the equation 
\begin{equation}\label{Eqn:9.6}
\begin{aligned}
\frac{\partial \phi(y)}{\partial t} = \int \{ H_{z^{-1}b}(x), \phi(y) \}_1 dx.
\end{aligned}
\end{equation}

To write the second equation as a Hamiltonian equation, we denote $\sum_{(i,0)\in \mathcal{I}} {u_i^0} \otimes \frac{\delta H_b}{\delta u_i^0}$ by $\frac{\delta H_b}{\delta u^0}$ and $\sum_{(i,j)\in\mathcal{I}, \ j >0}  {u_i^j} \otimes \frac{\delta H_b}{\delta u_i^j}$ by $\frac{\delta H_b}{\delta u^>}.$ Then
\begin{equation} \label{Eqn:9.7}
\begin{aligned}
\frac{\partial \phi(y)}{\partial t} & = \int \left(e^{-\ad S(x)}b^{\leq 0}\delta(x-w), \left[\frac{\delta \phi(y)}{\delta u} \delta(y-w) , \mathcal{L}_m(w)(q_m)\right] \right)_w dx \\
& = \int \left( (e^{-\ad S(x)}b^0-e^{-\ad S(x)}b^{\geq 0})\delta(x-w), \left[\frac{\delta \phi(y)}{\delta u} \delta(y-w) , \mathcal{L}_m(w)(q_m)\right] \right)_w dx \\
& = \int \left( \frac{\delta H_b(x)}{\delta u^0} \delta(x-w), \left[\frac{\delta \phi(y)}{\delta u}\delta(y-w), \mathcal{L}_m(w)(q_m)\right] \right)_w dx \\
 & -  \int \left(z^{-1} \frac{\delta H_b(x)}{\delta u^{>}}\delta(x-w),z \left[\frac{\delta \phi(y)}{\delta u} \delta(y-w) , \mathcal{L}_m(w)(q_m)\right] \right)_w dx \\
& = \int \{ H_b(x), \phi(y) \}_2 dx.
\end{aligned}
\end{equation}

By equations (\ref{Eqn:9.6}) and (\ref{Eqn:9.7}), we get the following theorem.

\begin{thm}\label{Thm:9.5}
Let $b \in (\ker  (\ad\Lambda_m) \cap \hat{\g}_{<0}) \otimes 1$ and $H_b:= ( b, h(q_m(x))$. Then $$ \frac{\partial \phi}{\partial t} =  \int \{H_{z^{-1}b}(x), \phi(y) \}_1 dx = \int \{H_b(x), \phi(y) \}_2 dx, \ \ \phi\in \WW(\g, \Lambda_m, k).$$

\end{thm}

In terms of Poisson vertex algebra theories, Theorem \ref{Thm:9.5} can be stated as follows:

\begin{thm} \label{Thm:9.6}
Let $b$ and $H_b$ be defined as in Theorem \ref{Thm:9.5} and (\ref{Eqn:9.6_0228}). Then
$$ \frac{\partial \phi}{\partial t} =  \{H_{z^{-1}b \ \lambda} \phi \}_1(y)|_{\lambda=0}=  \{H_{b \ \lambda} \phi \}_2(y)|_{\lambda=0}, \ \ \ \phi\in \WW(\g, \Lambda_m, k).$$
\end{thm}

\begin{proof}
It suffices to show that this theorem is equivalent to Theorem \ref{Thm:9.5}.
By the results of previous sections, we have
\begin{equation}\label{Eqn:9.11_0302}
\begin{aligned}
&\int \{H_{z^{-1}b}(x), \phi(y) \}_1 dx = \sum_{(i,j)\in \mathcal{I},(p,q)\in\mathcal{I},n,l\geq0} \int \frac{\partial H_{z^{-1}b}(x)}{\partial u_i^{j(n)}} \frac{\partial  \phi(y)}{\partial u_p^{q(l)}} \partial_x^n \partial_y^l \{u_i^j,  u_p^{q+1}\}(y) \delta(x-y) dx \\
 & =\sum_{(i,j),(p,q)\in \mathcal{I}, l\geq 0, 0 \leq l_1\leq l} \left( \partial_y^{l-l_1} \frac{\delta H_{z^{-1}b}(y)}{\delta u_i^{j}}\right) \frac{\partial  \phi(y)}{\partial u_p^{q(l)}} {l \choose l_1} \partial_y^{l_1} ( \{u_i^j,  u_p^{q+1}\}(y)) \\
 & = \sum_{(i,j),(p,q)\in \mathcal{I}, l\geq 0}\frac{\partial  \phi(y)}{\partial u_p^{q(l)}} \partial_y^{l} \{u_i^j,  u_p^{q+1}\}(y) \frac{\delta H_{z^{-1}b}(y)}{\delta u_i^{j}}\ \ =\ \ \{H_{{z^{-1}b} \ \lambda} \phi \}_1(y) |_{\lambda=0}.
\end{aligned}
\end{equation}
Hence
$\int \{ H_{b} (x), \phi(y) \}_1 dx = \{H_{{z^{-1}b} \ \lambda} \phi \}_1(y) |_{\lambda=0}.$
The same procedure works for the second bracket. 
\end{proof}

The following theorem is the main goal of this section.

\begin{thm}\label{Thm:9.7}
Suppose $\frac{\delta H_b}{\delta u^0}$
is not a constant.
Then the evolution equation
$$ \frac{\partial \phi}{\partial t} = \{ H_{z^{-1}b \ \lambda} \phi\}_1|_{\lambda=0}$$
is an integrable system. In fact, $H_{z^{-i} b}$, $i \geq 0$, are integrals of motion and they are linearly  independent.
\end{thm}

\begin{proof}
In order to show that each $H_{z^{-i}b}$ is an integral of motion, we need to prove that $  \int \{H_{z^{-1}b}, H_{z^{-i}b} \}_1(y) dy  =0. $
By Proposition \ref{Thm:9.3}, we have
\begin{align*}
&\int \{H_{z^{-1}b}, H_{z^{-i}b} \}_1(y) dy \\
=- &\iint ( e^{-\ad S(x)} z^{-1}b^{\geq 0}_{<(d+1)m+1} \delta(x-w), [e^{-\ad S(y)} z^{-i+1}b^{\geq 0}_{<(d+1)m+1} \delta(y-w), \mathcal{L}_m(w) ])_w dxdy \\
=-&  \iint ( z^{-1}b^{\geq 0}_{<(d+1)m+1} \delta(x-w), [ z^{-i+1}b^{\geq 0}_{<(d+1)m+1} \delta(y-w), \mathcal{L}_{m,0}(w) ])_w dxdy \\
=-& \iint ( [z^{-1}b^{\geq 0}_{<(d+1)m+1} \delta(x-w),  z^{-i+1}b^{\geq 0}_{<(d+1)m+1} \delta(y-w)], \mathcal{L}_{m,0}(w) )_w dxdy =0.
\end{align*}
Furthermore, Theorem \ref{Thm:9.6} implies the independency of the integrals of motion. Indeed, if we denote 
$$ \{a_\partial b\}_{\to}:= \sum_n c_n \partial^n \text{ where } \{a_\lambda b\}=\sum_n c_n \ \lambda^n \in \mathcal{V}_{[d,m]}[\lambda],$$
 Theorem \ref{Thm:9.6} can be written as
$$\{ u^j_{i \ {\partial}} u_p^q\}_{1 \ \to} \frac{\delta H_{z^{-(i+1)}b}}{\delta u_i^j} = \{ u^j_{i \ {\partial}} u_p^q\}_{2 \ \to} \frac{\delta H_{z^{-i}b}}{\delta u_i^j },$$
for any $i \geq 0$, $u_i^j, u_p^q \in \mathcal{B}$. Since $\frac{\delta H_b}{\delta u^0}$ is not a constant, there is $u_p^0$ such that $(u_i^0, u_p^0)\neq 0$. Then $\{ u^0_{i \ {\partial}} u_p^0\}_{2 \ \to} $ increases the total derivative order of $\frac{\delta H_{b}}{\delta u_i^0}$ by 1. However, $\{ u^0_{i \ {\partial}} u_p^0\}_{1 \ \to} $ preserve the total derivative order of $\frac{\delta H_{z^{-1}b}}{\delta u_i^0}.$ Hence $H_{b}$ and $H_{z^{-1}b}$ are linearly independent. Inductively, each $H_{z^{-i}b}$, for any $i\geq 0$, has a different total derivative order. So  $H_{z^{-i} b}$ are linearly independent.
\end{proof}

\begin{ex}
Let $\g=\sll_2$ and $m=1$. Then the associated universal Lax operator is 
$$ L_1 = k\partial + q_1 + \Lambda_1,$$
where
$$
q_1=  \left[ \begin{array}{cc}
x  & f  \\
e & -x 
\end{array} \right] 
+ 
\left[ \begin{array}{cc}
xz & fz \\
0 & -xz
\end{array} \right]  z^{-1}
\text{ \  and \  } 
\Lambda_1= 
-\left[ \begin{array}{cc}
0  &  z^{-2} \\
z^{-1} & 0 
\end{array} \right] .
$$
There exists a unique linear map $\ \gamma: \g \oplus \g_f z \to \WW(\sll_2, \Lambda_1,k)$ such that 
$$
q_1^{can}:=  \left[ \begin{array}{cc}
\gamma_x  & \gamma_f\\
\gamma_e & -\gamma_x 
\end{array} \right] 
+ 
\left[ \begin{array}{cc}
0  & \gamma_{fz} \\
0 & 0 .
\end{array} \right]  z^{-1}
$$
is gauge equivalent to $q_1.$  Then $\gamma_x, \gamma_f, \gamma_e, \gamma_{fz}$ freely generate the differential algebra $\WW(\g, \Lambda_1, k)$ (see (\ref{Eqn:8.10})). Let us find $S\in \text{im} (\ad \Lambda_1)$ and $h(q_1^{can}) \in \ker(\Lambda_1) \otimes \mathcal{V}_1$ such that 
\begin{equation}\label{Eqn:9.8_0301}
k\partial +h(q_1^{can}) + \Lambda_1 \otimes 1 = e^{\ad S} L_1^{can}, 
\end{equation}
where  $L_1^{can} = k\partial+ q_1^{can} + \Lambda_1$.
Suppose $h(q_1^{can}) = \sum_{i\geq -1} h_i$, $S= \sum_{i>0} S_i$ and $h_i$, $S_i \in \hat{\g}_i \otimes \mathcal{V}_1.$
Then the  $(\hat{\g}_{-1} \otimes \mathcal{V}_1)$-part of (\ref{Eqn:9.8_0301}) is:
$$ h_{-1} + [\Lambda_1, S_2] = \left( q_{1}^{can} \right)_{-1}= ez^{-1} \otimes \gamma_{fz} + f \otimes \gamma_e. $$
Hence we get
\begin{equation}\label{Eqn:9.8}
\begin{aligned}
 h_{-1}   = \Lambda_{1} z \otimes -\frac{1}{2}  \left( \gamma_{fz} + \gamma_e  \right), \ \  S_2  = hz \otimes \frac{1}{4} \left( \gamma_{fz}-\gamma_{e} \right). 
\end{aligned}
\end{equation}
Similarly, by equating $(\hat{\g}_{0} \otimes \mathcal{V}_1)$, $(\hat{\g}_{1} \otimes \mathcal{V}_1)$, $(\hat{\g}_{2} \otimes \mathcal{V}_1)$,  $(\hat{\g}_{3} \otimes \mathcal{V}_1)$-parts of (\ref{Eqn:9.8_0301}), we have
\begin{equation}\label{Eqn:9.9}
\begin{aligned}
h_0 & = 0, \ \ \ S_3 = Kz \otimes -\frac{1}{2}\gamma_x,\\
h_1 & = \Lambda_1 z^2 \otimes \left( -\frac{1}{2} \gamma_{f} - \frac{1}{4} (\gamma_{fz}-\gamma_{e})^2 \right) , \ \ \ 
S_4  = hz^2 \otimes \left( \frac{1}{4} \gamma_f +  \frac{1}{16} \gamma_{fz}^2 -\frac{3}{16} \gamma_e^2 +\frac{1}{8} \gamma_{fz}\gamma_{e} \right), \\
h_2 & = 0, \\
h_3 & = \Lambda_1  z^3 \otimes \left( \frac{1}{2} \gamma_x^2 +( \gamma_e-\gamma_{fz}) \left( \frac{1}{16} \gamma_{fz}^2-\frac{1}{16} \gamma_e^2 + \frac{1}{4} \gamma_f \right)\right),
\end{aligned}
\end{equation}
where $ K= (-e+fz).$\\
Let $b z^{-n}:= -\frac{1}{2}(e+f z) z^{-n} \otimes 1 $ and let 
$ H_n(q_1):=H_{b z^{-n}}(q_1) = ( b z^{-n} , h(q_1)).$
Then we obtain
$$
H_0  =  -\frac{1}{2} \gamma_{f} - \frac{1}{4} (\gamma_{fz}-\gamma_{e})^2, \ \ \ H_1 =  \frac{1}{2} \gamma_x^2 +( \gamma_e-\gamma_{fz}) \left( \frac{1}{16} \gamma_{fz}^2-\frac{1}{16} \gamma_e^2 + \frac{1}{4} \gamma_f \right).
$$
Using formulas (\ref{Eqn:8.11})-(\ref{Eqn:8.14}), we obtain the following Poisson $\lambda$-brackets:

\begin{equation}\label{Eqn:9.11}
\begin{aligned}
\{H_{0 \lambda} \gamma_f \}_1 & = k \lambda, \\ 
\{H_{0 \lambda} \gamma_x \}_1 & = \{H_{0 \lambda} \gamma_e \}_1= \{H_{0 \lambda} \gamma_{fz} \}_1=0 ,
\end{aligned}
\end{equation}

\begin{equation}\label{Eqn:9.12}
\begin{aligned}
&\{H_{0 \lambda} \gamma_f \}_2 =\{H_{1 \lambda} \gamma_f \}_1= \left(\gamma_ x+\frac{1}{2} k (\lambda+\partial) \right)(\gamma_e-\gamma_{fz}),\\
&\{H_{0 \lambda} \gamma_x \}_2  =\{H_{1 \lambda} \gamma_x \}_1= -\frac{1}{2}\gamma_e(\gamma_{fz}-\gamma_e) -\frac{1}{2} \gamma_f, \\
&\{H_{0 \lambda} \gamma_e \}_2  =\gamma_x-\frac{1}{2} k\lambda, \ \ \ \{H_{1 \lambda} \gamma_e \}_1=\gamma_x, \\
&\{H_{0 \lambda} \gamma_{fz} \}_2  = -\gamma_x +\frac{1}{2} k \lambda, \ \ \{H_{1 \lambda} \gamma_{fz} \}_1= -\gamma_x,
\end{aligned}
\end{equation}

\begin{equation}\label{Eqn:9.13}
\begin{aligned}
&\{H_{1 \lambda} \gamma_f \}_2 = \left( \frac{1}{2} \gamma_x+ \frac{1}{4} k (\lambda + \partial) \right) (\gamma_{fz}^2 -\gamma_e^2) + \frac{1}{2} k (\lambda+\partial) \gamma_f ,\\
&\{H_{1 \lambda} \gamma_x \}_2 = \left( \frac{1}{4} \gamma_e (\gamma_e-\gamma_{fz}) - \frac{1}{4} f\right) (\gamma_e+\gamma_{fz}) ,\\
&\{H_{1 \lambda} \gamma_e \}_2 = \gamma_e \gamma_x + \left( -\frac{1}{2} \gamma_x + \frac{1}{4}k(\lambda+\partial) \right)(\gamma_e -\gamma_{fz}),\\
&\{H_{1 \lambda} \gamma_{fz} \}_2 =-\gamma_e \gamma_x - \left( -\frac{1}{2} \gamma_x + \frac{1}{4}k(\lambda+\partial) \right)(\gamma_e -\gamma_{fz}).
\end{aligned}
\end{equation}
As a consequence, two equations
\begin{equation}\label{Eqn:9.14}
\begin{aligned}
 \frac{d\gamma}{dt} &= \{ H_{0 \lambda} \gamma\}_2 |_{\lambda=0} =  \{ H_{1 \lambda} \gamma\}_1 |_{\lambda=0},\\
  \frac{d\gamma}{dt} &= \{ H_{1 \lambda} \gamma\}_2 |_{\lambda=0},
\end{aligned}
\end{equation}
are compatible integrable systems. By (\ref{Eqn:9.12}),  the first equation in (\ref{Eqn:9.14}) is as follows:

\begin{equation}\label{Eqn:9.15}
\left\{
\begin{aligned}
\frac{d\gamma_{f}}{dt} & = (\gamma_x + \frac{1}{2} k \partial) (\gamma_{fz}-\gamma_e) \\
\frac{d\gamma_{x}}{dt} & = - \frac{1}{2} \gamma_e (\gamma_{fz} - \gamma_e) -\frac{1}{2} \gamma_f \\
\frac{d\gamma_{e}}{dt} & = -\frac{d\gamma_{fz}}{dt} = \gamma_x.
\end{aligned}
\right.
\end{equation}

Since $\gamma_e+\gamma_{fz}$ is in the center of $\WW(\g, f,k)$, we take quotient by the center of $\WW(\g, f,k)$ and obtain the following equation:
\begin{equation}
\left\{
\begin{aligned}
\frac{d\gamma_{f}}{dt} & = -2(\gamma_x + \frac{1}{2} k \partial) (\gamma_e) \\
\frac{d\gamma_{x}}{dt} & =  \gamma_e^2 -\frac{1}{2} \gamma_f \\
\frac{d\gamma_{e}}{dt} & = \gamma_x.
\end{aligned}
\right.
\end{equation}

Eliminating $\gamma_{f} $and $\gamma_x$, we get the equation
\begin{equation}\label{Eqn:9.16}
(\gamma_{e})_{ttt} = 3 \gamma_e (\gamma_e)_t +\frac{1}{2} k (\gamma_e)_x.
\end{equation}
This equation is the KdV equation with $x$ and $t$ exchanged. 

\end{ex}


\begin{thebibliography}{00}


\bibitem {BDHM} N.J. Burroughs, M.F. De Groot, T.J. Hollowood,
J.L.Miramontes, Generalized Drinfel'd-Sokolov hierarchies. Commun.Math.Phys. 153,
(1993).


\bibitem{BDK} A. Barakat, A. De Sole, V.G. Kac,Poisson vertex algebras in the theory of Hamiltonian equations. Jpn. J. Math. 4 (2009), no. 2, 141Ð252


\bibitem{DHM} M.F. De Groot, T.J. Hollowood, J.L.Miramontes, Generalized
Drinfel'd-Sokolov hierarchies. Commun.Math.Phys. 145, (1992).

\bibitem{DK} A. De Sole, V. G. Kac, Finite vs affine
$W$-algebras, Jpn. J. Math. 1 (2006), no. 1, 137-261

\bibitem{DKV} A. De Sole, V. G. Kac, D. Valeri, Classical W-algebras and generalized Drinfeld-Sokolov bi-Hamiltonian systems within the theory of Poisson vertex algebras,  arXiv:1207.6286

\bibitem{FL} V.A. Fateev, S.L. Lukyanov, The models of two dimensional conformal quantum field theory with $\ZZ_n$ symmetry, Int. J. Mod. Phys. A 3 (1988), 507-520.

\bibitem{FT} L.D. Faddeev, L.A. Takhtajan, Hamiltonian approach in soliton theory, Nauka, 1986

\bibitem{GG} W. L. Gan, V. Ginzburg, Quantization of Slodowy
slices, Int. Math. Res. Not. 2002, no. 5, 243-255

\bibitem{KR}  Victor G. Kac, Alexei Rudakov, Representations of the exceptional Lie superalgebra $E(3,6)$. III. Classification of singular vectors. J. Algebra Appl. 4 (2005), no. 1, 15-57

\bibitem{KW} Victor G. Kac, Minoru Wakimoto, Quantum reduction
and representation theory of superconformal algebras. Adv. Math. 185
(2004),  no. 2, 400--458.

\bibitem {P1} A. Premet, Special transverse slices and their
enveloping algebras, Adv. Math. 170 (2002), no. 1, 1-55

\bibitem{P2} A. Premet, Enveloping algebras of Slodowy slices and
the Joseph ideals, J. Eur. Math. Soc. (JEMS) 9 (2007), no. 3,
487-543 

\bibitem{DS} V.G. Drinfel'd, V.V. Sokolov, Lie algebras and equations of
Korteseg-de Vries Type. J.Sov.Math.Dokl. 23, 1975 (1985)

\bibitem{K} V.G. Kac, Infinite dimensional Lie algebras. 2nd ed. Cambridge
University Press 1985

\bibitem{KP} V.G. Kac, D.H. Peterson, 112 Constructions of the basic
represention of the loop group of $E_8$, Symposium on Anomalies, geometry and
topology,

\bibitem{S} Suh, Ph.D. Thesis, Structure of classical $\WW$-algebras (2013)

\bibitem{We} Charles A. Weivel, An introduction to homological
algebra. Cambridge Studies in Advanced Mathematics, 38. Cambridge
Universi  ty Press, Cambridge, 1994.


\bibitem{Wi} G.Wilson, The modifies Lax and two-dimensional Toda Lattice equations
associated with simple Lie algebras. Ergod.Th. and Dynam.Sys. 1, 361 (1981)

\end{thebibliography}
\end{document}